\documentclass[
12pt, 
oneside, 
english, 
doublespacing, 
headsepline, 
]{MastersDoctoralThesis} 

\usepackage{chngcntr}
\usepackage[utf8]{inputenc} 
\usepackage[T1]{fontenc} 
\usepackage{times} 
\usepackage{amstext,amsmath,amssymb,amsthm}
\usepackage{cite}
\usepackage{color}
\usepackage{subfigure}

\newcommand{\ket} [1] {\vert #1 \rangle}
\newcommand{\bra} [1] {\langle #1 \vert}

\newcommand{\Ord}[1]{\mathcal{O}\left(#1\right)}

\usepackage{caption}
\usepackage{physics}
\usepackage{graphicx}
\usepackage{microtype}
\usepackage{braket}
\usepackage{tabularx,colortbl}
\usepackage{pifont}
\usepackage{mathtools}
\usepackage{dsfont}
\usepackage{algpascal}
\usepackage{algorithm}
\usepackage{algpseudocode}
\usepackage{chngcntr}
\usepackage{complexity}
\usepackage[affil-it]{authblk}
\usepackage[table, x11names]{xcolor}
\usepackage{booktabs}
\usepackage{multirow}
\usepackage{tabulary}
\usepackage{array,makecell,multirow,ragged2e}
\usepackage[font=small,labelfont=bf]{caption}
\usepackage{csquotes}
\usepackage{braket}
\usepackage{hyperref}

\usepackage{ifpdf}
\ifpdf
\else
\PackageWarningNoLine{Qcircuit}{Qcircuit is loading in Postscript mode.  The Xy-pic options ps and dvips will be loaded.  If you wish to use other Postscript drivers for Xy-pic, you must modify the code in Qcircuit.tex}
\xyoption{ps}
\xyoption{dvips}
\fi

\thesistitle{Quantum Statistical Inference} 
\supervisor{Prof. Joseph \textsc{Fitzsimons}} 
\examiner{} 
\degree{Doctor of Philosophy} 
\author{Zhikuan \textsc{Zhao}} 
\addresses{} 
\university	{Singapore University Of Technology and Design}
\pillar{Engineering Product Development} 

\hypersetup{pdftitle=\ttitle} 
\hypersetup{pdfauthor=\authorname} 
\hypersetup{pdfkeywords=\keywordnames} 

\begin{document}

\frontmatter 

\pagestyle{plain} 


\begin{titlepage}
\begin{center}

\begin{figure}
\centering
\includegraphics[width=0.5\textwidth]{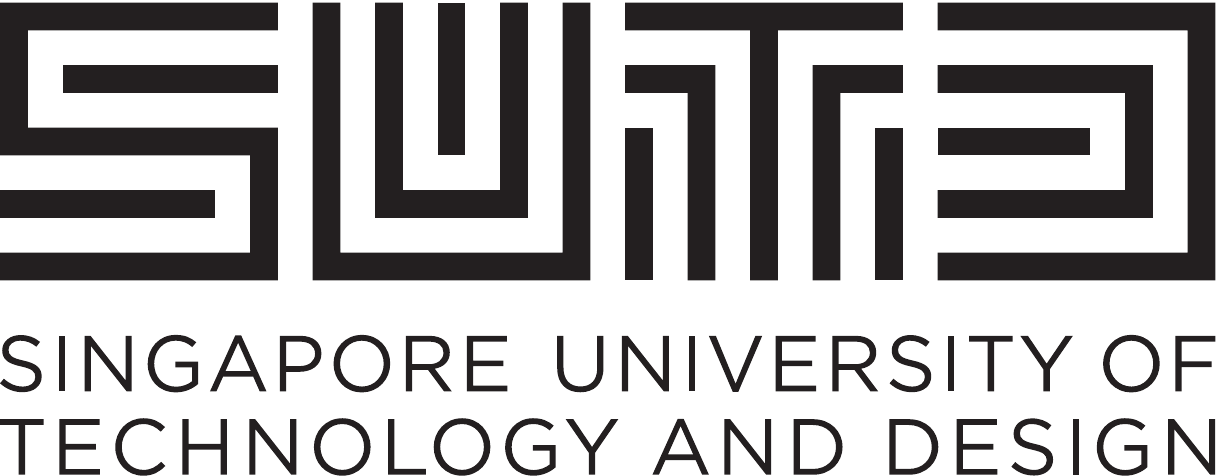}
\end{figure}

\hfill\break\\[1.0cm]
 
{\huge \bfseries \ttitle}\\[2cm] 

Submitted by\\[1cm]
Zhikuan \textsc{Zhao} 

\vspace{3em}

Thesis Advisor\\[1cm]
\supname 
 
\vspace{3em} 

 
\large{A thesis submitted to the Singapore University of Technology and Design in fulfilment of the requirement for the degree of \degreename.}\\[0.5cm] 

{\large \today}\\[4cm] 

\vfill
\end{center}
\end{titlepage}


\newpage

\noindent{\huge \textbf{Thesis Examination Committee}}
\\

\noindent TEC Chair: Ricky Ang\\
\noindent Thesis Advisor:  Joseph Fitzsimons\\
\noindent Internal TEC Member:  Shaowei Lin\\
\noindent Internal TEC Member:  Dario Poletti\\
\noindent External TEC Member:  Troy Lee\footnote{University of Technology Sydney}\footnote{Centre for Quantum Technologies, National University of Singapore}

\begin{declaration}
\addchaptertocentry{\authorshipname}

\noindent I, Zhikuan \textsc{Zhao}, declare that this thesis titled, \enquote{\ttitle} and the work presented in it are my own. I confirm that:

\begin{itemize} 
\item This work was done wholly or mainly while in candidature for a research degree at this University.
\item Where any part of this thesis has previously been submitted for a degree or any other qualification at this University or any other institution, this has been clearly stated.
\item Where I have consulted the published work of others, this is always clearly attributed.
\item Where I have quoted from the work of others, the source is always given. With the exception of such quotations, this thesis is entirely my own work.
\item I have acknowledged all main sources of help.
\item Where the thesis is based on work done by myself jointly with others, I have made clear exactly what was done by others and what I have contributed myself.\\
\end{itemize}
 
\noindent Signed:\\
\rule[0.5em]{25em}{0.5pt} 
 
\noindent Date:\\
\rule[0.5em]{25em}{0.5pt} 
\end{declaration}

\cleardoublepage


\vspace*{0.2\textheight}

\noindent\enquote{\itshape Product of optimism and knowledge is a constant.}\bigbreak

\hfill Lev Landau

\begin{List of Publications}
\addchaptertocentry{\listofpublications}	
\begin{itemize}
\item \textbf{Quantum Linear System Algorithm for Dense Matrices}\\
L. Wossnig, Z. Zhao, \& A. Prakash. Phys. Rev. Lett. 120, 050502 (2018).  (Contains work used in Chapter 3)

\item \textbf{A note on state preparation for quantum machine learning}\\
Z. Zhao, V. Dunjko, J. K. Fitzsimons, P. Rebentrost, \& J. F. Fitzsimons.  arXiv preprint arXiv:1804.00281 (2018). (Contains work used in Chapter 5)

\item \textbf{Quantum assisted Gaussian process regression}\\
Z. Zhao, J. K. Fitzsimons, \& J. F. Fitzsimons, arXiv preprint arXiv:1512.03929 (2015). (Contains work used in Chapter 5)

\item \textbf{Quantum algorithms for training Gaussian Processes}\\
Z. Zhao, J. K. Fitzsimons, M. A. Osborne, S. J. Roberts, \& J. F. Fitzsimons. arXiv preprint arXiv:1803.10520 (2018). (Contains work used in Chapter 6)

\item \textbf{Bayesian Deep Learning on a Quantum Computer}\\
Z. Zhao, A. Pozas-Kerstjens, P. Rebentrost, \& P. Wittek. arXiv preprint arXiv: 1806.11463 (2018). (Contains work used in Chapter 7)

\item \textbf{Geometry of quantum correlations in space-time}\\
Z. Zhao, R. Pisarczyk, J. Thompson, M. Gu, V. Vedral, \& J. F. Fitzsimons, Phys. Rev. A 98, 052312 (2018). (Contains work used in Chapter 8)

\item \textbf{Causal limit on quantum communication}\\
R. Pisarczyk, Z. Zhao, Y. Ouyang, V. Vedral, \&  J. F. Fitzsimons. arXiv preprint arXiv:1804.02594 (2018). (Contains work used in Chapter 9)

\end{itemize}

\end{List of Publications}

\begin{Statement on Collaborative Work}
\addchaptertocentry{\SOCW}

\noindent The results presented in this thesis came from several fruitful research collaborations.

The quantum linear system algorithm for dense matrices presented in Chapter \ref{QDLSA} was developed in collaboration with Leonard Wossnig and Anupam Prakash. I initiated the project, led and jointly contributed to the analytical work, and took the role as the corresponding author of the paper.

The state preparation technique for quantum machine learning used in Chapter \ref{QGP} was developed together with Vedran Dunjko, Jack Fitzsimons, Patrick Rebentrost and my supervisor, Joseph Fitzsimons who came up with the initial idea. I contributed to working out the technical details of the research.

The quantum assisted Gaussian process algorithm in Chapter \ref{QGP} came from a collaboration with Jack Fitzsimons, and Joseph Fitzsimons, with whom a discussion inspired the initial idea of the project. I contributed to a large part of the detailed algorithm design and most of the analysis involved. 

The quantum algorithms for training Gaussian processes discussed in Chapter \ref{QGPT} came from the collaborative work with Jack Fitzsimons, Michael Osborne, Stephen Roberts, who collectively provided expertise on classical machine learning, and Joseph Fitzsimons who initiated the research. I contributed to the algorithm design and most of the analytical work.

The quantum Bayesian deep learning algorithm described in Chapter \ref{Chapter: QBDL} was developed in collaboration with Alejandro Pozas-Kerstjens, Patrick Rebentrost, and Peter Wittek, with whom I jointly initiated the project and contributed to the programming of numerical simulations. I did most of the theoretical work, and proved the main theorem with the help of Patrick Rebentrost.  Alejandro Pozas-Kerstjens completed the experimental part of the research.

The work presented in Chapter \ref{Chapter: geometry} on the geometry of quantum correlations in space-time was done together with Robert Pisarczyk, Jayne Thompson, Mile Gu, Vlatko Vedral and Joseph Fitzsimons. The project originated from discussions with Jayne Thompson, Mile Gu, Vlatko Vedral and Joseph Fitzsimons. I proved the main results, and completed the theoretical details jointly with Robert Pisarczyk.

The work presented in Chapter \ref{CauCap} on bounding channel capacities with quantum causality was done in collaboration with Robert Pisarczyk,  Yingkai Ouyang, Vlatko Vedral and Joseph Fitzsimons. Vlatko Vedral and Joseph Fitzsimons initiated the research. I contributed to proving the theoretical results jointly with Robert Pisarczyk and Yingkai Ouyang.

\end{Statement on Collaborative Work}


\begin{abstract}
\addchaptertocentry{\abstractname} 
In this thesis, I present several results on quantum statistical inference  in the following two directions.
Firstly, I demonstrate that quantum algorithms can be applied to enhance the computing and training of Gaussian processes (GPs), a powerful model widely used in classical statistical inference and supervised machine learning. A crucial component of the quantum GP algorithm is solving linear systems with quantum computers, for which I present a novel algorithm that achieves a provable advantage over previously known methods. I will also explicitly address the task of encoding the classical data into a quantum state for machine learning applications. I then apply the quantum enhanced GPs to Bayesian deep learning and present an experimental demonstration on contemporary hardware and simulators. 
Secondly, I look into the notion of quantum causality and apply it to inferring spatial and temporal quantum correlations, and present an analytical toolkit for causal inference in quantum data.
I will also make the connection between causality and quantum communications, and present a general bound for the quantum capacity of noisy communication channels. 
\end{abstract}


\begin{acknowledgements}
\addchaptertocentry{\acknowledgementname} 
First and foremost, I would like to express my most sincere gratitude and appreciation to my supervisor, Joseph Fitzsimons for providing continuous support, patient guidance and perhaps most vitally, role model through his most rigorous attitude toward science, which has kept me going throughout the past four years of intellectual journey.
Thank you, Joe. Without your tutorship and mentorship, none of these would have been possible. 

I would like to thank to Ricky Ang, Shaowei Lin, Dario Poletti and Troy Lee for kindly agreeing to serve as the examination committee for this thesis. 
I am most grateful to my collaborators: Vedran Dunjko, Jack Fitzsimons,  Mile Gu, Michael Osborne, Robert Pisarczyk, Alejandro Pozas-Kerstjens, Anupam Prakash, Patrick Rebentrost, Stephen Roberts, Jayne Thompson, Vlatko Vedral, Peter Wittek, Leonard Wossnig and Ouyang Yingkai. It has been a great pleasure working together. I am also thankful to my peers, Joshua Kettlewell, Atul Mantri and Liming Zhao for the most memorable experience of growing up together, and my seniors in the group, especially Tiago Batalh\~ao, Tommaso Demarie, Michal Hajdu\v{s}ek, Nana Liu and Si-Hui Tan for the care, support and all the fun we had during the past years. 

Last but not least, I could never have made it through without the love and support of my families who have always been the heroes by my side during times of struggle. My most profound gratitude goes well beyond the scope of this thesis.

\end{acknowledgements}


\tableofcontents 

\dedicatory{Dedicated to my beloved families} 


\mainmatter 

\pagestyle{thesis} 

\part{Quantum computation and algorithms}
\label{Part1}

\chapter{Introduction} 
\label{Intro} 

Quantum mechanics is the theoretical framework that underpins our understanding of the physical world at the most fundamental level. Since its discovery in the early 20th century, quantum mechanics has proved to be tremendously successful in predicting physical phenomena at the microscopic scale, providing unprecedented insights ranging from the fundamental particles in nature to the origin of cosmos. Throughout history, our society has held a track record of coupling scientific discoveries with the invention of technologies that reshape everyday life. Quantum mechanics is no exception. Perhaps most pronouncedly, the understanding of the quantum nature of electronic structures in matter played the vital role in giving birth to the entire semiconductor industry, which is in turn responsible for the dawn of the information era, an era in which computation has taken centre stage and revolutionised the world. 
Broadly speaking, the conventional digital computer is called "classical" since it processes information in the form of logical bits, which omits the possibility of superposition and entanglement allowed by quantum mechanics. As such, despite its almost universal success, when classical computer is used for the task of simulating complex quantum mechanical systems, significant difficulties arise due to the memory requirement for keeping track of the exponentially large state space of the system. 
\newpage
Motivated initially by the problem of simulating physics, Feynman proposed to design and build computers that directly leverage the exponential state space in quantum mechanics \cite{Feynman}. Since this original vision, progress in finding algorithms for future quantum computers has come a long way, and well beyond the domain of quantum simulation alone. Among the most celebrated results are Grover's search algorithm \cite{grover1996fast} which shows a quadratic advantageous over its classical counter-part and Shor's factoring algorithm \cite{shor1999polynomial} which has the potential to break the (to our best knowledge) classically secure RSA cryptosystem. More recently, machine learning has rapidly emerged as an area where quantum algorithms can display dramatic advantages \cite{aimeur2006machine,pudenz2013quantum, superunsuper, QVSM, dunjko2017machine, perdomo2017opportunities}. 

In this thesis, I will present several new results in the more general context of quantum statistical inference, a term used here with two-fold meanings. Firstly, we demonstrate the power of applying quantum computation to statistical models for supervised machine learning with classical datasets. Secondly, we address the notion of causality in quantum information and present an analytical toolkit for inferring causal correlations when the data is itself inherently quantum. 
In Part \ref{Part1} of the thesis, I will start by introducing the basic concepts of quantum mechanics and quantum computation, then move on to review several essential quantum algorithms in Chapter \ref{IntroQA}. In Chapter \ref{QDLSA}, I will present a new algorithm for the quantum version of the linear system problem, which shows an advantage over the existing approaches, particularly when the matrix involved is inherently dense. 
In Part \ref{Part2}, we will see that quantum algorithms can be applied to improve the efficiency of supervised learning with Gaussian processes, with a novel application to deep learning.
In Part \ref{Part3}, we look into quantum causality. I will present results on the geometry of spatial and temporal quantum correlations and the operational role of causality in quantum communication.

\section{Quantum mechanics preliminaries}
Here we start by reviewing the fundamental postulates of quantum mechanics and introduce the notation and concepts elementary to the presentation of this thesis. These postulates underline the mathematical framework of quantum physics. Hence they hold a foundational role to future discussions about quantum computation and quantum statistical inference. We will keep our presentation at a basic level. An in-depth discussion of the postulates and a detailed introduction to quantum mechanics is presented in the canonical text of Ref. \cite{nielsen2010quantum}. 
\subsection{The state space}
\paragraph{Postulate 1} Any isolated physical system is associated with a complex vector space
with inner product, which is known as the state space (also known as the Hilbert space) of the system. The system is fully described by a unit vector in its state space, which is known as its state vector.
\paragraph{Dirac notation and superposition}
The state vectors in quantum mechanics are commonly denoted by a ``ket'', e.g., $\ket{\psi}$. Their Hermitian transpose is denoted by a ``bra'', so that $\ket{\psi}^{\dagger}=\bra{\psi}$. The inner product between two state vectors, $\ket{\psi}$ and $\ket{\phi}$ is denoted as the ``braket'', $\langle \psi|\phi \rangle$. It follows directly from Postulate 1 that any valid quantum state vector, $\ket{\psi}$, satisfies $\langle \psi|\psi \rangle = 1$. A quantum state $\ket{\psi}$ is in a superposition of the states $\{\ket{\phi_i}\}$ if it can be written as a set of mutually orthogonal states, $\ket{\psi}=\sum_i\alpha_i\ket{\phi_i}$, where $\sum_i |\alpha_i|^2=1$

\subsection{Evolution of states}

\paragraph{Postulate 2} The evolution of closed quantum systems is linear, and described by unitary transformations. The state, $\ket{\psi(t_2)}$ of a quantum system at time $t_2$ is related to the state, $\ket{\psi(t_1)}$, at an earlier time $t_1$ via a unitary transformation $U$ that only depends on $t_1$ and $t_2$, so that $\ket{\psi(t_2)}= U\ket{\psi(t_1)}$.
\paragraph{Schr\"{o}dinger equation}
The time-dependent Schr\"{o}dinger equation describes the time evolution of a closed quantum system,
\begin{align}
i\hbar \frac{d}{dt} \ket{\psi(t)} = H \ket{\psi(t)},
\end{align}
where the Hermitian operator $H$ is known as the Hamiltonian. The factor $\hbar$ is the Planck's constant. We work in units such that $\hbar=1$.

\subsection{Quantum measurements}
\paragraph{Postulate 3}
Quantum measurements are described by a set of measurement operators, $\{M_m\}$, where $\sum_m M_m^{\dagger}M_m = I$. If the system is in the quantum state $\ket{\psi}$ immediately before the measurement, then the probability of the measurement result $m$ occurring is given by 
$
p(m)=\bra{\psi}M_m^{\dagger}M_m\ket{\psi},    
$
and the post-measurement state of the system after is given by 
$
\frac{M_m\ket{\psi}}{\sqrt{\bra{\psi}M_m^{\dagger}M_m\ket{\psi}}}. 
$
\paragraph{Projective measurements}
An important special case of the quantum measurements is the projective measurement. In a projective measurement, the measurement operators are taken to be $M_m=\ket{\phi_m}\bra{\phi_m}$, where the set of state vectors $\{\ket{\phi_m}\}$ form an orthonormal basis for the system's Hilbert space. The corresponding probability of an outcome $m$ occurring is then given by $p(m)=|\langle \psi | \phi_m \rangle|^2$. Every projective measurement is associated with an observable, $M=\sum_m\ket{\phi_m}\bra{\phi_m}$. The expectation value of the observable given by $\langle M \rangle=\bra{\psi}M\ket{\psi}$.

\subsection{Composite systems}
\paragraph{Postulate 4}
The Hilbert space of a composite quantum system is given by the tensor product of the Hilbert spaces of the individual components. 
For a set of $n$ component systems initialised in the states $\{\ket{\psi_i}\}_{i=1}^n$, the state of the composite system is given by 
$\bigotimes\limits_{i=1}^{n} \ket{\psi_i} = \ket{\psi_1}\otimes\ket{\psi_2}\otimes...\otimes\ket{\psi_n}$. 
\paragraph{Entanglement}
If a composite system has the state as a tensor product of the states of its subsystems, we say the composite system is in a product state. 
Note that, however, the superposition of product states will not, in general, be in a product state. If the state of a system cannot be written as the tensor product of the states of its subsystems, we say it is entangled. For instance, if $\langle \psi_1 | \psi_2 \rangle=0$, the state $\frac{1}{\sqrt{2}}(\ket{\psi_1}\ket{\psi_2}+\ket{\psi_2}\ket{\psi_1})$ is maximally entangled. 

\section{Elements of quantum computation}
Having reviewed the fundamentals of quantum physics, we now move on to introduce the elementary concepts used in quantum computation. These include the basic unit of quantum computation, the qubit, the important observables given by the Pauli operators, and the unitary gates used to process quantum information.
\subsection{The qubit}
The qubit is the most basic non-trivial quantum system. It is also the smallest unit of quantum computation. A single qubit in a ``pure'' quantum state is a two dimensional complex vector, and can be written as $\ket{\psi}=\alpha\ket{0}+\beta\ket{1}$, where $\ket{0}=(1,0)^T$ and $\ket{1}=(0,1)^T$ are, and $|\alpha|^2+|\beta|^2=1$. Note that since the probability of a measurement outcome, by postulate 3 is invariant under $\ket{\psi}\rightarrow \mathrm{e}^{i\phi}\ket{\psi}$, a global phase factor $\mathrm{e}^{i\phi}$ is not an observable in quantum mechanics, and we can parameterise the single qubit state as $\ket{\psi}=\cos(\theta)\ket{0}+\mathrm{e}^{i\phi}\sin(\theta)\ket{1}$. As such the qubit can be visualised as a point lying on the surface of a unit sphere, known as the Bloch sphere. The vectors $\{\ket{0},\ket{1}\}$ forms the $Z$ basis (computational basis) of the single qubit state space. Alternatively, the basis can be chosen as any pair of orthogonal states, e.g. the $X$ basis, $\left\{ \ket{+_x}=\frac{\ket{0}+\ket{1}}{\sqrt{2}}, \ket{-_x}=\frac{\ket{0}-\ket{1})}{\sqrt{2}} \right\}$ and the $Y$ basis, $\left\{ \ket{+_y}=\frac{\ket{0}+i\ket{1}}{\sqrt{2}}, \ket{-_y}=\frac{\ket{0}-i\ket{1}}{\sqrt{2}} \right\}$. 

\subsection{Pauli operators}
The Pauli operators are observables corresponding to the projectors in the $X$, $Y$ and $Z$ bases. They are given by $\sigma_1=X=\ket{+_x}\bra{+_x}-\ket{-_x}\bra{-_x}$, $\sigma_2=Y=\ket{+_y}\bra{+_y}-\ket{-_y}\bra{-_y}$ and $\sigma_3=Z=\ket{0}\bra{0}-\ket{1}\bra{1}$. We will also use $\sigma_0=I$ to denote the $2\times 2$ identity operator. Note that the Pauli operators are traceless, Hermitian and unitary, i.e. $\Tr[\sigma_i]=0$, $\sigma_i=\sigma_i^\dagger$ and $\sigma_i^2=\sigma_0$ for $i= 1,2,3$. 

\subsection{Quantum gates}
An important part of quantum computation amounts to composing unitary operations acting on collections of qubits. These operation are known as quantum gates. Single-qubit gates correspond to unitary operators acting locally on one qubit, e.g. the Hadamard gate, $H\ket{j}=\frac{\ket{0}+(-1)^j\ket{1}}{\sqrt{2}}$, $j\in\{0,1\}$. In general, the Pauli operators can be used to construct arbitrary single qubit unitary rotations, $R_{\sigma_i}(\theta)=\mathrm{e}^{i\theta\sigma_i}$ around each respective axis.   
Many-qubit gates are unitary operations acting on more than one qubit. These operations are capable of generating quantum entanglement, e.g. the controlled-not gate, $CNOT\ket{i}\ket{j}=\ket{i}\ket{i\oplus j}$, where $i,j\in \{0,1\}$ and $\oplus$ denotes the addition modulo 2.


\section{Statistical ensemble of states}
\subsection{The density matrix}
\label{sub: DM}
The density matrix is a formalism to describe a probability mixture of pure quantum states. 
Suppose we are given a system which has a probability $p_i$ to be in the state $\ket{\psi_i}$, we say the system is in a statistical ensemble of pure states, $\{p_i,\ket{\psi_i}\}$. The density matrix (or density operator) of the system is then defined as
\begin{align}
    \rho = \sum_i p_i\ket{\psi_i}\bra{\psi_i},
\end{align}
where $\sum_i p_i=1$. The density matrix is an operator acting on the system's Hilbert space. In the special case when the state of the system is in $\ket{\psi_j}$ with unit probability, we say the system is in a pure state, and the density matrix is simply given by the projector, $\ket{\psi_j}\bra{\psi_j}$. Otherwise, we say the system is in a mixed state with a probability distribution $\{p_i\}$. The density matrix can be used to calculate the expectation value of any observable $M$ on the system as follows,
\begin{align}
    \langle M\rangle=\sum_i p_i\bra{\psi_i}M\ket{\psi_i}=\Tr [\rho M].
\end{align}
Since the eigenvalues of the density matrix physically correspond to a probability distribution over the eigenvectors of $\rho$ which are themselves pure quantum state vectors, the density matrix is necessarily positive semi-definite Hermitian operators with unit trace. On the other hand, any given $2^N\times 2^N$ matrix that satisfies the Hermitian, positive semi-definite and unit trace properties have the physical interpretation of an $N$-qubit density matrix. In Part \ref{Part3} of this thesis, we will consider a natural extension of the density matrix formalism where the multi-qubit observables on the mixed state are allowed to extend across the temporal domain.

\subsection{Quantum operations}
In the case of a closed system, the evolution of the density matrix translates straight-forwardly from the unitary and linear dynamics for pure states, i.e., if a unitary $U$ is applied on the ensemble $\{p_i,\ket{\psi_i}\}$, the corresponding density matrix transforms as $\rho\rightarrow U\rho U^{\dagger}$. In this section, we describe the general quantum operation on open quantum systems. 

Suppose now an initial system described by $\rho$ is coupled with an environment described (without loss of generality) by the pure state $\rho_{e}=\ket{e_0}\bra{e_0}$. Since the joint system, $\rho\otimes \rho_{e}$ is now a closed system, its general dynamics can be described by the unitary transformation, $U(\rho \otimes \ket{e_0}\bra{e_0})U^\dagger$. The resultant transformation on the initial system, $\varepsilon (\rho)$ is then given by a partial trace over the environment,
\begin{align}
    \varepsilon (\rho) =& \Tr_e \left[U(\rho\otimes\ket{e_0}\bra{e_0})U^\dagger\right]\nonumber\\
    =& \sum_k \bra{e_k}U(\rho\otimes\ket{e_0}\bra{e_0})U^\dagger\ket{e_k}\nonumber\\
    =& \sum_k E_k \rho E_k^\dagger,
\end{align}
where $\ket{e_k}$ denotes an orthonormal basis for the environment's state space. We have defined $E_k = \bra{e_k}U\ket{e_0}$ which are known as  the Kraus operators of the quantum operation $\varepsilon$. Trace preserving quantum operations are also known as quantum channels. A channel mathematically corresponds to a completely positive trace preserving (CPTP) map. In this case, the Kraus operators satisfy the completeness relation, $\sum_k E_k^\dagger E_k=I$. In general, when measurements are involved and extra information is obtained about the process, the quantum operation is not necessarily trace preserving, and the Kraus operators instead satisfy $\sum_k E_k^\dagger E_k\le I$.
The trace preserving cases (quantum channels) will be more relevant to the materials presented in Part \ref{Part3} of this thesis.

 \chapter{Essential quantum algorithms}

In this chapter, I introduce some essential quantum algorithms which will serve as building blocks later in the thesis.
We start with the more basic algorithms: The quantum Fourier transform which is regarded as the root of quantum advantage in many higher-level algorithms, quantum phase estimation which approximately computes the eigenvalues of a Hamiltonian matrix in a superposition, and quantum Hamiltonian simulation which amounts to constructing a unitary operator corresponding to the time evolution under a Hamiltonian. We then review a quantum algorithm that combines these basic techniques and provides an advantage in solving systems of linear equations under a quantum formulation of the problem.

\label{IntroQA} 

\section{Basic quantum algorithms}

\subsection{Quantum Fourier transform}
\label{subsec:QFT}
The quantum Fourier transform (QFT) is the foundation of many quantum algorithms, including the celebrated quantum factoring algorithm \cite{shor2004classical}. It can be seen as the quantum analog of the discrete Fourier transform in classical computation. Here we briefly introduce QFT and describe the unitary operator for its implementation. A detailed description can be found in all canonical texts of quantum information, such as Ref.\cite{kitaev2002classical, nielsen2010quantum}.

The normalised discrete Fourier transform of a vector $\mathbf{v}=(v_1...v_n)^T$ is given by the vector $\hat{\mathbf{v}}$ with entries,
$
\hat{v}_y=\frac{1}{\sqrt{n}}\sum_{x=1}^{n}v_x \mathrm{e}^{-\frac{2\pi xyi}{n}}    .
$
For $v_x$ with periodicity $P$, such that $v_x=v_{x+P}$, we have
\begin{align}
\hat{v}_y 
=& \frac{1}{\sqrt{n}}\left(\sum_{x=1}^P\sum_{m=0}^{\lfloor nP^{-1} \rfloor-1} v_{x+mP}\mathrm{e}^{-\frac{2\pi(x+mP)yi}{n}}+\sum_{x=\lfloor nP^{-1} \rfloor+1}^n v_x \mathrm{e}^{-\frac{2\pi xyi}{n}} \right)\nonumber\\
=&  \frac{1}{\sqrt{n}} \left(\sum_{m=0}^{\lfloor nP^{-1} \rfloor-1} \mathrm{e}^{-\frac{2\pi mPyi}{n}} \sum_{x=1}^P v_{x+mP}\mathrm{e}^{-\frac{2\pi xyi}{n}} +\sum_{x=\lfloor nP^{-1} \rfloor+1}^n v_x \mathrm{e}^{-\frac{2\pi xyi}{n}} \right).
\end{align}
Note that for $n\gg P$, the above expression only consists of small oscillations around zero unless $Py$ is an integer. Therefore the only surviving terms correspond to $y$ being an integer multiple of the frequency. 

The QFT is the discrete Fourier transform applied to quantum state vectors, and it is implemented by the unitary operator,
\begin{align}
    U_{QFT}=\frac{1}{\sqrt{n}}\sum_{y=1}^n\sum_{x=1}^n \mathrm{e}^{-\frac{2\pi xyi}{n}} \ket{y}\bra{x}.
\end{align}
One can easily verify the above indeed corresponds to the discrete Fourier transform of a quantum state by applying it to an arbitrary state vector $\ket{\mathbf{v}}$,
\begin{align}
    \bra{z}U_{QFT}\ket{\mathbf{v}}=&\bra{z}\frac{1}{\sqrt{n}}\sum_{y=1}^n\sum_{x=1}^n \mathrm{e}^{-\frac{2\pi xyi}{n}} \ket{y}\langle x|\mathbf{v}\rangle\nonumber\\
    =& \frac{1}{\sqrt{n}}\sum_{x=1}^{n} \langle x|\mathbf{v}\rangle \mathrm{e}^{-\frac{2\pi xyi}{n}}\nonumber\\
    =& \langle z|\hat{\mathbf{v}}\rangle.
\end{align}
Given access to a set of basic unitary gates and $n$ qubits, a quantum computer can perform the discrete Fourier transform on $2^n$ amplitudes with only $\mathcal{O}(n^2)$ Hadamard and controlled phase gates, providing an exponential advantage over the classical counterpart that takes $\mathcal{O}(n2^n)$ gates \cite{nielsen2010quantum}. It is worth noting that an improved version of the QFT presented in Ref. \cite{hales2000improved} has further suppressed the cost to $\mathcal{O}(n\log n)$.

\subsection{Quantum phase estimation}
\label{sub:PE}

The Quantum phase estimation, first introduced in Ref. \cite{Kitaev1995} is a quantum algorithm that takes as input an eigenvector of a unitary
operator and estimates the corresponding eigenvalue to a certain additive error. It is the root of the quantum advantage in many machine learning and linear algebraic applications. Here we define the quantum phase estimation algorithm for future reference. A detailed description of its procedures can be found for example in section 5.2 of Ref. \cite{nielsen2010quantum}. 

Let the unitary operator $U\in \mathbb{C}^{n\times n}$ have eigenvectors $\{\ket{v_j}\}$ with corresponding eigenvalues $\{e^{i \theta_j}\}$, such that $\ket{v_j} = e^{i \theta_j} \ket{v_j}$, where $\theta_j \in [ - \pi ,\pi ]$    for $j \in [n]$. Further define the precision parameter $\delta$ to denote an additive error. Given an oracle for implementing $U^l$ for $l=\mathcal{O}(1/\delta)$, the quantum phase estimation algorithm performs the following transformation, 
\begin{align}
\sum_{j \in [n]} \alpha_j \ket{v_j} \to \sum_{j \in [n]} \alpha_j \ket{v_j} \ket{\overline{\theta}_j},     
\end{align}
such that $|\overline{\theta_{j}} - \theta_{j}|\leq \delta$ for all $j\in [n]$ with probability $1-1/\text{poly}(n)$ in time that scales as $\Ord{T_U \log{(n)} / \delta}$, where $T_U$ denotes the time required to implement $U$.

\subsection{Black-box Hamiltonian simulation}
\label{sub: BHS}

Given a Hermitian Hamiltonian operator $H\in \mathbb{C}^{n\times n} $,
the black-box access to the matrix elements $H_{jk}$ is an oracle $O_{H}$ that allows for the operation,
\begin{align}
    O_{H} \ket{j,k} \ket{z} \to  \ket{j,k} \ket{z \oplus H_{jk}},
\end{align}
for an arbitrary input $\ket{z}$, where $j,k \in \{1,2,...,n\}$ and $\oplus$ denotes the bitwise addition modulo two operation. The time evolution of a quantum state $\ket{\psi(t)}$ under $H$ is described by the time-dependent Schr\"{o}dinger equation,
\begin{equation}
i \frac{d}{dt} \ket{\psi(t)} = H \ket{\psi(t)}. 
\end{equation}
The solution is given by $\ket{\psi(t)}=U(H,t)\ket{\psi(0)}$, where the unitary operator $U(H,t)=\exp\left(-iHt\right)$. Black-box Hamiltonian simulation amounts to constructing a quantum circuit that implements $U(H,t)$ given access to the oracle $O_{H}$.

In the general case, the results of \cite{berry2012black} shows that the black-box Hamiltonian simulation can be performed in time $\Ord{n^{2/3}\cdot\text{polylog}(n)/\delta_h^{1/3}}$ with an $\delta_h$ error in the trace distance using a method based on discrete time quantum walks \cite{Szegedy04}. Empirical results of \cite{berry2012black} suggested black-box Hamiltonian simulation can be implemented in time $\Ord{\sqrt{n}\cdot\text{polylog}(n)/ \delta_h^{1/2}}$ for several classes of Hamiltonians. However, the $\tilde{\mathcal{O}}(\sqrt{n})$ runtime is known to not hold in the worst case. The notation $\tilde{\mathcal{O}}(.)$ is used here to suppress slower growing factors in the runtime scaling.
In special cases, properties of $H$ such as sparsity can be leveraged to implement Hamiltonian simulation more efficiently. 
It was shown in Ref. \cite{berry2015hamiltonian} that combing techniques from quantum walk \cite{Szegedy04} and fractional query simulation \cite{berry2017exponential}, Hamiltonian simulation on an $s$-sparse matrix (that is, the maximum number of non-zero entries on any rows or columns is $s$) can be performed in time $\tilde{\mathcal{O}}(s\cdot\text{polylog}(n)/\delta_h^{1/2})$.

It is worth mentioning that the black-box model is not the uniquely interesting setting to consider. Other important models include the quantum signal processor \cite{low2016hamiltonian} and the density matrix encoding mode \cite{kimmel2017hamiltonian, lloyd2014quantum}. Detailed descriptions of quantum Hamiltonian simulation algorithms and a comprehensive review on this subject is beyond the scope of this thesis. Interested readers are referred to the Chapters 25 and 26 of Ref. \cite{childs2017lecture}.

\section{Quantum linear system algorithm}

Solving a linear system of equations is a problem that appears in many disciplines across science and engineering. Given a set of $n$ linear equations with $n$ unknown variables, we wish to find the $n$ dimensional vector $\mathbf{x}$ which satisfies $A\mathbf{x}=\mathbf{b}$, where $A$ and $\mathbf{b}$ a are known $n\times n$ dimensional matrix and a known $n$ dimensional vector respectively. The solution of the linear system can be written as $\mathbf{x}=A^{-1}\mathbf{b}$ for an invertible matrix $A$. In special cases, $A$ has convenient properties such as sparsity, of which one can take advantage and compute $A^{-1}$ in time proportional to $n$ with the conjugate gradient method\cite{shewchuk1994introduction}. In general, the best known classical method for matrix inversion scales as $\Ord{n^{2.373}}$, with the optimised CW-like algorithms \cite{Coppersmith1990, L14}. However, this sub-cubic scaling is practically difficult to achieve. A more typical implementation amounts to using the Cholesky decomposition which has a runtime that scales as $\Ord{n^{3}}$ for dense matrices. In modern statistical inference and machine learning applications, matrix inversion presents a computational bottleneck when the dimensionality $n$ of the underlying problem grows. 
Recent discoveries in quantum algorithms have shown promises for a more efficient solution of high-dimensional linear systems. Given the importance and generality of the problem, quantum linear system algorithms may manifest as the cornerstone of quantum advantage in many use cases. In this section, we review some of the earlier progress in this subject. In the next chapter, we will present a new result along the same line of research.

\subsection{Quantum formulation of linear systems}
Let $A\in \mathbb{R}^{n\times n}$ be a Hermitian matrix, with $\|A\|_*\le 1$. Here $\| . \|_*$ denotes the spectral norm which corresponds to the largest absolute value of the eigenvalues in the case of Hermitian matrices. Let $\mathbf{x}, \mathbf{b}\in \mathbb{R}^n$, such that $A\mathbf{x}=\mathbf{b}$. We define the following quantum formulation of the linear system problem: 

Given access to the elements of $A$ and an input quantum state vector $\ket{\mathbf{b}}$ of $\log n$ qubits which encodes the entries in $\mathbf{b}$ as
\begin{align}
\ket{\mathbf{b}}=\frac{\sum_j b_j\ket{j}}{\|\sum_jb_j\ket{j}\|_2},    
\end{align}
the quantum linear system problem amounts to finding the state vector $\ket{\mathbf{x}}$ of $\log n$ qubits which encodes the entries in solution vector $\mathbf{x}$ as 
\begin{align}
\ket{\mathbf{x}}=\frac{\sum_j x_j\ket{j}}{\|\sum_jx_j\ket{j}\|_2}.
\end{align}
\textbf{Remarks:}
\begin{itemize}

\item Note that the input and output of the quantum linear system problem are both quantum states. Therefore the initial state preparation and final solution readout procedures will need to be explicitly addressed for any applications that have classical vectors as inputs and outputs. This point has been discussed in Ref. \cite{Aaronson2015b} and will be revisited later in this thesis.

\item Defining $A$ to be a Hermitian matrix is in fact without loss of generality. As pointed out in Ref. \cite{Harrow2009a}, a general matrix $M$ can be embedded into a Hermitian matrix with a constant memory overhead by constructing
a block-wise anti-diagonal matrix $A$ as follows,
\begin{align}
A =\begin{pmatrix}
    0 & M^\dagger\\
    M & 0  
   \end{pmatrix}.
\end{align}  

\item The requirement on bounded spectral norm is not a strong restriction in practice since it can often be satisfied with a suitable choice of normalisation factor. 
\end{itemize}
\subsection{The HHL algorithm}
\label{HHL}
In the breakthrough work of Ref. \cite{Harrow2009a}, Harrow, Hassidim and Lloyd (HHL) introduced the first quantum linear system algorithm (QLSA) that computes the quantum state $\ket{\mathbf{x}}=\ket{A^{-1}\mathbf{b}}$ which corresponds to the solution of the linear system $A\mathbf{x}=\mathbf{b}$ in time $\Ord{\text{polylog}(n)}$ for a sparse and well-conditioned $A$. In this section, we review this seminal algorithm and discuss its implications. The procedure of the original quantum linear systems solver provided in Ref. \cite{Harrow2009a} can be summarised in the following five steps:

\begin{enumerate}
\item To start with, prepare a quantum state $\ket{\mathbf{b}}$ which encodes the vector $\mathbf{b}\in \mathbb{R}^n$ as 
    $
        \ket{\mathbf{b}}=(\mathbf{b}^T\mathbf{b})^{-1/2}\sum\limits_{i=0}^{n-1}b_i\ket{i}.
    $
    Then append to $\ket{\mathbf{b}}$ an ancillary register in a superposition state $\frac{1}{\sqrt{T}}\sum_{\tau = 0}^T \ket{\tau}$. The time period $T$ is chosen to be some large value as required in the variant of phase-estimation described in Ref. \cite{clock}, so that after Step 1 we have the quantum state, 
    \begin{align}
        \ket{\phi_1}= \frac{1}{\sqrt{\mathbf{b}^T\mathbf{b}}}\frac{1}{\sqrt{T}}\sum\limits_{i=0}^{n-1}\sum_{\tau = 0}^T b_i\ket{i}\ket{\tau}. \label{eqn: step1HHL}
    \end{align}
\item Perform Hamiltonian simulation treating the matrix $A$ as the Hamiltonian at time $\tau$. Apply the resultant controlled unitary operation to $\ket{\mathbf{b}}$ using techniques described in Ref. \cite{Berry2007}. By writing $\ket{\mathbf{b}}$ in the eigenbasis of ${A}$ after evolution, we obtain the state,
\begin{align}
\ket{\phi_2}=\frac{1}{\sqrt{\mathbf{b}^T\mathbf{b}}}\frac{1}{\sqrt{T}}\sum\limits_{i=0}^{n-1}\sum\limits_{\tau=0}^{T-1}\ket{\tau}e^{i\lambda_it_0\tau/T}\beta_i\ket{\mu_i},
\end{align} 

where $\lambda_i$ are the eigenvalues and $\ket{\mu_i}$ are the eigenvectors of $A$. The complex numbers $\beta_i$ are the probability amplitudes associated with $\ket{\mu_i}$. For some precision parameter $\epsilon$ which will feature as an additive error of the final result in the trace norm, we choose the time scale $t_0=\mathcal{O}(\kappa/\epsilon)$ where $\kappa$ denotes the condition number, the ratio between the largest and the smallest eigenvalues of $A$.
\item Complete phase estimation \cite{Kitaev1995,clock} by applying the quantum Fourier transform (QFT) to the first register in $\ket{\phi_2}$, which leads to 
\begin{align}
\ket{\phi_3}=\frac{1}{\sqrt{\mathbf{b}^T\mathbf{b}}}\sum\limits_{i=0}^{n-1}\beta_i\ket{t\bar{\lambda}_i}\ket{\mu_i},
\end{align}
where the first register now stores the estimated eigenvalues $\bar{\lambda}_i$ up to a constant multiplicative factor $t$.
\item Introduce another ancillary qubit and perform a controlled rotation on it based on the value in the first register, and obtain the extended state
\begin{align}
\ket{\phi_4}=\frac{1}{\sqrt{\mathbf{b}^T\mathbf{b}}}\sum\limits_{i=0}^{n-1}\beta_i\ket{t\bar{\lambda}_i}\ket{\mu_i}\left(\sqrt{1-\frac{c_\lambda^2}{\bar{\lambda_i}^2}}\ket{0}+\frac{c_\lambda}{\bar{\lambda_i}}\ket{1}\right).
\end{align}
Here the constant $c_\lambda$ is chosen such that the resultant probability amplitude is bounded by unity.
\item Reverse the phase estimation step on the first register to uncompute $\ket{t\bar{\lambda}_i}$. Measure the final ancillary qubit. Conditioned on obtaining $\ket{1}$ as the measurement result, an approximated solution of ${A}\ket{\mathbf{x}}=\ket{\mathbf{b}}$ is obtained, 
\begin{align}
\ket{\bar{\mathbf{x}}}=\ket{\phi_{5}}=\frac{1}{\sqrt{\mathbf{b}^T\mathbf{b}}}\sum\limits_{i=0}^{n-1}\frac{\beta_i}{\bar{\lambda_i}}\ket{\mu_i}.
\end{align}
For a precision parameter $\epsilon$, the additive error in the trace norm of the output state is bounded as $\|\ket{\bar{\mathbf{x}}}-\ket{\mathbf{x}}\|\le \epsilon$. Note that a post-selection of measurement outcomes is involved in this final step, and as a consequence multiple repetitions of the procedure may be needed in order to successfully obtain the desired outcome.
\end{enumerate}

\paragraph{Runtime and errors}
The required Hamiltonian simulation subroutine runs nearly linearly with the sparsity, $s$, with the black-box Hamiltonian simulation technique of \cite{berry2015hamiltonian}. The time scale parameter $t_0$ of phase estimation is chosen to be $\mathcal{O}(\kappa/\epsilon)$ to ensure the desired precision. Furthermore, $\mathcal{O}(\kappa)$ repetitions of the procedure are needed to obtain the desired outcome on the final measurement of the ancillary qubit, making use of the amplitude amplification based techniques of \cite{Ambainis2010}. From the above rough account, the total runtime scales as $\tilde{O}(\log(n)\kappa^2{s}^2/\epsilon)$. A detailed error and runtime analysis can be found in the supplementary material of \cite{Harrow2009a}.

\paragraph{Potential caveats}
The quantum linear algorithm described above can potentially provide a promising exponential speed-up. However, one needs to apply it with care. As Aaronson accurately described in Ref. \cite{Aaronson2015b}, there are four potential caveats that need particular care in any applications: (1) The time consumption of preparing $\ket{\mathbf{b}}$ encoding $\mathbf{b}$ needs to be taken into account; (2) the matrix $A$ has to be robustly invertible, meaning that the condition number $\kappa$ needs to grow at most polylogarithmically in $n$ in order to retain a polylogarithmic overall runtime; (3) one also needs to address the sparsity contribution to the total runtime, since the general phase estimation sub-routine costs time polynomial in $s$; (4) although the output of QLSA is the state $\ket{\mathbf{x}}$, there is no efficient procedure to extract every entry of $\mathbf{x}$. The quantum advantage only presents when the matter of practical interest does not require the full $\mathbf{x}$ but requires only information accessible with a few copies of $\ket{\mathbf{x}}$. For instance, if a known Hermitian matrix ${M}$ is of interest, one can efficiently estimate quantities such as $\bra{\mathbf{x}}{M}\ket{\mathbf{x}}$, since this amounts to the expectation value of the observable $M$ on $\ket{\mathbf{x}}$. 

\paragraph{Developments}
There have been several improvements to the QLSA since the original HHL proposal that have improved the running time to linear in the condition number $\kappa$ and the sparsity $s$, and to poly-logarithmic in the precision parameter $\epsilon$ \cite{Ambainis2010, Childs2015}. The work of Ref.\cite{Clader2013a} further introduced pre-conditioning for the QLSA and extended its applicability. In the next chapter, we build upon this line of research and present a linear system algorithm that circumvents the expensive Hamiltonian simulation step and has a provably better performance than the existing algorithms when applied to linear systems with dense matrices.

\chapter{Quantum dense linear system algorithm}
\label{QDLSA}

In this chapter, I present an alternative approach to solving the quantum linear systems problem, which is based on a quantum subroutine for singular value estimation (SVE). The SVE-based linear system algorithm, introduced in Ref. \cite{wossnig2018quantum} has a runtime scaling of $\Ord{\kappa^2 \|A\|_F \cdot\text{polylog}(n)/\epsilon}$ for an $n\times n$ dimensional Hermitian matrix $A$ with a Frobenius norm $\|A\|_F$ and condition number $\kappa$. As before, $\epsilon$ is the precision parameter defined by the desired output error in the trace norm. Unlike the HHL algorithm, the SVE-based method does not require performing Hamiltonian simulation on $A$, making it advantageous particularly when $A$ is dense. Therefore, we refer to it as a quantum dense linear system (QDLS) algorithm. An important component of the QDLS algorithm is the quantum singular value estimation (QSVE), introduced in\cite{Kerenidis2016}. It makes use of a memory model that supports efficient preparation of states which correspond to the row vectors of $A$ and the vector of the row Euclidean norms of $A$. We will start by introducing this memory model, followed by an outline of the quantum SVE algorithm. Finally, we put the components together and present the QDLS algorithm.

\section{Memory model}
\label{sub: Memory}
In order to keep our description of the memory model general, we consider a rectangular matrix $A \in \mathbb{R}^{m \times n}$.
Instead of using a model that allows for black-box access to the matrix elements, here we work in a model that realises a data structure which satisfies the following properties:
\begin{itemize}
  \item Given access to the data structure, a quantum computer can perform the following mappings in $\Ord{\text{polylog}(mn)}$ time.
  \begin{align}
U_\mathcal{M}: \ket{i}\ket{0}\rightarrow\ket{i,\vec{A_i}} &= \frac{1}{\|\vec{A_i}\|}\sum\limits_{j=1}^{n}A_{ij}\ket{i,j},\notag \\
U_\mathcal{N}: \ket{0}\ket{j}\rightarrow\ket{\vec{A}_F,j} &= \frac{1}{\|A\|_F}\sum\limits_{i=1}^{m}\|\vec{A_i}\|\ket{i,j},
  \end{align}
where $\vec{A_i}\in \mathbb{R}^{n}$ are the row vectors of the matrix $A$ and $\vec{A}_F\in \mathbb{R}^{m}$ is a vector of the Euclidean norms 
of the rows, i.e.\ $(\vec{A}_F)_i=\norm{A_i}$. 
  \item The time needed to store a new entry $A_{ij}$ is in $\Ord{\text{log}^2(mn)}$. The data structure has size $\Ord{w\log (mn)}$ with $w$ denoting the number of non-zero entries in $A$.
  \end{itemize}
One way to construct the data structure as described with the above-desired properties is based on a size $m$ array of binary trees, where each tree contains no more than $n$ leaves. The leaves store the squared values of the corresponding matrix element $|A_{ij}|^2$, together with the sign, $sgn(A_{ij})$. Each internal node of a binary tree stores the summation of the values in the subtree rooted at it, as such the root of the $i^{th}$ tree stores $\norm{\vec{A}_i}^2, \, i \in [m]$. In order to access row Frobenius norm vectors, one merely need to construct an additional binary tree, in which the $i^{th}$ leaf contains $\norm{\vec{A}_i} ^2$. We show the schematic diagram to demonstrate one of the trees in Figure \ref{one tree}. More details about the realisation of this data structure can be found in Ref. \cite{Kerenidis2016} and \cite{Prakash2014}. 

\begin{figure}[H] 
\centering
\includegraphics[width=\textwidth, height=90mm] {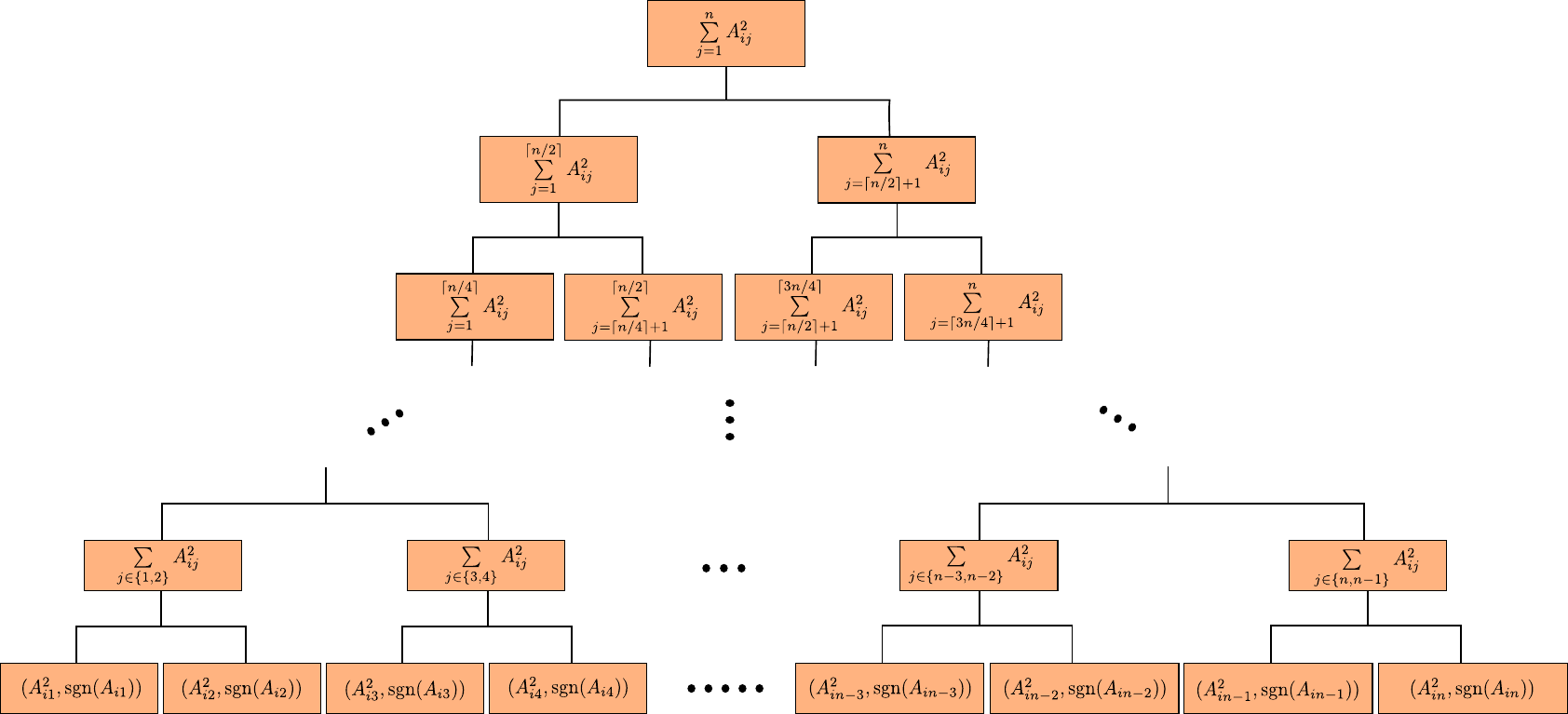} 
\caption{The schematic diagram for one out of the $(m+1)$ binary trees in the data structure used to store the matrix $A$. The depth of the tree is in $\mathcal{O}(\log n)$.} 
\label{one tree} 
\end{figure} 

\section{Quantum singular value estimation}
\label{QSVE}

Having stated the memory model, we are now in the position to outline the quantum singular value estimation (QSVE) subroutine. The QSVE can be seen as an extension of phase estimation to non-unitary matrices. Let the matrix $A\in \mathbb{R}^{m \times n}$ have the singular value decomposition 
\begin{align}
A = \sum_{i}\sigma_{i}\mathbf{u_i}\mathbf{v_i}^\dagger,
\end{align}
where $\mathbf{u_i}$ and $\mathbf{v_i}$ are the left and right singular vectors respectively, and $\sigma_{i}$ are the corresponding singular values. 
Since the left and the right singular vectors each form a complete set of orthonormal bases, an arbitrary input state can be written as the superposition, $\sum_{i}\alpha_i\ket{\mathbf{v}_i}$, where $\ket{\mathbf{v}_i}$ is a quantum state vector which encodes $\mathbf{v}_i$. The quantum SVE subroutine performs the following mapping, 
\begin{align}
    \sum_{i}\alpha_i\ket{\mathbf{v}_i} \to \sum_{i} \alpha_i \ket{\mathbf{v}_i} \ket{\overline{\sigma}_i},
\end{align}
where $\ket{\overline{\sigma_{i}}}$ is a state vector encoding the estimates for the singular values of $A$ with a precision $\delta$, so that $|\overline{\sigma_{i}} - \sigma_{i}|\leq \delta$ for all $i$. 

An algorithm for QSVE with a runtime of $\tilde{\mathcal{O}}(\norm{A}_{F} /\delta)$ was introduced in Ref. \cite{Kerenidis2016}, and applied to quantum recommendation systems. It is the main tool required for the quantum dense linear system (QDLS) algorithm \cite{wossnig2018quantum} to be presented in this chapter. In this section, we first give an overview of QSVE with essential mathematical background and high-level intuition. Then we outline the procedures of QSVE. Finally, we provide by a brief analysis of the algorithm, while a more thorough analysis can be found in Ref. \cite{Prakash2014, Kerenidis2016}.

\subsection{Overview}

The QSVE algorithm is based on the idea of quantum walks. It makes use of the connection between the singular values $\sigma_i$ of the matrix $A=\sum_{i}\sigma_{i}\mathbf{u_i}\mathbf{v_i}^\dagger$ and the principal angles, $\theta_i$ between certain associated subspaces. There exist a factorisation, 
\begin{align}
\frac{A}{\|A\|_F}= \mathcal{M}^{\dagger}\mathcal{N},    \label{factorisation}
\end{align}
where $\mathcal{M}\in\mathbb{R}^{mn\times m}$ and $\mathcal{N}\in\mathbb{R}^{mn\times n}$ are isometries with column spaces denoted $\mathcal{C}_\mathcal{M}$ and $\mathcal{C}_\mathcal{N}$ respectively.
The isometries $\mathcal{M}$ act on an arbitrary input state vector $\ket{\alpha}$ as a mapping that appends a register which stores the row vectors $\vec{A_i}$, such that
\begin{align}
\mathcal{M}: \ket{\alpha}&=\sum\limits_{i=1}^{m}\alpha_i\ket{i}\rightarrow\sum\limits_{i=1}^{m}\alpha_i\ket{i,\vec{A_i}}=\ket{\mathcal{M}\alpha}.
\end{align}
similarly, the isometries $\mathcal{N}$ act on an arbitrary input state vector $\ket{\alpha}$ as a mapping that appends a register which stores the vector $\vec{A_F}$, in which the entries are the row vector Euclidean norms $\|\vec{A_i}\|$, such that 
\begin{align}
\mathcal{N}: \ket{\alpha}=\sum\limits_{j=1}^{n}\alpha_j\ket{j}\rightarrow\sum\limits_{j=1}^{n}\alpha_j\ket{\vec{A_F},j}=\ket{\mathcal{N}\alpha}.
\end{align}
The above maps can be efficiently implemented given the memory model as described in Section \ref{sub: Memory}. The factorisation Eq. \ref{factorisation} then follows directly from the amplitude encodings of $\vec{A_i}$ and $\vec{A}_F$, as we have
\begin{align}
\ket{i,\vec{A_i}}&=\frac{1}{\|\vec{A_i}\|}\sum\limits_{j=1}^{n}A_{ij}\ket{i,j},\nonumber\\
\ket{\vec{A}_F,j}&=\frac{1}{\|A\|_F}\sum\limits_{i=1}^{m}\|\vec{A_i}\|\ket{i,j}.    
\end{align}
Taking the inner product of the above equations leads to
\begin{align}
(\mathcal{M}^{\dagger}\mathcal{N})_{ij}=
\langle i,\vec{A_i} \vert \vec{A}_F,j \rangle
=\frac{A_{ij}}{\|A\|_F}. 
\end{align}
A similar calculation shows that $\mathcal{M}$ and $\mathcal{N}$ have orthonormal columns and thus $\mathcal{M}^{\dagger}\mathcal{M}=I_m$ and $\mathcal{N}^{\dagger}\mathcal{N}=I_n$. 
The singular values of the normalised matrix $\frac{A}{\norm{A}_{F}}$ have a one-to-one correspondence to the principal angles between the subspaces $\mathcal{C}_\mathcal{M}$ and $\mathcal{C}_\mathcal{N}$. 
The efficiency of QSVE relies on the fact that given the matrix $A$ stored in a data structure as described in Section \ref{sub: Memory}, the following unitary operator $W$ can be implemented efficiently, 
\begin{align}
W=(2\mathcal{M}\mathcal{M}^{\dagger}-I_{mn})(2\mathcal{N}\mathcal{N}^{\dagger}-I_{mn}),     
\end{align}
where $I_{mn}$ denotes the $(mn)\times (mn)$ identity matrix. 
Note the fact that $W$ acts on $\ket{\mathcal{N}\mathbf{v}_i}$ as a rotation in the plane of $\{\mathcal{M}\mathbf{u}_i,\mathcal{N}\mathbf{v}_i\}$ by $\theta_i$. Hence the two dimensional sub-space spanned by $\{\mathcal{M}\mathbf{u}_i,\mathcal{N}\mathbf{v}_i\}$ is invariant under $W$. The eigenvectors of $W$, $\ket{w_i^\pm}$, therefore have corresponding eigenvalues $\exp(\pm i\theta_i)$. We can write in the eigenbasis of $W$, $\ket{\mathcal{N}\mathbf{v}_i}=\omega_i^+\ket{w_i^+}+\omega_i^-\ket{w_i^-}$, with $|\omega_i^-|^2+|\omega_i^+|^2=1$, and phase estimation can be performed to estimate $\pm {\theta}_i$. Finally the singular values of $A$, $\{\sigma_{i}\}$ are computed via $\sigma_{i} = \cos(\theta_{i}/2) \norm{A}_{F}$, a relation which will be shown in Section \ref{sub: Brief analysis}.

Intuitively the operators $2\mathcal{M}\mathcal{M}^{\dagger}-I_{mn}$ and $2\mathcal{N}\mathcal{N}^{\dagger}-I_{mn}$ can be seen as a generalisation of the Grover diffusion operator \cite{grover1996fast}. They act on the subspaces $\mathcal{C}_\mathcal{M}$ and $\mathcal{C}_\mathcal{N}$ respectively as reflection operators. Thus applying $W$ represents two sequential reflections, on the $\mathcal{C}_\mathcal{N}$ and then the $\mathcal{C}_\mathcal{M}$ subspaces. As such $W$ has the interpretation of taking a step in the bipartite quantum walk as formulated in Ref. \cite{Szegedy04} with the discriminant matrix given by our normalised target matrix $\frac{A}{\|A\|_F}$. The connections between quantum walks and the eigenvalues of the discriminant matrix are well-known in the literature and have been used in numerous previous works \cite{Szegedy04, C10, santha2008quantum}.

\subsection{Procedures}
Having introduced the mathematical background, we are now in the position to outline the procedures of QSVE.
\begin{enumerate}
  \item Create an arbitrary input state $\ket{\alpha} = \sum_i \alpha_{i} \ket{\mathbf{v}_i}$, where $\ket{\mathbf{v}_i}$ encodes the $i^{th}$ normalised left singular vector of $A$.
  \item Append a register with size $\log{m}$, $\ket{0^{\lceil \log{m} \rceil}}$. Query the data structure to apply $U_\mathcal{N}$, and create the state 
 \begin{align}
     \ket{\mathcal{N} \alpha} = \sum_i \alpha_{i} \ket{\mathcal{N} v_i}=\sum_i \alpha_{i}(\omega_i^+\ket{w_i^+}+\omega_i^-\ket{w_i^-}).
 \end{align}
  \item Perform phase estimation \cite{Kitaev1995} with precision $2 \delta >0$ on input $\ket{\mathcal{N}\alpha}$ for $W=(2\mathcal{M}\mathcal{M}^{\dagger}-I_{mn})(2\mathcal{N}\mathcal{N}^{\dagger}-I_{mn})$ and obtain $\sum_i \alpha_{i}(\omega_i^+\ket{w_i^+,\overline{\theta}_i}+\omega_i^-\ket{w_i^-,-\overline{\theta}_i})$, where $\overline{\theta_i}$ is the estimated phase $\theta_i$, such that $|\overline{\theta}_i -\theta_i| \leq 2 \delta$.
  \item On the output register of phase estimation  compute $\overline{\sigma}_i = \cos{(\pm\overline{\theta_i}/2)}||A||_F$ to obtain $\sum_i \alpha_{i}(\omega_i^+\ket{w_i^+}+\omega_i^-\ket{w_i^-})\ket{\overline{\sigma}_i}$.
  \item Apply the reversed computation of Step 2 to obtain
      \begin{align}
          \sum\limits_i \alpha_{i} \ket{\mathbf{v}_i} \ket{\overline{\sigma}_i}.
      \end{align}
\end{enumerate}
\subsection{Brief analysis}
\label{sub: Brief analysis}

Stated in a compact manner, the correctness and efficiency of the QSVE algorithm rely on the following:
\begin{itemize}
\item The mappings, $\ket{\alpha} \rightarrow \ket{\mathcal{M}\alpha}$ and $\ket{\alpha} \rightarrow \ket{\mathcal{N}\alpha}$ can be performed in time $\Ord{\text{polylog}(mn)}$, where the isometries $\mathcal{M}$ and $\mathcal{N}$ satisfy $\mathcal{M}^{\dagger}\mathcal{M} = I_m$ and $\mathcal{N}^{\dagger}\mathcal{N} = I_n$ and the factorisation $A/\norm{A}_F = \mathcal{M}^{\dagger}\mathcal{N}$. 
\item The reflection operators $(2\mathcal{M}\mathcal{M}^{\dagger}-I_{mn})$ and $(2\mathcal{N}\mathcal{N}^{\dagger}-I_{mn})$, hence the unitary $W=(2\mathcal{M}\mathcal{M}^{\dagger}-I_{mn})(2\mathcal{N}\mathcal{N}^{\dagger}-I_{mn})$ can be implemented in time $\Ord{\text{polylog}(mn)}$. The unitary $W$ acts on $\ket{\mathcal{N}v_i}$ as a rotation in the plane of $\{\mathcal{M}\mathbf{u}_i,\mathcal{N}\mathbf{v}_i\}$ by $\theta_i$, such that $\sigma_i = \cos\frac{\theta_i}{2}\|A\|_F$, where $\sigma_{i}$ is the $i^{th}$ singular value for $A$. 
\end{itemize}
As previously shown, the first two items in the above listing are guaranteed by applying the appropriate data structure in Section \ref{sub: Memory}. It remains to show the relationship between the eigenvalues of $W$ and the singular values of $A$.
We start by considering the action of $W$ as follows,
\begin{align}
W\ket{\mathcal{N}v_i}=&(2\mathcal{M}\mathcal{M}^{\dagger}-I_{mn})(2\mathcal{N}\mathcal{N}^{\dagger}-I_{mn})\ket{\mathcal{N}v_i}\nonumber\\
=&(2\mathcal{M}\mathcal{M}^{\dagger}-I_{mn})\ket{\mathcal{N}v_i}\nonumber\\=&2\mathcal{M}\frac{A}{\|A\|_F}\ket{\mathbf{v}_i}-\ket{\mathcal{N}v_i}.
\end{align}
Using the singular value decomposition $A=\sum_i\sigma_i\ket{\mathbf{u}_i}\bra{\mathbf{v}_i}$, and the fact that the right singular vectors $\{\mathbf{v}_i\}$ are mutually orthonormal, we have
\begin{align}
W\ket{\mathcal{N}v_i}=\frac{2\sigma_i}{\|A\|_F}\ket{\mathcal{M}u_i}-\ket{\mathcal{N}v_i}.
\end{align} 
It is now visible that $W$ has rotated $\ket{\mathcal{N}v_i}$ in the plane of $\{\mathcal{M}\mathbf{u}_i,\mathcal{N}\mathbf{v}_i\}$ by an angle $\theta_i$, such that 
\begin{align}
    \cos\theta_i&=\bra{\mathcal{N}v_i}W\ket{\mathcal{N}v_i} \notag \\ 
    &=\frac{2\sigma_i}{\|A\|_F^2} \bra{v_i} A^{\dagger} \ket{\mathbf{u}_i}-1 \notag \\
    &=\frac{2\sigma_i^2}{\|A\|_F^2}-1.\label{cos}
\end{align} 
Note that we have used the fact $(2\mathcal{M}\mathcal{M}^\dagger-I_{mn})$ is a reflection in $\ket{\mathcal{M}u_i}$
and that $A^{\dagger} = \mathcal{N}^{\dagger}\mathcal{M}= \sum_i \sigma_i \ket{\mathbf{v}_i} \bra{u_i}$. 
Hence we have established the angle between $\ket{\mathcal{N}v_i}$ and $\ket{\mathcal{M}u_i}$ is $\frac{\theta_i}{2}$, which amounts to half of the total rotation angle. Comparing the last line of Eq. \ref{cos} with the half-angle formula for cosine functions, leads to the desired relation,
\begin{align}
\cos\left(\frac{\theta_i}{2}\right)=\frac{\sigma_i}{\|A\|_F}.    
\end{align}

The run time of QSVE is dominated by the phase estimation which returns an $\delta$-error estimated eigenvalue $\overline{\theta}_i$, such that $|\overline{\theta}_i -\theta_i| \leq 2 \delta$. This error propagates to the estimated singular value as $\overline{\sigma}_i = \cos{(\overline{\theta}_i/2)} \norm{A}_F$. 
The error in the estimated singular values can hence be bounded from above by $|\overline{\sigma}_i - \sigma_i | \leq \delta \norm{A}_F$. 
The unitary $W$ can be implemented in time $\Ord{\text{polylog}(mn)}$ by using the suitable data structure. In summary, the total runtime for quantum singular value estimation with additive error $\delta \norm{A}_F$ is in $\Ord{\text{polylog}(mn)/\delta}$.

\section{The QDLS algorithm}
The application of the QSVE algorithm is particularly interesting for solving linear systems with a dense matrix since the QSVE runtime depends on the Frobenius norm $\norm{A}_{F}$, instead of the sparsity $s(A)$.
We now show that applying the QSVE algorithm leads to an efficient quantum algorithm for solving dense linear systems. Recall the fact that given a symmetric matrix $A \in \mathbb{R}^{n\times n}$ with spectral decomposition 
\begin{align}
A=\sum_{i \in [n]}\lambda_i\mathbf{s}_i\mathbf{s}_i^{\dagger},
\end{align}
the singular value decomposition of $A$ has the form of 
\begin{align}
A=\sum_{i \in [n]}|\lambda_i| \mathbf{u}_i\mathbf{v}_i^{\dagger},
\end{align}
 where the left and right singular vectors $\mathbf{u}_i$ and $\mathbf{v}_i$ are equivalent to the eigenvectors $\mathbf{s}_i$ up to an ambiguity of sign, such that $\mathbf{s}_i= \mathbf{u}_i=\pm\mathbf{v}_i$.     
Applying the QSVE algorithm to a positive definite matrix immediately yields the solution of the linear system as the estimated singular values and eigenvalues are equal, $\overline{\sigma_i}=|\overline{\lambda_i}|$. For a general symmetric matrix, QSVE leads to the estimation of $|\lambda_{i}|$ but not its sign, $sign(\lambda_{i})$. Therefore in order to solve general linear systems, we need a to recover $sign(\lambda_{i})$.

In Section \ref{sub: QDLS}, we will present a linear system algorithm with a simple technique to recovers the signs using the QSVE procedure as an oracle incurring only a constant multiplicative overhead. We assume that $A$ has been rescaled so that its eigenvalues lie within the interval $[-1, -1/\kappa] \cup [ 1/\kappa, 1]$, where $\kappa$ denotes the condition number of $A$. This is the same assumption made in \cite{Harrow2009a} and also indicated in the review \cite{H14}. We will show in Section \ref{sub: QDLS analysis} the algorithm runs in $\tilde{\mathcal{O}}(\sqrt{n})$ time for arbitrary matrices with a bounded spectral norm, and hence has no explicit dependence on the sparsity. 

\subsection{Procedures}
\label{sub: QDLS}
We are now in the position to outline the procedures of the quantum dense linear system algorithm.
  \begin{enumerate}
    \item Prepare a quantum state $\ket{\mathbf{b}}$ which encodes the vector $\mathbf{b}\in \mathbb{R}^n$ as 
    \begin{align}
                \ket{\mathbf{b}}=(\mathbf{b}^T\mathbf{b})^{-1/2}\sum_i \beta_{i} \ket{\mathbf{v}_i},
    \end{align}

where $\ket{\mathbf{v}_i}$ stores the $i^{th}$ singular vectors of $A$.
\item Perform the QSVE algorithm for matrices $A$ and for $A^\prime = A+ \mu I$ with $\delta < 1/2\kappa $ 
    and $\mu=1/\kappa$ to obtain
    \begin{align}
        (\mathbf{b}^T\mathbf{b})^{-1/2}\sum_i \beta_{i}  \ket{\mathbf{v}_i}_\mathrm{A} \ket{|\overline{\lambda}_i|}_\mathrm{B} \ket{|\overline{\lambda}_i + \mu|}_\mathrm{C}.
    \end{align}

\item Append an ancillary register and set its value to $1$ if the value in register $\mathrm{B}$ is greater than that in register $\mathrm{C}$ and apply a conditional
    phase gate, which leads to
    \begin{align}
        (\mathbf{b}^T\mathbf{b})^{-1/2}\sum_i (-1)^{f_i} \beta_{i}  \ket{\mathbf{v}_i}_\mathrm{A} \ket{|\overline{\lambda}_i|}_\mathrm{B} \ket{|\overline{\lambda}_i + \mu|}_\mathrm{C} \ket{f_i}_\mathrm{D}.
    \end{align}
\item Append another ancillary register and apply a rotation conditioned on register $\mathrm{B}$ with $\gamma=\mathcal{O}(1/\kappa)$.  
        Then uncompute the registers $\mathrm{B}$, $\mathrm{C}$ and $\mathrm{D}$ to obtain
    \begin{align}
    (\mathbf{b}^T\mathbf{b})^{-1/2}\sum_i(-1)^{f_i} \beta_{i} \ket{\mathbf{v_i}}\left( \frac{\gamma}{|\overline{\lambda_i}|}\ket{0} + \sqrt {1 - \left(\frac{\gamma}{|\overline{\lambda_i}|}\right)^{2}}  \ket{1} \right).
    \end{align}
A high-level circuit diagram that describes the QDLS algorithm until this step is shown in Figure \ref{fig: QDLSA}.      
\item Measure the last register in the computational basis. Conditioned on obtaining $\ket{0}$, the system is in the desired state,
\begin{align}
    (\mathbf{b}^T\mathbf{b})^{-1/2}\sum_i(-1)^{f_i}\frac{\beta_{i}}{|\overline{\lambda_i}|}\ket{\mathbf{v_i}}.
\end{align}
\end{enumerate}

\begin{figure}[H] 
\centering
\includegraphics[width=\textwidth] {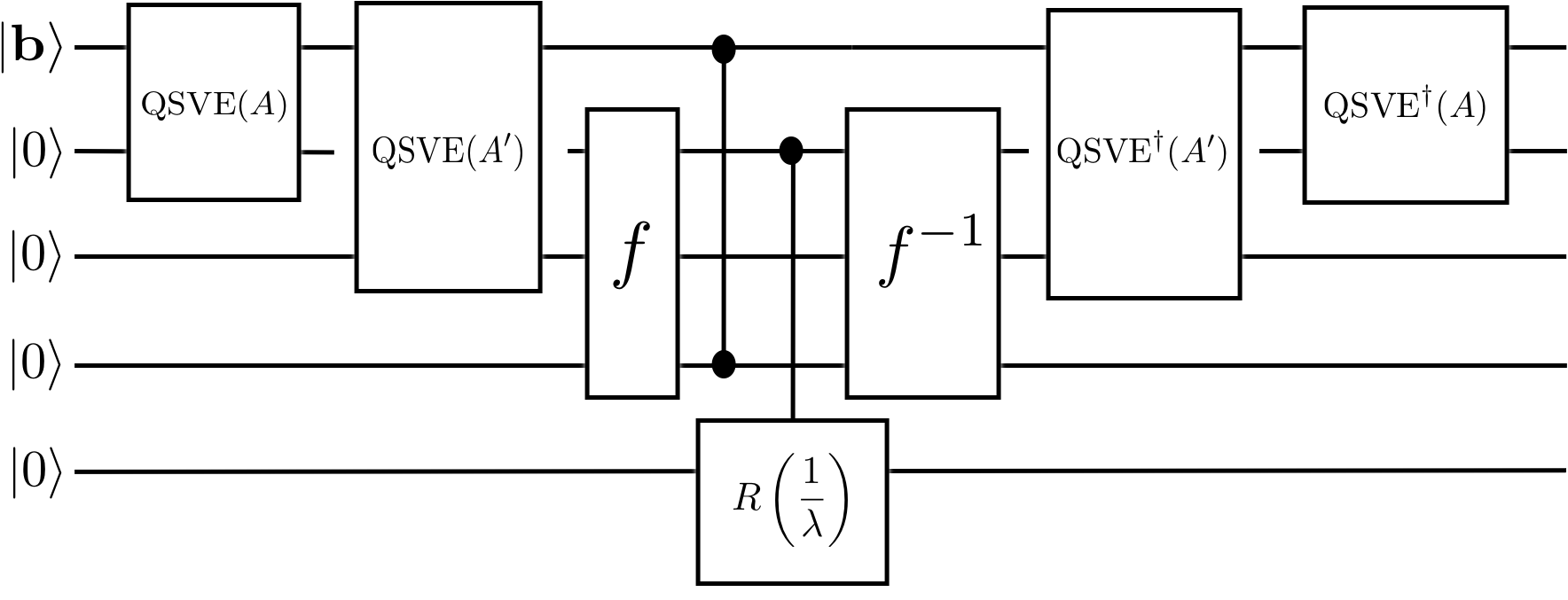} 
\caption{The circuit diagram for the QDLS algorithm until Step. 4. The QSVE subroutine is applied for $A$ and $A^\prime = A+ \mu I$, obtaining the respective singular values in superposition stored in two quantum registers. Then the function $f$ compares the value of the second and third registers, and stores the outcome in the fourth register. A phase gate is then applied conditioned on the value of the fourth register, which successfully recovers the sign of the desired eigenvalues for inversion. The subsequent controlled rotation and uncomputation proceed similarly as the original linear system algorithm of Ref.\cite{Harrow2009a}.}
\label{fig: QDLSA} 
\end{figure} 

\subsection{Analysis}
\label{sub: QDLS analysis}

\paragraph{Sign recovery}We first argue that above QDLS algorithm correctly recovers the sign of the $\lambda_{i}$. The algorithm compares the estimates obtained by performing QSVE for $A$ and for $A^\prime = A+\mu I_{n}$, where $\mu$ is a positive scalar chosen to be the inverse of the condition number, $\kappa$. The matrix $A^\prime$ has eigenvalues $\lambda_i + \mu$ while the corresponding eigenvectors are the same as those of $A$. Not that if $\lambda_{i} \geq 0$, we have 
\begin{align}
|\lambda_i + \mu| = |\lambda_i| + 
    |\mu| \geq |\lambda_{i}|.    
\end{align}
However if $\lambda_{i} < -\mu/2 $, then we instead have
\begin{align}
|\lambda_{i} + \mu| < |\lambda_{i}|.     
\end{align}
Thus if we had perfect estimation for the singular values, then choosing $\mu < 2/\kappa$ would recover the sign correctly as the eigenvalues of $A$ lie in the interval $[-1, -1/\kappa] \cup [ 1/\kappa, 1]$. In an imperfect setting of QSVE, with the choice of $\mu=1/\kappa$ and $\delta < 1/2\kappa$ the signs are correctly recovered for all $\lambda_{i}$. 
       
\paragraph{Runtime and errors} As we will show later in this section, the additive error $\epsilon$ for the linear system solver is related to the QSVE precision parameter $\delta$ via $\delta=\Ord{\frac{\epsilon}{\kappa\|A\|_F}}$. 
  The success probability of the final post-selection step requires on average $\Ord{\kappa^2}$ repetitions of the coherent computation. However, applying amplitude amplification \cite{Brassard2002,Ambainis2010} can reduce this cost to $\Ord{\kappa}$.
    Hence an upper-bound of the runtime of our algorithm is given by $\Ord{\kappa^2 \cdot \text{polylog}(n) \norm{A}_F /\epsilon}$.
The error dependence on the Frobenius norm suggests that our algorithm is most accurate when the $\norm{A}_F$ is bounded by some constant, in which case the algorithm returns the output state with a constant $\epsilon$-error in polylogarithmic time even if the matrix is non-sparse. More generally, as in the HHL algorithm \cite{Harrow2009a}, we can assume that the spectral norm $\norm{A}_*$ is bounded by a constant, although the Frobenius norm may scale with the dimensionality of the matrix. In such cases we have $\norm{A}_F=\Ord{\sqrt{n}}$, and the QDLS algorithm runs in time $\Ord{\kappa^2\sqrt{n} \cdot\text{polylog}(n)/\epsilon}$ and returns the output with a constant $\epsilon$ additive error. Furthermore, since $\norm{A}_F \leq \sqrt{r} \norm{A}_*$, where $r$ denotes the rank of $A$, the runtime may also be written as $\Ord{\kappa^2\sqrt{r} \cdot\text{polylog}(n)/\epsilon}$. 
Hence an exponentially more advantageous runtime is achievable if the rank of $A$ is polylogarithmic in $n$.

\paragraph{Error bound details}
We now establish error bounds on the final state. In a similar fashion to the analysis of Ref. \cite{Harrow2009a},
we use the filter functions $\mathrm{f}$ and $\mathrm{g}$~\cite{Hansen1998}, which allow us to invert only the well-conditioned part of the matrix, that is, the space which is spanned by the eigenvectors with eigenvalues, $\lambda_i\geq 1/\kappa$. We define the functions, 
\begin{align}
\label{eq:HHLfiltfunctions}
\mathrm{f}(\lambda) =
\begin{cases}
\frac{1}{\kappa\gamma\lambda}, & |\lambda| \geq 1/\kappa;\\
\eta_{1}(\lambda), & \frac{1}{\kappa}>|\lambda|>\frac{1}{2\kappa};\\
0, & \frac{1}{2\kappa}\geq|\lambda|;
\end{cases}
\end{align}
and
\begin{align}
\mathrm{g}(\lambda) =
\begin{cases}
0, & |\lambda| \geq 1/\kappa;\\
\eta_{2}(\lambda), & \frac{1}{\kappa}>|\lambda|>\frac{1}{2\kappa};\\
\frac{1}{2}, & \frac{1}{2\kappa}>|\lambda|,
\end{cases}
\end{align}
where $\gamma=\Ord{1/\kappa}$ is the parameter chosen in Step 4 of the algorithm in Section \ref{sub: QDLS} to ensure that the probability amplitudes are bounded by unity after the controlled rotation by any eigenvalues. The functions $\eta_{1}$ and $\eta_{2}$ are interpolating functions chosen such that $\mathrm{f}^{2}(\lambda) + \mathrm{g}^{2}(\lambda) \leq 1$ for all $\lambda \in \mathbb{R}$. A possible (non-unique) choice of $\eta_{1}$ and $\eta_{2}$ can be
\begin{align}
\eta_{1}(\lambda)=\frac{1}{2}\sin(\frac{\pi}{2}\cdot\frac{\lambda-\frac{1}{\kappa'}}{\frac{1}{\kappa}-\frac{1}{\kappa'}}),
\end{align}
and
\begin{align}
\eta_{2}(\lambda)=\frac{1}{2}\cos(\frac{\pi}{2}\cdot\frac{\lambda-\frac{1}{\kappa'}}{\frac{1}{\kappa}-\frac{1}{\kappa'}}),
\end{align}
Note that the presented QDLS algorithm corresponds to the choice $\mathrm{g}(\lambda)=0$. We then define the map
\begin{align}
     \ket{\mathrm{h}(\lambda)} := \sqrt{1-\mathrm{f}(\lambda)^2-\mathrm{g}(\lambda)^2} \ket{\mathrm{NO}} + \mathrm{f}(\lambda) \ket{\mathrm{WC}} + \mathrm{g}(\lambda) \ket{\mathrm{IC}},    
\end{align}
where $\ket{\mathrm{NO}}$ indicates that no matrix inversion has taken place, $\ket{\mathrm{IC}}$ means that 
part of $\ket{\mathbf{b}}$ is in the ill-conditioned subspace of $A$, and $\ket{\mathrm{WC}}$ means that the matrix inversion has taken place and is
in the well-conditioned subspace of $A$. 
This allows us to invert only the well-conditioned part of the matrix while it flags the ill-conditioned ones and interpolates between those two behaviours when $1/(2\kappa) < |\lambda| < 1/\kappa$. We therefore only invert eigenvalues which are larger than $1/(2\kappa)$. This subtlety is the motivation behind our choice of $\mu$ in the algorithm.

Let $\mathcal{Q}$ denote the error-free operation corresponding to the QSVE subroutine followed by the controlled rotation without post-selection, such that
\begin{align}
    \ket{\psi} := \mathcal{Q} \ket{\mathbf b}\ket{0} \rightarrow \sum_i \beta_i \ket{\mathbf  v_i} \ket{\mathrm{h}(\lambda_i)}. \label{eq: phi}
\end{align}
 $\overline{\mathcal{Q}}$ in contrast describes the same procedure but the phase estimation step is erroneous, such that
 \begin{align}
     \ket{\overline{\psi}} := \overline{\mathcal{Q}} \ket{\mathbf  b}\ket{0} \rightarrow \sum_i \beta_i \ket{\mathbf v_i} \ket{\mathrm{h}(\overline{\lambda}_i)}. \label{eq: phi bar}
 \end{align}
In order to bound the error, $\norm{\overline{\mathcal{Q}} - \mathcal{Q}}$, we choose a general state $\ket{\mathbf b}$, and find the equivalent error bound $\norm{ \mathcal{Q} \ket{\mathbf b} - \overline{\mathcal{Q}} \ket{\mathbf b}} := \norm{ \ket{\overline{\psi}} - \ket{\psi} }$. 
We need to make use of the fact that the map $\lambda \rightarrow \ket{h(\lambda)}$ is $\Ord{\kappa}$-Lipschitz \cite{Harrow2009a}. That is to say $\forall \lambda_i \neq \lambda_j$ for some $c \leq \pi/2 = \Ord{1}$, we have
    \begin{align}
        \norm{ \ket{\mathrm{h}(\lambda_i)} - \ket{\mathrm{h}(\lambda_j)} } 
        \leq c \kappa | \lambda_i - \lambda_j|.\label{Lipschitz}
    \end{align}
Note that it suffices 
to lower-bound $\text{Re}(\langle \overline{\psi} \vert \psi \rangle)$ since we have 
\begin{align}
\norm{ \ket{\overline{\psi}} - \ket{\psi} } = \sqrt{2 \left(1-\text{Re}(\langle \overline{\psi} \vert \psi \rangle) \right)},    
\end{align}
Now we take the inner product between Eq. \ref{eq: phi} and Eq. \ref{eq: phi bar} to obtain 
\begin{align}
      \text{Re}(\langle \overline{\psi} \vert \psi \rangle) = \sum\limits_i |\beta_i |^2  \text{Re}(\langle\mathrm{h}(\overline{\lambda}_i)\vert\mathrm{h}(\lambda_i)\rangle). 
\end{align}
Next we use the error bounds of the QSVE subroutine for the eigenvalue distance, i.e. $| \lambda_i - \overline{\lambda}_i | \leq \delta \norm{A}_F$, which leads to 
\begin{align}
\text{Re}(\langle \overline{\psi} \vert \psi \rangle) \geq \sum\limits_i |\beta_i |^2 \left( 1 - \frac{c^2 \kappa^2 \delta^2 \norm{A}_F^2}{2} \right).    
\end{align}
This is a consequence of the finite accuracy phase estimation, and the $\Ord{\kappa}$-Lipschitz property of Eq. \ref{Lipschitz}. Since $ 0 \leq \text{Re}(\langle \overline{\psi} \vert \psi \rangle) \leq 1$, it follows that 
\begin{align}
    1 - \text{Re}(\langle \overline{\psi} \vert \psi \rangle) \leq 
    \sum\limits_i |\beta_i |^2 \left(\frac{c^2 \kappa^2 \delta^2 \norm{A}_F^2}{2} \right). 
\end{align}
Finally we use the fact that $\sum_i |\beta_i|^2 = 1$, the distance can be bounded as 
\begin{align}
    \norm{ \ket{\overline{\psi}} - \ket{\psi} } \leq \Ord{\kappa \delta \norm{A}_F}.
\end{align}
If this additive error on the output state is needed to be on the order of $\epsilon$, we need to take the phase estimation accuracy to be $\delta = \Ord{\frac{\epsilon}{\kappa \norm{A}_F}}$. This results in a runtime that scales as $\Ord{\kappa \norm{A}_F \cdot \text{polylog}(n)/\epsilon}$.
In order to successfully perform the final post-selection step, we need to repeat the algorithm on average $\kappa^2$ times. This additional multiplicative factor of $\kappa^2$ can be reduced to $\kappa$ using amplitude amplification \cite{Brassard2002,Ambainis2010}.
Putting everything together, we have an overall runtime that scales as $\Ord{\kappa^2 \norm{A}_F \cdot \text{polylog}(n)/\epsilon}$.

\section{Summary and discussions}

We have shown in this chapter that given $\mathbf{b} \in \mathbb{R}^{n}$ and a Hermitian matrix $A \in \mathbb{R}^{n \times n}$ with spectral decomposition 
     $A = \sum_i \lambda_i \mathbf{s}_i \mathbf{s}_i^{\dagger}$ stored in a suitable data structure, the QDLS algorithm returns the state 
      $\ket{\overline{A^{-1} \mathbf b}}$ 
     such that $\norm{ \ket{ \overline{ A^{-1} \mathbf b } } - \ket{ A^{-1} \mathbf b} } \leq \epsilon $. The runtime of the algorithm scales as $\Ord{\kappa^2 \cdot \text{polylog}(n) \cdot \norm{A}_F/\epsilon}$, where $\kappa$ is the condition number and $\norm{A}_F$ is the Frobenius norm of $A$.

\paragraph{Bounded spectral norm} Assuming the spectral norm, $\norm{A}_*$, is bounded by a constant or grows no faster than polylogorithmically in $n$, the overall runtime scaling reduces to $\Ord{\kappa^2 \sqrt{n} \cdot \text{polylog}(n)/\epsilon}$, since we have 
\begin{align}
    \norm{A}_F&=\sqrt{\sum_i^n\lambda_i^2}
    \le\sqrt{n|\lambda|_{max}^2}
    \le\sqrt{n}\norm{A}_*.
\end{align}
This amounts to a polynomial speed-up over the runtime scaling achieved in Ref. \cite{Harrow2009a} when applied to dense matrices with black-box Hamiltonian simulation. The bounded spectral norm is a realistic assumption if classical normalisation preprocessing can be applied so that the maximum absolute values of $\lambda_i$ is bounded. As the same bounded spectral norm assumption is also required in the error analysis of Ref. \cite{Harrow2009a}, the algorithm presented in this chapter represents a new state-of-the-art for solving dense linear systems on a quantum computer.

\paragraph{Low-rank} In special cases, the matrix $A$ has a low-rank structure, such that the number of non-zero eigenvalues grows no faster than polylogarithmically in $n$. In such scenarios, the runtime of the presented QDLS algorithm scales as $\Ord{\kappa^2 \cdot \text{polylog}(n)/\epsilon}$, which amounts to an exponential improvement over previously existing algorithms for solving dense linear system problems.

\paragraph{Distinction in memory models} Note that the memory model described in Section \ref{sub: Memory} is distinct from the black-box model. This QSVE-based linear system algorithm achieves a $\tilde{\mathcal{O}}(\sqrt{n})$-scaling for dense matrices in this augmented quantum memory model, and it is an interesting question whether a similar scaling is achievable in the black-box matrix element access model. 

\paragraph{Non-invertible matrix} The SVE-based algorithm also applies to more general scenarios where the matrix $A$ is not invertible. Then the algorithm will instead compute the Moore-Penrose pseudo-inverse. The runtime in these cases will be bounded by $1/|\lambda_{min}|$ instead of $\kappa$, where $\lambda_{min}$ is the non-zero eigenvalue for $A$ with the smallest absolute value.     

\paragraph{Outlooks} 
From a practical point of view, the constant runtime overhead for a given set of elementary fault-tolerant quantum gates is an important consideration. Scherer \textit{et al.}~\cite{Scherer2017} showed that implementations of the HHL algorithm \cite{Harrow2009a} potentially suffer from a large constant overhead with currently available technology, which may hinder the prospects of near-term applications. Whether or not the SVE-based QDLS algorithm has considerably smaller constant overhead, due to the absence of Hamiltonian simulation, remains an interesting open question.
 
\part{Gaussian processes with quantum algorithms} 
\label{Part2}
\chapter{Gaussian processes in classical machine learning} 

\label{GP} 
In the previous chapters, we have introduced the basics of quantum computation and have seen some examples of quantum algorithms. Particularly, in Section \ref{HHL} of Chapter \ref{IntroQA} and in Chapter \ref{QDLSA}, we have seen a quantum variant of the linear systems problem can be efficiently solved by a quantum computer with the access to suitable memory models.
In this part of the thesis, we apply some of these quantum ideas to a powerful model in supervised machine learning, Gaussian processes (GP). To start with, in this chapter, we will follow the notation of Ref. \cite{Rasmussen2004} and introduce the basics of GPs and review the typical classical implementations of inference with GP models as well as GP model selection. In Chapters \ref{QGP} and \ref{QGPT}, we will follow closely Ref. \cite{zhao2015quantum} and \cite{zhao2018quantum} and present quantum algorithms for computing GP regression models and training GP regression models respectively. In Chapter \ref{Chapter: QBDL}, we will make use of the quantum GP algorithms to present a quantum approach to Bayesian deep learning.

\section{Introduction}

Supervised machine learning amounts to inferring a function from labelled training data \cite{mohri2012foundations}. The GPs represent an approach to supervised learning that models the latent functions associated with the outputs of an inference problem as an infinite-dimensional generalisation of a Gaussian distribution \cite{Rasmussen2004}. 
The GP approach offers numerous desirable properties such as being capable of capturing a wide range of behaviours with only a simple set of parameters, the ability to easily express uncertainty, and admitting a natural Bayesian interpretation.
As such GP models have been widely used across a broad spectrum of applications, ranging from robotics, data mining, geophysics (where GP approaches are also known as kriging), climate modelling, and predicting price behaviour of commodities in financial markets. 

Although GP models are becoming increasingly popular in the classical community of machine learning, they are known to be computationally expensive, which hinders their widespread adoption. A practical implementation of Gaussian process regression (GPR) model with $n$ training points typically requires $\Omega(n^3)$ basic operations \cite{Rasmussen2004}. This has lead to significant amount of effort aimed at reducing the computational cost of working with such models, with investigations into low-rank approximations of GPs \cite{quinonero2005unifying}, variational approximations \cite{hensman2013gaussian} and Bayesian model combination for distributed GPs \cite{deisenroth2015distributed}. A thorough discussion of these approximation methods is beyond the scope of this thesis. Instead, we will argue that quantum computation offers efficient exact implementation of GPR even when the size of the input data is classically infeasible. 

The contents of this chapter are organised as follows: In Section \ref{sec: GP prelim}, we will introduce some preliminary definitions and concepts necessary for describing GPs as regression models. In Section \ref{sec: GPR}, we will present the basics of GPR as well as its typical classical implementation. In Section \ref{sec: TGP} we will review the classical GP model selection procedures, with an emphasis on the figure of merit for a given model's performance. In Section \ref{sec: gpdl}, we will discuss the connection between GPs and deep neural networks, mainly following the results of \cite{lee2017deep}.
 This chapter provides only a basic level introduction to GPs. Readers are referred to Ref. \cite{Rasmussen2004, neal1994priors, lee2017deep, gmatthews2018gaussian} for further details.

\subsection{Preliminaries}
\label{sec: GP prelim}

\paragraph{Multivariate Gaussian distributions}
If a vector of random variables $\mathbf{x}\in \mathbb{R}^k$ follows a multivariate Gaussian distribution with a mean vector, $\boldsymbol{\mu}$ and a covariance matrix, $\Sigma$, its probability density function is given by,
\begin{align}
    p(\mathbf{x})=\frac{1}{\sqrt{(2\pi)^k|\Sigma|}}\exp\left(-\frac{1}{2}(\mathbf{x}-\boldsymbol{\mu})^T\Sigma^{-1}(\mathbf{x}-\boldsymbol{\mu}) \right),
\end{align}
where $|\Sigma|$ denotes the determinant of the covariance matrix. We denote this distribution as $\mathbf{x}\sim\mathcal{N}(\boldsymbol{\mu},\Sigma)$.
\paragraph{Gaussian processes}
A Gaussian process (GP) is defined as a set of random variables, any finite subset of which follows a joint multivariate Gaussian distribution \cite{Rasmussen2004}. A GP model is entirely specified by a prior mean function, $\mu(\mathbf{x})=\mathbb{E}[{f}(\mathbf{x})]$, and a covariance function (kernel), $k(\mathbf{x},\mathbf{x}^\prime)=\mathbb{E}[({f}(\mathbf{x})-\mu(\mathbf{x}))({f}(\mathbf{x}^\prime)-\mu(\mathbf{x}^\prime))]$ of some underlying actual process $f(\mathbf{x})$. We write 
\begin{align}
    f(\mathbf{x})\sim\mathcal{GP}(\mu(\mathbf{x}),k(\mathbf{x},\mathbf{x}^\prime))
\end{align}
to denote a Gaussian process.
For simplicity, we will assume the prior mean to be zero without loss of generality.

\paragraph{Marginalisation property}
As a requirement for consistency, models for statistical inference need to satisfy the following marginalisation property: 
Given a set of random variables $S$ and a statistical model that specifies a probability distribution $P$, for any subsets $S^\prime \subset S$, the corresponding probability distribution is given by the marginal distribution for $S^\prime$ in $P$. Intuitively, this means that the distribution of a larger set of variables needs to be consistent with the distribution of its subsets. This requirement is automatically satisfied by the GP definition.

\section{Gaussian process regression}
\label{sec: GPR}

In this section, we introduce Gaussian processes as a regression model, following closely Chapter 2 of Ref. \cite{Rasmussen2004}. We will consider a supervised learning problem with a training dataset $\mathcal{T}$ with $n$ $d$-dimensional input points, $\{\mathbf{x}_i\}_{i=0}^{n-1}$, and their corresponding output points, $\{y_i\}_{i=0}^{n-1}$, such that $\mathcal{T}=\{\mathbf{x}_i,y_i\}^{n-1}_{i=0}$. The goals is to infer an underlying function $f(\mathbf{x})$ from the observed input-output pairs subject to Gaussian random noise, 
\begin{align}
 y = {f}(\mathbf{x}) + \varepsilon_{\text{noise}},
\end{align}
where $\varepsilon_{\text{noise}}\sim\mathcal{N}(0, \sigma_n^2)$ is independent, identically distributed (i.i.d.) noise that follows a Gaussian distribution with $0$ mean and $\sigma_n^2$ variance. Since the underlying $f(\mathbf{x})$ is not directly observed, it is known as the ``latent function''. When given a new input ``test point'', $\mathbf{x}_*$, our model aims at generating a predictive distribution for ${f}_*={f}(\mathbf{x}_*)$. The Gaussian process regression approach models the latent function $\{{f}(\mathbf{x_i})\}_{i=0}^{n-1}$ as a joint multivariate Gaussian distribution \cite{Rasmussen2004}.

\subsection{Linear model with Gaussian noise} We start by considering the standard model of linear regression so that the underlying function $f(\mathbf{x})$ of an input vector $\mathbf{x}$ is given by
\begin{align}
    f(\mathbf{x})=\mathbf{x}^T\mathbf{w},
\end{align}
where the weight vector $\mathbf{w}$ contains the parameters of the linear model. Under our Gaussian noise assumption, the actual observed values $y$ are given by $y=\mathbf{x}^T\mathbf{w}+\varepsilon_{\text{noise}}$.

\paragraph{Likelihood}
The likelihood is defined by the probability density of the observed values conditioned on the given parameters. The independence assumption allows us to factor over the points in the whole training set, and write the likelihood as
\begin{align}
p(\mathbf{y}\vert X, \mathbf{w}) =& \prod_{i=1}^n p(y_i\vert \mathbf{x}_i,\mathbf{w})\nonumber\\
                                 =& \prod_{i=1}^n \frac{1}{\sigma_n\sqrt{2\pi}} \exp\left(-\frac{(y_i-\mathbf{x}^T\mathbf{w})^2
                                 )}{2\sigma_n^2}\right)\nonumber\\
                                 =& \left(\frac{1}{\sigma_n\sqrt{2\pi}}\right)^n \exp\left(-\frac{|\mathbf{y}-X^T\mathbf{w}|^2
                                 }{2\sigma_n^2}\right)\nonumber\\
                                 =& \mathcal{N}\left(X^T\mathbf{w},\sigma_n^2 I\right),
\end{align}
where $X\in \mathbb{R}^{d\times n}$ denotes the matrix containing the entire set of input data points, and $\mathbf{y}\in \mathbb{R}^n$ denotes the vector containing the entire set of output data points.
We have shown the likelihood of a linear model with Gaussian noise follows a multivariate Gaussian distribution with mean vector $X^T\mathbf{w}$ and covariance matrix $\sigma_n^2 I$. 
\paragraph{Bayesian inference}
In a Bayesian approach, we assume a prior distribution over $\mathbf{w}$ which expresses a belief about the parameters before observing the outputs. We choose a Gaussian prior with $\mathbf{0}$ mean and covariance matrix $\Sigma_p$: $\mathbf{w}\sim\mathcal{N}\left(\mathbf{0},\Sigma_p\right)$. 
Bayesian inference of the linear model follows the posterior distribution over $\mathbf{w}$, which can be evaluated using the Bayes' theorem,
\begin{align}
    p(\mathbf{w}|\mathbf{y},X)=\frac{p(\mathbf{y}|X,\mathbf{w})p(\mathbf{w})}{p(\mathbf{y}|X)}.
\end{align}
The factor $p(\mathbf{y}|X)$, known as the marginal likelihood, is independent of $\mathbf{w}$ and acts as a normalisation constant. We have posterior distribution, $p(\mathbf{w}|\mathbf{y},X)$ proportional to
\begin{align}
p(\mathbf{y}|X,\mathbf{w})p(\mathbf{w})
                                       = & \exp\left(-\frac{1}{2\sigma_n^2}(\mathbf{y}-X^T\mathbf{w})^T(\mathbf{y}-X^T\mathbf{w})\right)\exp\left(-\frac{1}{2}\mathbf{w}^T\Sigma_p^{-1}\mathbf{w}\right)\nonumber\\
                                       = & \exp\left(-\frac{1}{2}(\mathbf{w}-\bar{\mathbf{w}})^T\left(\frac{X^TX}{\sigma_n^2}+\Sigma_p^{-1} \right) (\mathbf{w}-\bar{\mathbf{w}})\right),
\end{align} 
where we have used the shorthand $\bar{\mathbf{w}}=\sigma_n^{-2}C^{-1}X\mathbf{y}$ with $C=\frac{XX^T}{\sigma_n^2}+\Sigma_p^{-1}$. Thus the posterior follows a Gaussian distribution, $p(\mathbf{w}|\mathbf{y},X)\sim\mathcal{N}\left(\sigma_n^{-2}C^{-1}X\mathbf{y},C^{-1}\right)$.
\paragraph{Predictive distribution}
Given a new input test point $\mathbf{x}_*$, the predictive distribution for $f_*=f(\mathbf{x}_*)$ can then by evaluated with following Bayesian integral,
\begin{align}
    p(f_*|\mathbf{x}_*, X, \mathbf{y}) =& \int p(f_*|\mathbf{x}_*, \mathbf{w})p(\mathbf{w}|X, \mathbf{y})d\mathbf{w} \nonumber \\
                                       =& \mathcal{N}\left(\frac{1}{\sigma_n^2}\mathbf{x}^T_*A^{-1}X\mathbf{y}, \mathbf{x}^T_*A^{-1}\mathbf{x}_*\right).
\end{align}
Therefore the predictive distribution $f_*$ conditioned on the training set with $X$ and $\mathbf{y}$ is a Gaussian with mean $\frac{1}{\sigma_n^2}\mathbf{x}^T_*A^{-1}X\mathbf{y}$ and variance $\mathbf{x}^T_*A^{-1}\mathbf{x}_*$.
\subsection{Feature space projection}
 The standard linear model described above often suffers from limited expressiveness and fails to capture interesting higher order features. A simple enhancement known as feature space projection overcomes this issue. The idea is to have the inputs mapped into certain chosen space with higher dimensions, which is specified by a set of basis functions. That is, we make use of a function $\phi(\mathbf{x})$ to project an input vector $\mathbf{x}\in \mathbb{R}^D$ into a feature space with dimension $N$. Instead of being applied directly on the inputs, the linear model is instead applied in this projected feature space. For instance, for a scalar input $x$, possible choices of the feature space basis functions include $\phi(x)=(1, x, x^2, ...)$, $\phi(x)=(1, \sin(x), \cos(x), ...)$, etc. The problem of choosing the appropriate basis functions is related to the model selection of GP, which we will address in Section \ref{sec: TGP}.
\paragraph{Prediction}
After the feature space projection, the regression model is augmented into $f(\mathbf{x})=\phi(\mathbf{x})^T\mathbf{w}$, with $\mathbf{w}\in\mathbb{R}^N$. Following an analogous Bayesian analysis as before, we arrive at the following formula for the predictive distribution, 
\begin{align}
p(f_*|\mathbf{x}_*,X,\mathbf{y})=\mathcal{N}\left(\frac{1}{\sigma_n^2}\phi(\mathbf{x}_*)^TC^{-1}\phi(X)\mathbf{y}, \phi(\mathbf{x}_*)^TC^{-1}\phi(\mathbf{x}_*)\right),    
\end{align}
with $C=\sigma_n^{-2}\phi(X)\phi(X)^T+\Sigma_p^{-1}.$ Alternatively, this can be written as 
\begin{align}
\mathcal{N}\left(\phi_*^T\Sigma_p\Phi(K+\sigma_n^2I)^{-1}\mathbf{y}, \phi_*^T\Sigma_p\phi_*-\phi_*^T\Sigma_p \Phi(K+\sigma_n^2I)^{-1}\Phi^T\Sigma_p\phi_*\right), \label{}
\end{align}
with the shorthand notations, $\Phi=\phi(X)$, $\phi_*=\phi(\mathbf{x}_*)$ and $K=\Phi^T\Sigma_p\Phi$.
\paragraph{Kernel trick}
We now replace the inner products in the feature space by functions
in the input space by defining the covariance function, $k(\mathbf{x},\mathbf{x}^\prime)=\phi(\mathbf{x})^T\Sigma_p\phi(\mathbf{x}^\prime)$, and the associated vector $\mathbf{k}=\Phi^{T}\Sigma_p \phi_*$. This leads to
\begin{align}
p(f_*|\mathbf{x}_*,X,\mathbf{y})=\mathcal{N}\left(\mathbf{k}^T(K+\sigma_n^2 I)^{-1}\mathbf{y}, k(\mathbf{x}_*,\mathbf{x}_*)-\mathbf{k}^T(K+\sigma_n^2 I)^{-1}\mathbf{k} \right)    .
\end{align}
The predictive distribution of $f_*$ is therefore a Gaussian distribution specified by $p({f}_*|\mathbf{x}_*,\mathcal{T})\sim\mathcal{N}(\bar{{f}_*},\mathbb{V}[{f}_*])$, where
\begin{align}
 \bar{{f}_*}&=\textbf{k}_*^T(K+\sigma_n^2{I})^{-1}\textbf{y} \label{eq:mean}\\
 \mathbb{V}[{f}_*]&=k\left(\textbf{x}_*,\textbf{x}_*\right)-\textbf{k}_*^T(K+\sigma_n^2{I})^{-1}\textbf{k}_*.\label{eq:variance} 
\end{align}
Hence the mean predictor Eq. \ref{eq:mean} and the variance Eq. \ref{eq:variance} are the central quantities of interest and the main goals of computation in Gaussian process regression.

\subsection{Classical computation and complexity}
The typical classical implementation of GPR is based on computing the Cholesky decomposition of $(K+\sigma_n^2{I})$. This amounts to finding the Cholesky factor, the lower-triangular matrix $L$ that satisfies $(K+\sigma_n^2{I})=LL^T$. Computing the Cholesky factor has a cost proportional to $n^3$, and it is numerically stable. The mean predictor can be expressed as $\bar{{f}_*}=\textbf{k}_*^T\boldsymbol{\alpha}$ by defining $\boldsymbol{\alpha}=(K+\sigma_n^2{I})^{-1}\textbf{y}$. The vector $\boldsymbol{\alpha}$ is then obtained by solving $LL^T\boldsymbol{\alpha}=\textbf{y}$.
Let $\textbf{y}^\prime = L\backslash \textbf{y}$ denote the solution to the 
triangular linear system $L\textbf{y}^{\prime}=\textbf{y}$. The vector $\boldsymbol{\alpha}$ can then be rewritten as $\boldsymbol{\alpha}=L^T\backslash L\backslash\textbf{y}$, hence computing $\boldsymbol{\alpha}$ simply amounts to solving two triangular systems. This has a runtime which scales as $\mathcal{O}(n^2)$.
Similarly, the variance $\mathbb{V}[{f}_*]$ can be expressed in terms of the Cholesky factor as $\mathbb{V}[{f}_*]=k(\mathbf{x}_*,\mathbf{x}_*)-(L\backslash\textbf{k}_*)^T(L\backslash\textbf{k}_*)$. Hence it also has a 
$\mathcal{O}(n^2)$ runtime. 
Thus the overall runtime of classically computing the mean predictor and the associated variance for a GP model scales as $\mathcal{O}(n^3)$.  When dealing with large-scale problems with greater than $10^3$ input points, exact inference with GPR is practically intractable.

\section{Training Gaussian processes}
\label{sec: TGP}

Model selection refers to the process of choosing the preferred variations of the model used in a supervised learning task, to achieve better predictive performance. 
In the context of GPs, this amounts to selecting a covariance function. In practice, a family of functions is usually considered. The parameters of the family of kernels are referred to as the kernel hyperparameters, and a range of optimisers are used in order to tune these hyperparameters based on the observed data. This model selection process is commonly known as the training of a Gaussian process.
Since training typically involves repeated evaluation of certain cost functions that characterise how well a given model is performing on the problem, it generally carries a runtime overhead that scales polynomially with the input size. In this section, 
we will follow the conventions of Chapter 5 of \cite{Rasmussen2004} and review the basics of training a GP model. Our emphasis is on introducing the log marginal likelihood ($\mathrm{LML}$) as a measure for the model's suitability and the classical computation of the $\mathrm{LML}$ function. 

\subsection{Log marginal likelihood}

The natural figure or merit that measures the performance of a supervised machine learning model is the marginal likelihood. In the context of GPR, it is the probability density of the observed output vector conditioned on the model's covariance matrix and the Gaussian noise variance, $p(\mathbf{y}|K+\sigma_n^2 I)$. As such, training the GP model amounts to optimising the conditioned probability of the observed data given the GP prior by choosing the covariance function and tuning the respective hyperparameters. 

For simplicity, we will keep the assumption that the model has a zero prior mean. Since the prior distribution of the observed vector of outputs $\mathbf{y}$ only differs from that of the latent function $f(\mathbf{x})$ by a Gaussian noise with variance $\sigma_n^2$, it is clear that we can write down the distribution as $\mathbf{y}\sim\mathcal{N}(0,K+\sigma_n^2{I})$. The logarithm of marginal likelihood $\mathrm{LML}=\log[p(\mathbf{y}|K+\sigma_n^2{I})]$ then follows straight-forwardly from the definition of Gaussian distribution, and we have
\begin{align}
\mathrm{LML} =-\frac{1}{2}\mathbf{y}^T(K+\sigma_n^2\mathit{I})^{-1}\mathbf{y}
-\frac{1}{2}\log\det[K+\sigma^2_n\mathit{I}] -\frac{n}{2}\log 2\pi.\label{eq:lml}
\end{align}
Since the logarithm is monotonic, maximising $\mathrm{LML}$ is equivalent to directly maximising $p(\mathbf{y}|K+\sigma_n^2{I})$. 

\paragraph{Interpretations}
Note that only the first term of Eq. \ref{eq:lml} involves the observed outputs $\mathbf{y}$. This is the contribution to $\mathrm{LML}$ that actually measures how the model is performing at fitting the training data. The second term depends only on the covariance matrix with the identity noise entry and can be interpreted as a penalty on the model's complexity. It generally disfavours models that happen to overfit the training set. The last term is normalisation constant ensuring the probability is bounded by one. It is easily computable. Hence only the first two terms in $\mathrm{LML}$ involve extensive matrix computations, and could potentially present bottlenecks in the efficiency of training.

\paragraph{Hyperparameter optimisation}
Training requires tuning the model's hyperparameters in order to maximise the $\mathrm{LML}$. A standard approach is based on gradient descent methods. This requires evaluating the variation of $\mathrm{LML}$ with respect to a change in each hyperparameter $\theta_j$. We will come back to this point in the context of applying quantum algorithms in Chapter \ref{QGPT}.

\subsection{Implementations and complexity}

The runtime of classically computing $\mathrm{LML}$ is dominated by the matrix inversion and determinant computation. In standard implementations based on Cholesky decompositions, the runtime scales with the input size as $\mathcal{O}(n^3)$. Because of the high computation cost in exact implementations, numerous compromising approaches have been proposed in the machine learning community. For examples, GPs are sometimes chosen to have covariance matrices with fixed ranks to make them computational trackable. In such scenarios, the cost of training can be reduced to $\mathcal{O}(nr^2)$, with $r$ denoting the rank of the covariance matrix in the model \cite{quinonero2005unifying}. This, however, significantly limits the range and the complexity of the functions accessible to the GP model, which could translate into sub-optimal performance. In low-dimensional cases, approaches such as hierarchical matrix factorisation \cite{minden2016fast} provide good options for implementing GP training, but they do not generalise well to problems with high dimensional datasets, which are often essential to consider in machine learning.

\subsection{Stochastic trace estimation}

As an alternative approach, stochastic trace estimation has gained popularity in recent years \cite{pace2004chebyshev, boutsidis2015randomized}. 
These methods make use of the fact that given a matrix $A\in \mathbb{R}^{n\times n}$, the logarithm of its determinant is equal to the trace of the $\log(A)$, as we have
\begin{align}
\Tr[\log(A)] = \sum\limits_{i=1}^{n}\log\lambda_i  = \log[\det(A)],
\label{eqn: logdet}
\end{align}
where $\{\lambda_i \}$ are the eigenvalues of $A$.

The matrix logarithm in Eq. \ref{eqn: logdet} can then approximated by truncating the Taylor series of the logarithmic function,
\begin{equation}
\centering
\log(A) \approx \sum_{a=1}^{d} \frac{(I-A)^a}{a}.
\end{equation}
Alternatively it can be approximated with a Chebyshev polynomial of a specified degree $d$. Using a trace estimation approach will still require matrix-vector multiplication when raising factors such as $A$, or $(I - A)$. However the advantage arises as the inner product form $\mathbf{z}^\dagger\log(A)\mathbf{z}$ can be evaluated in $\mathcal{O}(n^2)$ for some $\mathbf{z} \in \mathbb{R}^{n}$, where the vector $\mathbf{z}$ is a so called `probing vector'. These vectors can be chosen such in a number of ways \cite{avron2011randomized, MUB}, and they should satisfy $\mathbb{E}[\mathbf{z}^\dagger\log(A)\mathbf{z}] = \Tr(\log(A))$. Note that there are two major sources of error that can occur in such an approach, namely the errors due approximating $\log(A)$ with a finite expansion, and the errors directly related to the stochastic trace estimation. 
 We draw special interests to these stochastic trace estimation methods as the approach based on quantum algorithms to be presented in Chapter \ref{QGPT} can be understood as an extension of this class of trace estimation algorithms. As we will show, besides offering a reduction in computational time, the quantum algorithms also use an exact representation of $\log(K+\sigma_n^2I)$ to machine precision, which implies a significant suppression in approximation error.
 
\section{Connection with deep learning}
\label{sec: gpdl}
In this section, we briefly review the connection between Gaussian processes and deep neural network models based on the results of Ref. \cite{lee2017deep}, which provides a Bayesian approach to deep learning. We will leverage this connection to construct a quantum algorithm for Bayesian deep learning in Chapter \ref{Chapter: QBDL}.

\paragraph{Single hidden layer} The correspondence between Gaussian processes and a neural network with only a single hidden layer is well-known and discussed \cite{neal1994priors}. 
Let $\mathbf{z}(\mathbf{x})\in \mathbb{R}^{d_{out}}$ denote the output vector of a neural network with an input vector $\mathbf{x}\in \mathbb{R}^{d_{in}}$, with $z_i(\mathbf{x})$ denoting the $i^{th}$ component of the output layer. 
If we assume the weight and bias parameters of the neural network are independent and identically distributed (i.i.d.), each $z_i$ will be a sum of i.i.d terms.
As such, if the hidden layer has an infinite width, the Central Limit Theorem implies that $z_i$ follows a Gaussian distribution.
Now consider a set of $n$ input points, with corresponding outputs $\{z_i(\mathbf{x}_{1}), z_i(\mathbf{x}_{2}),\ldots z_i(\mathbf{x}_{n})\}$. Any finite collection of this output set will follow a joint multivariate Gaussian distribution. By definition, $z_i$ corresponds to a GP, $z_i\sim \mathcal{GP}(\mu,K)$, with a covariance matrix $K(\mathbf{x},\mathbf{x}^\prime)=\mathbb{E}[z_i(\mathbf{x})z_i(\mathbf{x}^\prime)]$.
Conventionally, the weight and bias parameters are chosen to have zero mean so that $\mu=0$. 
     
\paragraph{Deep networks} 
The above correspondence between the GP and the single hidden layer network is generalised to a deep neural network architecture in a recursive manner \cite{lee2017deep, gmatthews2018gaussian}.
Let $z_i^l$ denote the $i^{th}$ component of the output of the $l^{th}$ layer. By induction, it follows that $z_i^l\sim\mathcal{GP}(0,K^l)$.
The covariance matrix on the $l^{th}$ layer is given by $K^l(\mathbf{x}, \mathbf{x}^\prime)=\mathbb{E}[z_i^l(\mathbf{x})z_i^l(\mathbf{x}^\prime)]$.
In order to explicitly compute $K^l(\mathbf{x},\mathbf{x}^\prime)$, we need to specify the Gaussian variance on the weight and bias parameters, $\sigma_w^2$ and $\sigma_b^2$, as well as the non-linear activation functions, $\phi$ at each layer.
We have the following recursive formula for the $l^{th}$ layer covariance function, 
\begin{align}
K^l(\mathbf{x},\mathbf{x}^\prime)=\sigma_b^2+\sigma_w^2\mathbb{E}[\phi(z_i^{l-1}(\mathbf{x}))\phi(z_i^{l-1}(\mathbf{x}^\prime))],    
\end{align}
where $z_i^{l-1}\sim\mathcal{GP}(0,K^{l-1})$. The base case of the induction is given by the layer zero covariance function, 
\begin{align}
K^0(\mathbf{x},\mathbf{x}^\prime)=\sigma_b^2+\sigma_w^2\left(\frac{\mathbf{x}.\mathbf{x}^\prime}{d_{in}}\right).    
\end{align}
The Bayesian training of the neural network amounts to computing the mean and variance of the predictive distribution, while 
selecting the GP covariance function and tuning the hyper-parameters is related to choosing the neural network model class, depth, nonlinearity and parameter initialisations.     
Numerical experiments suggest that neural networks with infinite-width hidden layers trained with Gaussian priors outperform finite-width neural networks trained with stochastic gradient descent in many cases \cite{lee2017deep}.

\chapter{Quantum enhanced Gaussian processes} 

\label{QGP} 



Having reviewed the basics of Gaussian processes as classical regression models in the previous chapter, now we move on to present a quantum algorithms for enhancing the efficiency of computing GPR. 
We will start by describing a quantum state preparation procedure that encodes a classical input vector into a quantum state. Quantum state preparation would be important not only for the quantum Gaussian processes algorithm but more generally for all machine learning applications where one desires to use a quantum computer to analyse classical datasets. We will then describe the procedure for the quantum Gaussian process algorithm, followed by a discussion of practicality, and potential caveats in applying the proposed quantum algorithm. The material of this chapter is based on the work of Ref. \cite{zhao2018note} and \cite{zhao2015quantum}. 

\newpage
\section{State preparation}
\label{State_prep}
When applying quantum computation to problems with classical input, it is almost always necessary to prepare quantum states that encode the classical input vectors. For instance, in Chapter \ref{IntroQA} and \ref{QDLSA} we have seen that the quantum linear system problem requires an input quantum state that encodes the known vector in the corresponding classical linear system. 
\subsection{Quantum random access memory}
We are specifically concerned with the task of state preparation which involves creating
\begin{equation}
    \ket{\mathbf{v}} = \|\mathbf{v}\|_2^{-1}\sum_{i=1}^n v_i \ket{i},
\end{equation}
given some vector $\mathbf{v} \in \mathbb{R}^n$
stored in quantum random access memory (QRAM) \cite{GLM08a}. 
Such a memory structure allows the quantum computer to access data stored in multiple memory locations in a quantum superposition. That is, it allows for operations of the following type,
\begin{equation}
    \sum_{i,j} \alpha_{ij} \ket{i}\ket{j} \xrightarrow{\text{QRAM}} \sum_{i,j} \alpha_{ij} \ket{i}\ket{j + m_i},
\end{equation}
where $m_i$ denotes the $i$th entry stored in memory. 
As such, QRAM enables probabilistically producing $\ket{\mathbf{v}}$ for any $\mathbf{v}$ stored in memory. 

As a general procedure, to create $\ket{\mathbf{v}}$ for any $\mathbf{v}$, we start with an initial query state, $n^{-\frac{1}{2}}\sum_i \ket{i}\ket{0}$, and then use the QRAM to map the query state into $n^{-\frac{1}{2}} \sum_i \ket{i}\ket{v_i}$.
We then append a register with ancillary qubits prepared in state $\ket{0}$ and rotated conditioned on the value in second register, which leads to the state  
\begin{align}
 n^{-\frac{1}{2}} \sum_i \ket{i}\ket{v_i}\left(\sqrt{1-|v_i|^2} \ket{0} + v_i \ket{1}\right),     \label{QRAMCROT}
\end{align}
where we have assumed for simplicity that the vector $\mathbf{v}$ is normalised such that $|v_i|\leq 1$. Next, perform a second QRAM call to reverse the computation of $\ket{v_i}$. Finally, post-selecting on the state $\ket{1}$ leads to the desired state that encodes the classical vector, $\ket{\mathbf{v}}$.

\paragraph{Success probability}
The probability of projecting onto the desirable subspace in the final step is given by $n^{-1} \sum_i |v_i|^2$. In the case where the entries of $\mathbf{v}$ are of similar magnitude, $\ket{\mathbf{v}}$ can be prepared using only a constant number of queries. 
However, a potential caveat arises when a small number of entries in the vector are significantly larger than the others, in such cases, projecting on the correct state requires $\Omega\left(\sqrt{n}\right)$ QRAM queries \cite{soklakov2006efficient}, this can be seen as a consequence of the lower bounds on unordered search \cite{boyer1996tight}. 
The same issue would persist when it is only required to prepare an approximate vector $\ket{\mathbf{v'}}$ that satisfies $|\mathbf{v'} - \mathbf{v}|_2\le \epsilon$, where $\epsilon$ is a sufficiently small constant error.

\subsection{Robustness and rounding conventions}

Fortunately, data processing tasks in practical machine learning almost always implicitly assume a certain level of robustness against small perturbation in the $\infty$-norm which measures only the largest entry-wise error. 
In particular, any digital data processing based on fixed or floating-point arithmetic only makes sense if the outcome of the analysis remains valid if the features in the input vector deviate from the original values below the machine precision. Due to the sheer nature of measurements in the real world, it is practically reasonable to assume the data points are specified with finite precision. Hence the appropriate error constraint which reflects the realistic analytic scenarios is only that $|\mathbf{v'} - \mathbf{v}|_\infty\leq \epsilon$, instead of requiring a close approximation in the 2-norm. 

\paragraph{Alternative rounding} 
Assuming the data processing inherently have tolerance against an $\epsilon$ perturbations in $\infty$-norm allows us to work with the vector with entries $v'_i$ which are half-integer multiples of the base precision $\epsilon$. In this alternative numerical rounding convention (as shown in Figure \ref{fig:numbering}), $\mathbf{v'}$ is chosen to be the closest representable vector to $\mathbf{v}$, which satisfies $|\mathbf{v'} - \mathbf{v}|_\infty\leq \frac{\epsilon}{2}$, and the distance from the original value of the data is less than $\epsilon$. 
    Note that this offset rounding does not contain an exact representation of $0$.  This new convention can be either directly realised in the loading stage of the QRAM, or equivalently implemented at the controlled rotation stage, as shown in Eq. \ref{QRAMCROT}.
\begin{figure}[h]
    \includegraphics[width = 0.8 \linewidth]{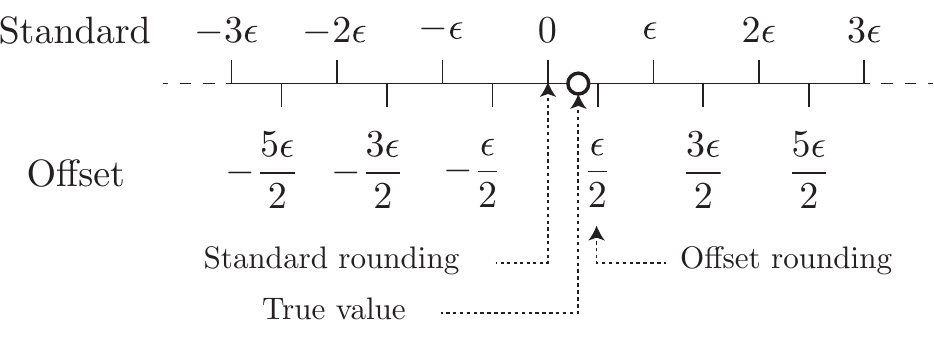}
    \caption{Numerical rounding conventions. In the standard rounding convention, scalar values are rounded to the nearest integer multiple of precision $\epsilon$. Alternatively, we can consider an offset rounding convention, where the rounding is to the nearest half-integer multiple of $\epsilon$. In either scheme, the rounded value is always within $\frac{\epsilon}{2}$ of the true value. \label{fig:numbering}}
\end{figure}
In some cases, always using a positive sign offset ($+\epsilon/2$) to data-points will introduce an undesirable systematic error in the loaded vector. To overcome this potential issue, one can choose to implement a nearly white noise offset. This can be achieved by either utilising a suitable pseudo-random number generator which is seeded by the corresponding memory location, or by including random data stored in other locations of the QRAM.

The robustness requirement against small perturbation in the $\infty$- norm guarantees the overall analysis is insensitive to using the above-described offset rounding convention.
Furthermore, note that the success probability of the final projection step is lower bounded by $\frac{\epsilon^2}{4}$. Hence preparing the quantum state that encodes $\mathbf v$ can succeed independent of the dimensionality, $n$. This is due to the absence of an exactly representable of $0$ in the offset rounding convention. This success probability in state preparation can further be enhanced to $\Omega(\epsilon)$ with the technique of fixed-point quantum amplitude amplification described in Ref. \cite{yoder2014fixed}. 
Note that the base precision parameter $\epsilon$ need not be on the order of machine precision. Any values of $\epsilon$ which is small compared with the known accuracy level of the input data will ensure that the final error is negligible. Generally speaking, low precision data will have a constantly more efficient loading procedure then high precision data.
Most importantly, the number of necessary QRAM queries for successful state preparation procedure will always be upper bounded by the inverse of a constant precision parameter which is independent of the size of the database.

In summary, efficient quantum state preparation to encode a classical input vector is possible in any data processing application which is robust under small $\infty$-norm perturbations.
 As a consequence, the caveat related to state preparation highlighted by Aaronson in Ref. \cite{Aaronson2015b} can generally be overcome in the context of machine learning, due to the inherent robustness assumption.
 However, this robustness feature not necessarily shared by other application such as computational physics or numerical mathematics where exact vector entries representations could potentially be hard requirements of any meaningful analysis.  As one important example of robust applications of quantum machine learning, Gaussian processes are the main topic of this part of the thesis. We will explicitly introduce a state preparation procedure in the next section.

\subsection{State preparation for GPR}

In order to adapt the QRAM based state preparation scheme to Gaussian processes applications, we need to modify it to prepare a state corresponding to the $s_\mathbf{v}$-sparse vector $\mathbf{v}$ with entries $v_i$.
We start with a register prepared in a superposition
\begin{align}
s_\mathbf{v}^{-1/2}\sum_{i: v_i \neq 0} \ket{i} \otimes \ket{0}.
\end{align}
Then we use the index stored in the first register, to conditionally rotate the ancillary register, so that the rotation is based on the $i$th non-zero entry of $\mathbf{v}$. The resultant state of the system is 
\begin{align}
\ket{\tilde{\mathbf{v}}} = \frac{1}{\sqrt{s_\mathbf{v}}}\sum_{i:v_i\neq 0}\ket{i}\otimes \left(\sqrt{1-c_\mathbf{v}^2v_i^2}\ket{0} + c_\mathbf{v} v_i\ket{1}\right),
\end{align}
where $c_\mathbf{v} \leq \min_i |v_i|^{-1}$ is the chosen constant to normalise the unitary rotation. Finally, post-selecting on the ancillary register being in state $\ket{1}$ projects the first register to the required state $\ket{\mathbf{v}} = \frac{\mathbf{v}}{||\mathbf{v}||}$. In rare cases, the vector could be vastly dominated by a handful of large value entries, the previously described offset rounding convention can then be applied to ensure a constant success probability in preparing the quantum state for Gaussian processes.

\section{Quantum Gaussian process algorithm}

The essential idea of applying quantum algorithms to GPR comes from the observation that the computation of the central quantities of interest in GPR, ${f}_*$ and $\mathbb{V}[{f}_*]$, as written in Eq. \ref{eq:mean} and Eq. \ref{eq:variance}, involves solving linear systems of the forms $(K+\sigma_n^2{I})\boldsymbol{\alpha}=\mathbf{y}$ and $(K+\sigma_n^2{I})\boldsymbol{\eta}=\mathbf{k}_*$ respectively, where $\mathbf{k}_*^T\boldsymbol{\alpha}={\bar{f}}_*$ and $k\left(\textbf{x}_*,\textbf{x}_*\right)-\mathbf{k}_*^T\boldsymbol{\eta} =\mathbb{V}[{f}_*]$. The common linear structure suggests that we can apply the quantum linear system algorithm to extract useful information.

\subsection{Inner product estimation}

As a prerequisite component to the quantum Gaussian process algorithm we here introduce a mechanism to estimate the inner product $\langle \mathbf{u} | \mathbf{v} \rangle$ for a given pair of real vectors $\mathbf{u}$ and $\mathbf{v}$.
Although the squared version, $|\langle \mathbf{u} | \mathbf{v} \rangle|^2$, can be easily computed using a controlled-swap test, as presented in Ref. \cite{Liming}, for the purpose of GPs we need to compute both the magnitude and sign of this inner product. Since the controlled-swap test gives the result estimate in terms of a probability, the sign of $\langle \mathbf{u} | \mathbf{v} \rangle$ is not directly accessible. Thus in order to estimate the inner product, we instead use an augmented version of the state preparation technique, in which an additional ancillary qubit is introduced to determine whether the target state is $\ket{\mathbf{u}}$ or $\ket{\mathbf{v}}$. Specifically we initialise the ancillary qubit in the state
\begin{align}
\ket{+}=\frac{1}{\sqrt{2}}(\ket{0}+\ket{1}), 
\end{align}
which results in a joint state, 
 \begin{align}
 \ket{\Phi_{\mathbf{u},\mathbf{v}}} =& \frac{1}{\sqrt{2s_\mathbf{u}}}\sum_{i:u_i \neq 0} \ket{0}\ket{i} \left(\sqrt{1-c_\mathbf{u}^2 u_i^2}\ket{0} + c_\mathbf{u}u_i \ket{1}\right)\nonumber\\
 &+ \frac{1}{\sqrt{2s_\mathbf{v}}}\sum_{i:v_i \neq 0} \ket{1} \ket{i} \left(\sqrt{1-c_\mathbf{v}^2 v_i^2}\ket{0} + c_\mathbf{v}v_i \ket{1}\right). 
\end{align}
Then measuring the operator $M = X \otimes I \otimes \ket{1}\bra{1}$ 
results in an expectation value
\begin{align}
\langle M\rangle=s_\mathbf{u}^{-1/2}s_\mathbf{v}^{-1/2}c_\mathbf{u} c_\mathbf{v} \mathbf{u}^T \mathbf{v}.
\end{align}

\subsection{Procedures}

Now we are in a position to introduce a quantum algorithm for computing the quantities of the form $\mathbf{u}^T A^{-1} \mathbf{v}$, which can, in turn, be applied to compute the central quantities of GP regression. To do so, we combine the techniques of state preparation, inner product estimation together with the quantum linear system algorithm (QLSA) described in Chapters \ref{IntroQA} and \ref{QDLSA}. The general procedure is as follows:

\begin{enumerate}
\item Initialise the system in the state $\ket{+}_A\ket{0}_B\ket{0}_C\ket{0}_D$, where the subscripts $A$, $B$, $C$ and $D$ label different registers.
\item Conditioned on register $A$ being in state $\ket{0}$, query the QRAM and prepare registers $B$ and $C$ in the state $\ket{\tilde{\mathbf{u}}}$, such that the ancillary qubit is placed in register $C$ with the rest of the state in register $B$, and apply an $X$ gate to register $D$.
\item Conditioned on register $A$ being in state $\ket{1}$, query the QRAM and prepare registers $B$ and $C$ in the state $\ket{\tilde{\mathbf{v}}}$ such that the ancillary qubit is placed in register $C$ with the rest of the state in register $B$. 
\item Conditioned on both registers $A$ and $C$ being in state $\ket{1}$, apply QLSA using $B$ as the input register and using $D$ as the ancillary register. A fifth register $E$ is introduced for the phase estimation subroutine in the QLSA, but since it is eventually uncomputed and returned to the zero state, we will omit it in the description of the states after each step for simplicity.
\item Measure the system with the observable $M=X_A I_B \ket{1}\bra{1}_C \ket{1}\bra{1}_D$.
\end{enumerate}
The measurement result will be a random variable with an expectation value, 
\begin{align}
\langle M\rangle=    c s_\mathbf{u}^{-1/2}s_\mathbf{v}^{-1/2} c_\mathbf{u} c_\mathbf{v} \mathbf{u}^T A^{-1} \mathbf{v}.
\end{align}
A circuit diagram describing the above procedures for computing $\mathbf{u}^T A^{-1} \mathbf{v}$ is shown in Figure \ref{fig: QGPR}.

\begin{figure}[H] 
\centering
\includegraphics[width=\textwidth] {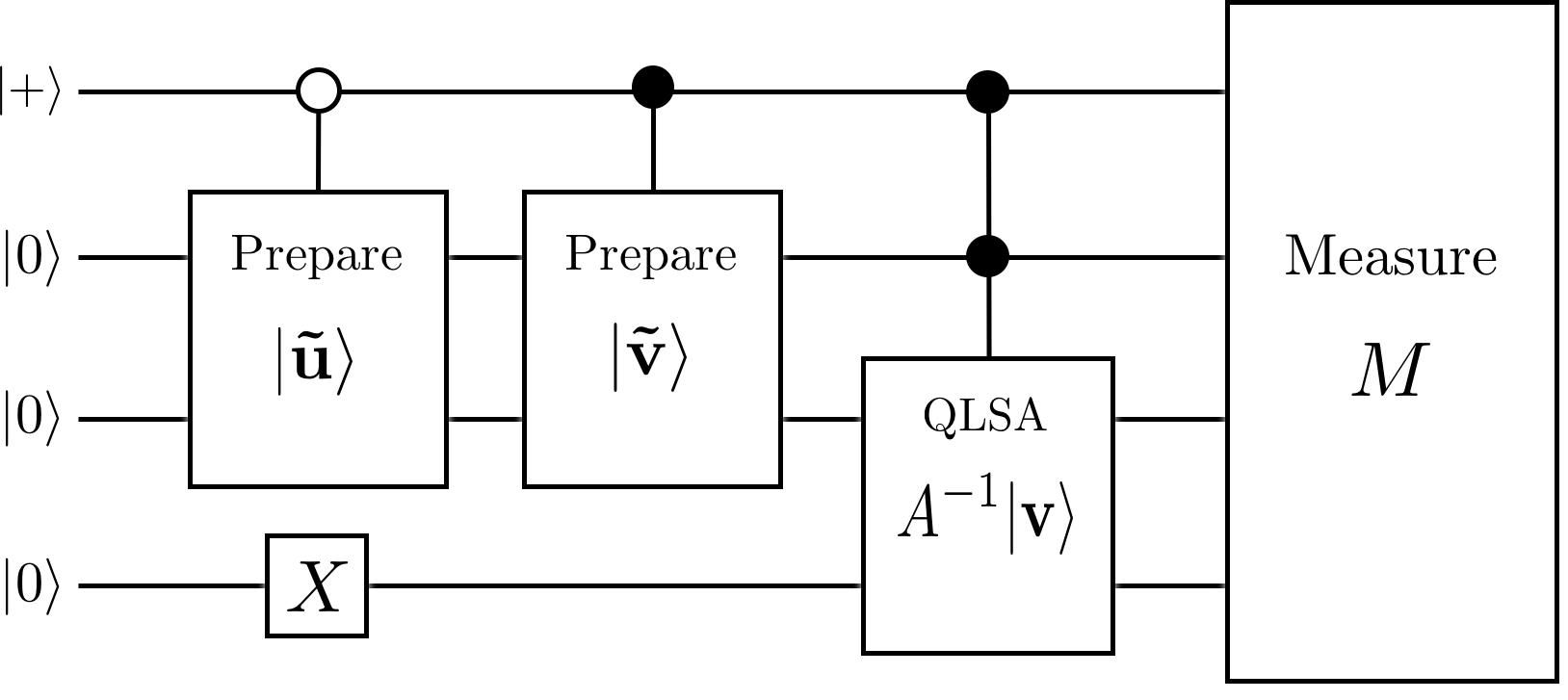} 
\caption{Circuit diagram for computing the form $\mathbf{u}^T A^{-1} \mathbf{v}$, where $M=X\otimes I\otimes \ket{1}\bra{1}\otimes \ket{1}\bra{1}$.}
\label{fig: QGPR} 
\end{figure}

\paragraph{Derivations} To see the validity of the above algorithm, note that the state of the system after Step 4 is given by
\begin{align}
&\frac{1}{\sqrt{2s_\mathbf{u}}}\ket{0}_A\sum_{i:u_i\neq 0} \ket{i}_B \left(\sqrt{1-c_\mathbf{u}^2u_i^2}\ket{0}_C + c_\mathbf{u} u_i\ket{1}_C\right)\ket{1}_D\nonumber\\
&+\frac{1}{\sqrt{2s_\mathbf{v}}} \ket{1}_A\sum_{i:v_i\neq 0} c_\mathbf{v} \beta_i \ket{\mu_i}_B \ket{1}_C \left(\sqrt{1-\frac{c^2}{\lambda_i^2}}\ket{0}_D + \frac{c}{\lambda_i}\ket{1}_D\right)\nonumber \\
&+\frac{1}{\sqrt{2s_\mathbf{v}}} \ket{1}_A\sum_{i:v_i\neq 0}  \sqrt{1-c_\mathbf{v}^2v_i^2}  \ket{i}_B \ket{0}_C \ket{0}_D,
\end{align}
where $\ket{\mu_i}$ denotes the $i$th eigenvector of $A$ with corresponding  eigenvalue $\lambda_i$, and $\{\beta_i\}$ denotes the coordinates of $\mathbf{v}$ in the basis of $\{\ket{\mu_i}\}$. The subsequent projection of this state onto $\ket{1}$ for registers $C$ and $D$ results the sub-normalised state
\begin{align}
\frac{1}{\sqrt{2s_\mathbf{u}}}\ket{0}_A\sum_{i:u_i \neq 0} c_\mathbf{u} \gamma_i\ket{\mu_i}_B+\frac{1}{\sqrt{2s_\mathbf{v}}}\ket{1}_A\sum_{i:v_i \neq 0}^n \frac{c}{\lambda_i} c_\mathbf{v} \beta_i \ket{\mu_i}_B,
\end{align}
where $\{\gamma_i\}$ are the coordinates of $\mathbf{u}$ in the basis of $\{\ket{\mu_i}\}$.
As a result, the expectation value of the final measurement is given by
\begin{align}
&\sum_i \frac{1}{4}\left(\left(\frac{c_\mathbf{u}}{\sqrt{s_\mathbf{u}}} \gamma_i + \frac{c_\mathbf{v}c}{\sqrt{s_\mathbf{v}}} \frac{\beta_i}{\lambda_i}\right)^2-\left(\frac{c_\mathbf{u}}{\sqrt{s_\mathbf{u}}} \gamma_i - \frac{c_\mathbf{v}c}{\sqrt{s_\mathbf{v}}} \frac{\beta_i}{\lambda_i}\right)^2\right)\nonumber\\
=&\frac{c_\mathbf{u}c_\mathbf{v}c}{\sqrt{s_\mathbf{u}s_\mathbf{v}}} \mathbf{u}^T A^{-1}\mathbf{v}.    
\end{align}
The expectation value for the measurement in the final step, $\langle M\rangle$, must match the above, thus we have 
\begin{align}
    \langle M\rangle=\frac{c_\mathbf{u}c_\mathbf{v}c}{\sqrt{s_\mathbf{u}s_\mathbf{v}}}\mathbf{u}^T A^{-1}\mathbf{v}.
\end{align}  
It should be noted the estimation of $\langle M\rangle$ in involves sampling $m$ on repeated runs of the algorithm, which results in a sampling variance that scales as $m^{-1}$.

The above-outlined algorithm for estimating the inner product form $\mathbf{u}^T A^{-1} \mathbf{v}$ can be used to construct a quantum algorithm for approximating both the mean predictor and variance predictor in computing GP regression, which we will illustrate in the following.

\subsection{Mean predictor}
\label{sub: mean predictor}
In order to approximate the mean predictor, 
$\mathbf{k}_*^T (K+\sigma_n^2 I)^{-1}\mathbf{y}=\mathbf{y}^T (K+\sigma_n^2 I)^{-1}\mathbf{k_*}$, 
we set $\mathbf{u} = \mathbf{y}$, $A = K+\sigma_n^2I$ and $\mathbf{v}=\mathbf{k_*}$. 
Since $K$ is positive semi-definite, the minimum eigenvalue of $A$ is lower bounded by $\sigma_n^2$, and hence we take the normalisation constant $c=\sigma_n^2$ in each run of the QLSA. 
This leads to
\begin{align}
\langle M\rangle = \frac{\sigma_n^2 c_{\mathbf{k}_*}c_\mathbf{y}}{\sqrt{s_{\mathbf{k}_*}s_\mathbf{y}}} \mathbf{y}^T (K+\sigma_n^2 I)^{-1}\mathbf{k_*},    
\end{align}
and therefore 
\begin{equation}
\bar{f}_* = \frac{\sqrt{s_{\mathbf{k}_*}s_\mathbf{y}}}{\sigma_n^2 c_{\mathbf{k}_*} c_\mathbf{y}}\langle M \rangle.
\end{equation}
Here $c_{\mathbf{k}_*}$ and $c_\mathbf{y}$ are taken to be the inverted maximum absolute values of the entries in $\mathbf{k}_*$ and $\mathbf{y}$ respectively, which we treat as constants.
Hence the variance in estimating the value of $\bar{f}_*$ will scale as $s_{\mathbf{k}_*}s_\mathbf{y}m^{-1}$. 
In the case of $K$ being $s$-sparse, we have $s_\mathbf{k_*}\le s$ since $\mathbf{k_*}$ reflects the same dependencies as $K$. 
While $\mathbf{y}$ will not, in general, be sparse, we can instead replace it in the estimation procedure with a vector $\mathbf{y'}$ with a small number of non-zero entries and still obtain a good approximation to $\bar{f}_*$, whenever the spectral norm of $K+\sigma_n^2 I$ is bounded, which will virtually always be the case for GP regression. 
This is because of the fact that
\begin{align}
(K+\sigma_n^2 I)^{-1} = \sum_d (-1)^d(K+(\sigma_n^2-1) I)^d,    
\end{align}
and hence that $(K+\sigma_n^2 I)^{-1}$ can be approximated by a polynomial in $(K+(\sigma_n^2-1) I)$ of some fixed degree, which will result in a matrix of constant sparsity. 
Hence $(K+\sigma_n^2 I)^{-1}\mathbf{k}_*$ will be an approximately sparse vector, and its inner product with $\mathbf{y}$ can be well approximated by the inner product with a vector $\mathbf{y'}$ where the only non-zero entries correspond to the location of non-negligible entries of $(K+\sigma_n^2 I)^{-1}\mathbf{k}_*$. 
In conclusion, only a constant number of repetitions of the algorithm is needed to achieve a fixed variance of estimation.

\subsection{Variance predictor}
In order to approximate the variance $\mathbb{V}[{f}_*]$, 
the same procedure is followed as for the mean predictor, except that $\mathbf{u}$ is now taken to be $\mathbf{k}_*$ instead of $\mathbf{y}$. This yields 
\begin{align}
\langle M\rangle = \frac{\sigma_n^2 c_{\mathbf{k}_*}^2}{s_{\mathbf{k}_*}} \mathbf{k}_*^T (K+\sigma_n^2 I)^{-1}\mathbf{k}_*,     
\end{align}
and therefore we have
\begin{equation}
\mathbb{V}[{f}_*] = k(\textbf{x}_*,\textbf{x}_*)-\frac{s_{\mathbf{k}_*}}{\sigma_n^2 c_{\mathbf{k}_*}^2} \langle M \rangle.
\end{equation}
As with the mean predictor in Section \ref{sub: mean predictor}, $\langle M \rangle$ needs to be measured on a constant number of independent runs of the algorithm in order to yield a desired fixed variance on the estimate.

\section{Discussions}

We have shown that the QLSA introduced in Section \ref{HHL} can be applied to evaluating the two central objective quantities in GPR problems, the mean predictor, and the variance predictor. Inherited from the computational time of QLSA, this quantum GPR procedure achieves an exponential speed-up over classical implementations under two assumptions about the covariance matrix, $(K+\sigma_n^2{I})$, namely, the matrix is sparse and well-conditioned. We discuss the practicalities of these assumptions.

\paragraph{Sparsely constructed GP}
GPs with sparse covariance matrices are of significant interests in many real-world applications, particularly when the problem involves inference from large datasets \cite{sparse2009}. For example, these sparsely constructed Gaussian processes are used to make a unified framework for robotic mapping \cite{robotic}. In the field of pattern recognition, sparsely constructed Gaussian processes have been used to solve realistic action recognition problems \cite{recognition}. A widely used technique to construct a sparse covariance matrix is setting the covariance function to zero beyond a certain distance between any two data points with a compactly supported function. This is known as covariance tapering and has been proven to approximate the Mat\'ern family of covariance functions with a small squared error \cite{furrer2006covariance}. An explicit example in geostatistics kriging where the dataset gives rise to a highly sparse covariance matrix is presented in Ref. \cite{barry1997kriging}. In the above cases where the GPR computation only involves sparse covariance matrices, our proposed algorithm circumvents the major potential caveats of QLSA, and an exponential advantage over its classical counter-part is attainable. For other applications where $s$ scales linearly with $n$, our algorithm provides a polynomial speed-up over the best-known classical GPR algorithm, even though an exponential speed-up is not always guaranteed.

\paragraph{Conditioning}
To implement quantum GPR efficiently, the matrix $(K+\sigma_n^2{I})$ needs to be well-conditioned. The ratio of largest and smallest eigenvalue $\kappa$ needs to stay low as $n$ increases for the matrix to be robustly invertible. In classical GPR, conditioning is already a well-recognised issue. A general strategy to cope with the problem is to increase the noise variance $\sigma_n^2{I}$ manually by a certain amount to dilute the ratio $\kappa$ without severely affecting the statistical properties of the model. This increase in $\sigma_n^2{I}$ can be seen as a small amount of noise (jitter) in the input signal. This technique is not new to quantum GPR and may be seen throughout the classical GP literature and mainstream implementations \cite{gaussian1998}. Therefore, for almost all practical purposes, we can assume the matrix is well-conditioned before applying the quantum algorithm. Moreover, when we apply our algorithm on a sparse kernel, the preconditioning method presented in Ref. \cite{Clader2013a} can be applied to suppress the growth of $\kappa$ further. In fact, under the realistic assumption that the maximum entry of a sparse $K$ is bounded by a constant, the maximum eigenvalue of $(K+\sigma_n^2{I})$ must be bounded by a constant. This is a consequence of the Gershgorin circle theorem \cite{varga2010gervsgorin} which can be expressed in terms of the following inequality,
\begin{align}
|\lambda - A_{ii}|\le \sum_{j\neq i}|A_{ij}|.
\label{ineq: Gersh}
\end{align}
Note that since $A=(K+\sigma_n^2{I})$, the minimum eigenvalue of $A$ is lower bounded by $\sigma_n^2$. Likewise, we have the diagonal elements bounded by $A_{ii} \geq \sigma_n^2$ and the off-diagonal sum $\sum_{j\neq i}|A_{ij}|$ upper bounded by the sparsity of $A$ scaled by the magnitude of its maximum entry. Hence from Eq. \ref{ineq: Gersh} we deduce the maximum eigenvalue of $A$ is upper bounded by a constant that is independent of $n$. 
As a result, under the sparse and bounded element kernel matrix assumption, conditioning does not provide a barrier to our proposed quantum GPR algorithm.
In summary, we have argued that conditioning does not hinder the application of quantum GPR, and the algorithm is most advantageous when one is concerned with a sparse kernel. Under such circumstances, an exponential speed-up is achievable. Hence having addressed all the major potential caveats of QLSA \cite{Aaronson2015b}, the quantum GPR algorithm is shown to be a robust application with practical significance.

\chapter{Training quantum Gaussian processes}
\label{QGPT} 

As presented in the previous chapter, the quantum Gaussian process algorithm provides a speed-up in computing predictions and the associated variances given a fixed kernel. It is desirable to also have a correspondingly efficient quantum routine for kernel and hyperparameter selection. In particular, it would be desirable to evaluate a measure of the model's performance with a quantum routine that supplements the main learning algorithm. With this motivation, we propose a quantum approach to improve the efficiency of GP training based on evaluating the logarithm of marginal likelihood ($\mathrm{LML}$) of the Gaussian distribution of the observed data. The material of this chapter is based on the work of Ref. \cite{zhao2018quantum}.

\section{Quantum $\mathrm{LML}$ algorithm}
Here we introduce a quantum algorithm for estimating the $\mathrm{LML}$ given the kernel matrix of a Gaussian process, which serves as the standard metric for a kernel's performance on the given data set.
The complete estimation of $\mathrm{LML}$ is obtained by combining the ``penalty'' and the ``data fit'' terms. For the purpose of GP training, we are concerned with estimating the variation, $\delta\mathrm{LML}$, with respect to a training step, where the prefix $\delta$ denotes the variation in a quantity between training steps.
\subsection{Augmented linear algorithm}
The data fit term of the $\mathrm{LML}$ Eq. \ref{eq:lml}, $\frac{1}{2} \mathbf{y}^T(K+\sigma_n^2)^{-1}\mathbf{y}$ relates the outputs $\mathbf{y}$ to the covariance matrix $K$. Here we demonstrate a modified version of the QLSA \cite{Harrow2009a}, and show that it can be used to calculate the data fit term.
As discussed in Chapter \ref{IntroQA}, The QLSA makes use of the quantum phase estimation to obtain the superposition of the eigenvalues, $\lambda_i$ of $A\in \mathbb{R}^{n\times n}$ encoded in the form of binary bit-strings, where $A$ is the matrix in the linear system $A\ket{\mathbf{x}} = \ket{\mathbf{b}}$.
An ancillary qubit is then rotated conditioned on the values of $f(\lambda_i)$. In the case of the original linear system algorithm, the function $f$ is simply chosen to be $f(\lambda)=1/\lambda$. Post-selecting this ancillary qubit followed by the reversal of the phase estimation step results in finding $A^{-1}\ket{\mathbf{b}}$ with success probability $\bra{\mathbf{b}}(A^{-1})^\dagger A^{-1}\ket{\mathbf{b}}$. As noted in Ref. \cite{Harrow2009a}, the same method can be extended to obtain $f(A)\ket{\mathbf{b}}$ for any computable function $f$.

Here we apply an augmented version of the QLSA by choosing $f(\lambda) = \frac{1}{\sqrt{\lambda}}$ instead of the original inversion. The procedure for estimating the data fit term is given as follows:
\begin{enumerate}
    \item Use QRAM queries to prepare $\ket{\mathbf{y}}=\frac{\mathbf{y}}{\|\mathbf{y}\|}$ with the state preparation technique described in Section \ref{State_prep}.
    \item Set $\ket{\mathbf{b}}=\ket{\mathbf{y}}$ and $A = K+\sigma_n^2\mathit{I}$, and run the augmented QLSA with $f(\lambda) = \frac{1}{\sqrt{\lambda}}$, which leads to $A^{-\frac{1}{2}}\ket{\mathbf{y}}$ with success probability $\bra{\mathbf{y}}A^{-1}\ket{\mathbf{y}}$     
    \item Sampling on multiple runs of the augmented QLSA thus gives a Monte Carlo estimate of the data fit term with mean $\mathbf{y}^TA^{-1}\mathbf{y}$ and variance bounded by $\frac{1}{4}\|\mathbf{\mathbf{y}}\|^2\sigma_n^{-2}$.
\end{enumerate}
Note that on top of leading to the desired estimation for the data fit term, this choice of $f(\lambda)$ also reduces the inconvenient effect of poor conditioning by a square-root as the success probability of the measurement step is increased as $\sqrt \lambda \geq \lambda$ for all $0 \leq \lambda \leq 1$. When $A$ is well-conditioned and sparse, the runtime of sampling from such a distribution is logarithmic in the dimension of $\mathbf{y}$, inherited from the computational cost of QLSA in Ref. \cite{Harrow2009a}.

\subsection{Log determinant algorithm}

The second term of the $\mathrm{LML}$ in Eq. \ref{eq:lml}, $-\frac{1}{2}\log\det[K+\sigma_n^2 I]$ can be estimated via a quantum algorithm that samples the eigenvalues of a Hermitian matrix $A \in \mathbb R^{n\times n}$ uniformly at random. 
The algorithm proceeds as follows: 

\begin{enumerate}
\item Prepare $\log_2 n$ qubits in maximally mixed state, $\frac{1}{n}\sum\limits_{i=1}^{n}\ket{i}\bra{i}$, and store this in a first register. This can be achieved simply by preparing the register in a random computational basis state. Note that a maximally mixed state is maximally mixed in any basis, hence we can choose to represent the density matrix for the system in the eigenbasis $\{\ket{e_i}\}$ of a matrix $A=K+\sigma_n^2 I$: 
\begin{align}
\frac{1}{n}\sum\limits_{i=1}^{n}\ket{e_i}\bra{e_i}.
\end{align}    
\item Append a second register in a superposition state given by $\frac{1}{\sqrt{T}}\sum\limits_{\tau=1}^{T}\ket{\tau}$, so that the composite system is in the state
\begin{align}
\frac{1}{nT}\sum\limits_{\tau,\tau^\prime=1}^{T}\sum\limits_{i=1}^{n}\ket{e_i}\bra{e_i}\otimes\ket{\tau}\bra{\tau^\prime},
\end{align}
where the time period parameter $T$ is chosen to be a sufficiently large value in the same way as in Eq. \ref{eqn: step1HHL}.

\item Perform a Hamiltonian simulation and evolve the first register with the Hermitian matrix $(-A)$ for time specified by the second register. This is achieved by applying the conditional unitary evolution $\sum\limits_{\tau=1}^{T}\mathrm{e}^{iAt_0\tau/T}\otimes\ket{\tau}\bra{\tau}$, where $t_0=O(1/\epsilon)$ is chosen with respect to the $\epsilon$-bounded error required in the algorithm. We thus obtain
\begin{align}
\frac{1}{nT}\sum\limits_{\tau,\tau'=1}^{T}\sum\limits_{i=1}^{n}\mathrm{e}^{i\lambda_it_0(\tau-\tau')/T}\ket{e_i}\bra{e_i}\otimes\ket{\tau}\bra{\tau'}. 
\end{align}

\item Complete the phase estimation by performing a quantum Fourier transform of the second register. The resulting estimated eigenvalues of $A$, $\{\lambda_i\}$, are then stored in the second register as a binary bit-string up to a finite precision. This results in the system being in state,
\begin{align}
\frac{1}{n}\sum\limits_{i=1}^{n}\ket{e_i}\bra{e_i}\otimes\ket{\lambda_i}\bra{\lambda_i}.
\end{align}

\item Measure the second register in computational basis to obtain a random $\lambda_i$. By using the identity,
\begin{align}
\langle \log\lambda_i\rangle = \frac{1}{n}\sum\limits_{i=1}^{n}\log\lambda_i = \frac{1}{n}\Tr[\log(A)] = \frac{1}{n}\log[\det(A)],
\end{align}
The desired quantity $\log[\det(A)]$ is given then by $n\langle \log\lambda_i\rangle$, which will needs to be estimated by sampling eigenvalues of $A$ on repeated runs of the procedure.
\end{enumerate}
Hence the ``penalty'' term of the $\mathrm{LML}$ can be estimated using the above eigenvalue sampling procedure, by setting $A = K+\sigma_n^2\mathit{I}$. This procedure can be seen as a finite dimensional analogue of the continuous variable model proposed in Ref. \cite{nana}. 

\paragraph{Runtime}
The optimised phase estimation procedure \cite{clock,luis1996optimum} comes with an error, $\epsilon_{\lambda_i}$, which scales as $\mathcal{O}(1/t_0)$ in estimating each $\lambda_i$. This implies the error associated with the logarithm of a single eigenvalue scales as
$\epsilon=\left|\frac{d\log\lambda_i}{d\lambda_i}\epsilon_{\lambda_i}\right|=\mathcal{O}\left(\frac{1}{\lambda_it_0}\right)$.
Furthermore, in the context of GP training, there generally exists a $\sigma_n^2 \mathit{I}$ noise contribution to the covariance matrix, due to uncertainty in the observed data. Thus, in general, we have the minimum eigenvalue, $\lambda_{\text{min}}\ge\sigma_n^2$. Hence, the total bounded-error single-run of the algorithm takes time scaling logarithmically in $n$ as $t=\tilde{\mathcal{O}}\left(\frac{s\log n}{\sigma_n^2\epsilon}\right)$.


Due to the linear sparsity dependence from the Hamiltonian simulation step, the proposed quantum algorithm performs best when the covariance matrix is some constant $s$-sparse, in which case our algorithm provides an exponential speed-up over the classical GP training procedure. Such sparsely constructed GPs have found applications in a range of interesting problems, especially when large datasets are involved \cite{sparse2009}, as discussed in Chapter \ref{QGP}. 

When dealing with non-sparse but low-rank matrices, another technique of Hamiltonian simulation involving density matrix exponentiation \cite{lloyd2014quantum} can potentially be applied. Note that the covariance matrices are by definition symmetric, real and positive semi-definite, and therefore have a very similar mathematical structure to the density matrix representation of quantum states. Hence this seminal technique of density matrix exponentiation potentially allows us to implement $\mathrm{e}^{-iAt}$ in $\tilde{\mathcal{O}}(\log n)$ time, even if the matrix is not sparse. However,  the covariance matrix needs to be normalised to have a unit trace for the application of density matrix exponentiation. This pre-processing can be done efficiently if one can exploit the analytical structure of the covariance matrix. Also note that if the eigenvalues of the covariance matrix are relatively uniform, the time required to implement the unitary for a complete cycle will scale as $\mathcal{O}(n)$. Hence applying density matrix exponentiation is most effective when the covariance matrix is approximately low-rank \cite{lloyd2014quantum}. 
\paragraph{Stochastic trace estimation}
We briefly compare the quantum log determinant algorithm with classical stochastic trace estimation methods. It is clear that the quantum algorithm offers a precise method to compute $\log(A)$ rather than either the truncated Taylor series or Chebyshev polynomial approximations. When measurements of the second register are taken, a single $\log(\lambda_i)$ is computed, and hence our proposed approach can be seen as quantum stochastic trace estimation. The main advantage, however, comes from the reduction in computation time from polynomial to sub-linear. A natural question which arises is whether the complete GP training can scale sub-linearly in $n$, since if not, an exponential improvement in computing the $\mathrm{LML}$ in each step would yield only a polynomial improvement in precision.

\section{Variation estimation}

The figure of merit for the estimation error is the relative variance, as it quantifies the amount of dispersion between the estimated and the actual value of $\mathrm{LML}$. In order to demonstrate the quantum advantage in the training process, it is therefore necessary to show that the relative variance with respect to a change in hyperparemeter, $\delta\theta$, does not scale up with $n$. We consider the following,
\begin{align}
\frac{\text{Var}\left[\delta \mathrm{LML}\right]}{\left[\delta \mathrm{LML}\right]^2}=\frac{ \var\left[\log[\det(A)]\right]+\var\left[ \mathbf{y}^TA^{-1}\mathbf{y} \right] }{\left[\frac{\partial}{\partial\theta}\left(\log[\det(A)]+\mathbf{y}^TA^{-1}\mathbf{y}\right)\delta\theta\right]^2}.
\end{align}
Now we write the $\mathbf{y}$ as a linear combination of the eigenvectors, $\mathbf{e}_i$ of $A$, such that $\mathbf{y}=\sum_i\gamma_i\mathbf{e}_i$, and $\mathbf{y}^TA^{-1}\mathbf{y}=\sum_i|\gamma_i|^2\lambda_i^{-1}$, we have
\begin{align}
\frac{ \var\left[\log[\det(A)]\right]+\var\left[ \mathbf{y}^TA^{-1}\mathbf{y} \right] }{\left[\frac{\partial}{\partial\theta}\left(\log[\det(A)]+\mathbf{y}^TA^{-1}\mathbf{y}\right)\delta\theta\right]^2}
&\le\frac{n^2 \left(\var\left[\log \lambda_i\right] + \frac{1}{4}\left<y_i^2\right>\sigma_n^{-2}\right)}{\left[\frac{\partial}{\partial\theta}\left(\sum_i\log\lambda_i+\sum_i|\gamma_i|^2\lambda_i^{-1}\right)\delta\theta\right]^2}\nonumber\\
&\le\frac{\left<(\log\lambda_i)^2\right>+\frac{1}{4}\left<y_i^2\right>\sigma_n^{-2}}{\left<\delta\lambda_i/\lambda_i+\delta\left(|\gamma_i|^2/\lambda_i \right)\right>^2}, \label{revar}
\end{align} 
where the expectation value notation is used to denote the average over all choices of $i$. Hence the relative variance in estimating the variation of $\mathrm{LML}$ with respect to a training step has no explicit dependence on $n$.

Note that the number of hyperparameters is dependent only on the kernel, and thus potentially independent of the number of data points. Provided we are working to constant precision, the number of optimisation steps which require $\mathrm{LML}$ computation is upper bounded by a constant.

\section{Summary}

We have shown a quantum procedure for calculating $\mathrm{LML}$ which improves the efficiency from a classical $\mathcal{O}(n^3)$ scaling to a logarithmic scaling with respect to the size of input under certain assumptions. Specifically, if either the structure of the covariance matrix is constant $s$-sparse or approximately low-rank, the quantum approach provides an exponential speed-up. Even in the cases when the Hamiltonian simulation step inevitably consumes a $\tilde{O}(n\log n)$ time overhead, this quantum algorithm still achieves a polynomial speed-up over the best known classical approach to training full-rank GPs. When applied to a non-sparse covariance matrix that has a low-rank structure, the density matrix exponentiation procedure \cite{lloyd2014quantum} can still lead to a logarithmic time algorithm. In other cases, the singular value estimation based linear system algorithm presented in Chapter \ref{QDLSA} can be applied to achieve a runtime that scales as $\mathcal{O}(\sqrt{n}\log n)$, which provides a polynomial speed-up over its best known classical counterpart, provided that the spectral norm of $A$ is bounded by a constant with respect to the growth of $n$.

The quantum GP training procedure presented in this chapter provides an efficient way to evaluate the performance of a given kernel matrix, which is a crucial component of the model selection problem in supervised learning. This procedure applied in conjunction with the quantum GP algorithm in Chapter \ref{QGP} provides a complete quantum approach for statistical inference with GP models, which can lead to an exponential or polynomial speedup over its best-known classical counterpart, depending on the specific kernel matrix structures.

\theoremstyle{plain}
\newtheorem{claim}{Claim}
\newtheorem{theorem}{Theorem}
\newtheorem{corollary}{Corollary}
\newtheorem{lemma}{Lemma}
\theoremstyle{definition}
\newtheorem{definition}{Definition}
\newtheorem{protocol}{Protocol}

\chapter{Quantum Bayesian Deep Learning}
\label{Chapter: QBDL} 
We have presented a complete quantum approach to supervised learning with Gaussian processes in Chapters \ref{QGP} and \ref{QGPT}. By now we have seen the quantum algorithms for computing the predictive mean and variance of a GP posterior as well as the $\mathrm{LML}$ which is the core component of training a GP model. In this chapter, we exploit the connection between GPs and neural networks as discussed in Section \ref{sec: gpdl}, and apply the quantum enhanced GPs to design a quantum algorithm for deep learning. We will also experimentally demonstrate the algorithm on contemporary quantum computers and analyse its robustness with respect to realistic noise models. Specifically, we will make use of both the Rigetti Forest~\cite{smith2016practical} and the IBM QISKit~\cite{cross2017open} software stacks to implement the quantum algorithm and provide an analysis of the performance of simulators under a realistic noise model. When using real quantum processing units, we implement a simplified, shallow-circuit version of the algorithm, and compare the outcome with the simulations. The results presented in this chapter are based on Ref. \cite{zhao2018bayesian}.

\newpage
\section{Quantum Bayesian training of neural networks}\label{sec: algorithm}

Bayesian methods provide great advantages compared to traditional techniques in machine learning, which include automated ways of learning structure and avoiding overfitting, robustness to adversarial attacks~\cite{bradshaw2017adversarial,grosse2017how} and the ability to estimate uncertainties associated with predictions as previously discussed. The Bayesian framework has novelly been extended to various deep architectures \cite{blundell2015weight,gal2016dropout}.
Recent advances in this direction have established a connection between deep feedforward neural networks and Gaussian processes. This connection novelly allows for Bayesian training of deep neural networks with a Gaussian prior, circumventing the more traditional backpropagation procedure\cite{lee2017deep,gmatthews2018gaussian}.
We have briefly reviewed this correspondence between GP and deep neural networks in Section \ref{sec: gpdl}.
Recall that the base case covariance matrix $K^0$ has elements 
\begin{align}
K^0(\mathbf{x},\mathbf{x}^\prime)=\sigma_b^2+\sigma_w^2\left(\frac{\mathbf{x}\cdot \mathbf{x}^\prime}{d_{in}}\right).    
\end{align}
To compute the covariance matrix corresponding to the $l^{th}$ layer of the network, we use the following recursive formula to forward propagate the kernel,
\begin{align}
K^l(\mathbf{x},\mathbf{x}^\prime)=\sigma_b^2+\sigma_w^2\mathbb{E}[\phi(z_i^{l-1}(\mathbf{x}))\phi(z_i^{l-1}(\mathbf{x}^\prime))].
\end{align}
For a general non-linear activation function $\phi$, this can only be evaluated with numerical integration. Therefore a complete quantum algorithm for general activation functions is likely to be untraceable. 
Fortunately, there is a useful special case in which only the ReLU activation function, $f(x)=\text{max}(0,x)$, is used on each layer.
In this case, the $l^{th}$ layer covariance function has the following analytical form \cite{lee2017deep}:
\begin{align}
&K^l(\mathbf{x},\mathbf{x}^\prime)\nonumber\\
=&\sigma_b^2+\frac{\sigma_w^2}{2\pi}\sqrt{K^{l-1}(\mathbf{x}^\prime,\mathbf{x}^\prime)K^{l-1}(\mathbf{x},\mathbf{x})}\left(\arcsin(\theta^{l-1}_{\mathbf{x},\mathbf{x}^{\prime}})-(\pi-\theta^{l-1}_{\mathbf{x},\mathbf{x}^{\prime}})\arccos(\theta^{l-1}_{\mathbf{x},\mathbf{x}^{\prime}})\right)
\label{Relu},
\end{align}
where 
\begin{align}
\theta^{l}_{\mathbf{x},\mathbf{x}^{\prime}}
=\arccos\left(\frac{K^{l}(\mathbf{x},\mathbf{x}^\prime)}{\sqrt{K^{l}(\mathbf{x},\mathbf{x})K^{l}(\mathbf{x}^\prime\mathbf{x}^\prime)}}\right).    
\end{align} 
Note that the non-linear functions featured in Eq. \ref{Relu} can be approximated by polynomial series with certain convergence conditions. The factor $K^{l}(x,x)K^{l}(x^\prime,x^\prime)$ represents outer products between the two identical vectors of diagonal entries in $K^{l}$. As such, the computation of Eq. \ref{Relu} can be decomposed into such outer product operations combined with element-wise matrix multiplication. For a $L$-layer infinite width neural network, the formula Eq. \ref{Relu} needs to be evaluated for all positive integer values of $l\le L$.


\paragraph{Applying quantum GP}
Recall that the quantum GP algorithm in Chapter \ref{QGP} computes the mean predictor, 
$\bar{{f}_*}=\mathbf{k}_*^T(K+\sigma_n^2I)^{-1}\mathbf{y}$ and the variance predictor, 
$\mathbb{V}[{f}_*]=k\left(\mathbf{x}_*, \mathbf{x}_*\right)-\mathbf{k}_*^T(K+\sigma_n^2I)^{-1}\mathbf{k}_*$ of a GP posterior, where $(K+\sigma_n^2I)$ is the covariance matrix with Gaussian noise entries of variance $\sigma_n^2$, and $\mathbf{k}_*$ is the row in the covariance matrix that corresponds to the target point for prediction. 
Assuming the oracular access to the matrix elements of $K$, the quantum GP algorithm simulates $(K+\sigma_n^2I)$ as a Hamiltonian acting on an input state, $\ket{\mathbf{b}}$, and performs phase estimation to extract the eigenvalues of $(K+\sigma_n^2I)$. By inverting the eigenvalues in a superposition and performing a controlled-rotation on an ancillary system base on the inverted eigenvalues, the algorithm probabilistically completes a computation of $(K+\sigma_n^2I)^{-1}\ket{\mathbf{b}}$. 
We then use a quantum inner product estimation procedure to obtain a good estimate for $\mathbf{k}_*^T(K+\sigma_n^2I)^{-1}\mathbf{b}$. The encoding state $\ket{\mathbf{b}}$ is chosen to be $\ket{\mathbf{b}}=\ket{\mathbf{y}}$ or $\ket{\mathbf{b}}=\ket{\mathbf{k}_*}$ for computing the mean or variance predictor respectively. To apply the quantum GP algorithm for the Bayesian training of a $L$-layer infinite width neural network, we simply use $\{\mathbf{x}_i\}^n_{i=1}$ and $\mathbf{y}$ to represent the input and output points of the training set of the neural network, and choose the elements of $K$ by evaluating the covariance function $K^L(\mathbf{x},\mathbf{x}^\prime)$. The non-trivial extension to the quantum GP algorithm needed is for coherently evaluating $K^L(\mathbf{x},\mathbf{x}^\prime)$, which we will address in the following Sections \ref{sub: single}, \ref{sub: multi} and \ref{sub: Element-wise}.
It is important to clearly state the assumptions about how the matrix elements of $K^0$ can be accessed. We consider the following two different (but related) models: 
Firstly, we can assume black-box access to the elements of $K^0$. In this model, the Hamiltonian simulation subroutine discussed in Section \ref{sub: BHS} can be directly used in the quantum GP algorithm. Secondly, we can assume that $K^0$ is presented as the quantum density matrix of a qubit system. Multiple copies of such a density matrix allow for a technique inspired by the quantum principle component analysis algorithm~\cite{rebentrost2014quantum}. We will use the first model for the simplest case of a single-layer network and the second model for the multiple-layer deep architecture.

\subsection{Single-layer case}
\label{sub: single}
For the simplest single-layer case, we assume black-box access to the matrix elements of the base case such that we have the oracle $O_{K^0}$ to perform the following mapping,
\begin{align}
O_{K^0} \ket{j,k} \ket{z} \to  \ket{j,k} \ket{z \oplus K^0_{jk}},    
\end{align}
where the matrix elements are denoted as $K^0_{jk}=K^0(\mathbf{x}_j,\mathbf{x}_k)$.
The desired kernel function of Eq. \ref{Relu} can be implemented by direct classical computation on oracle queries. The desired covariance matrix, $K^1$ is then simulated as a Hamiltonian, as discussed in Section \ref{sub: BHS}, in order to construct the controlled unitary operation needed for the quantum GP algorithm.


\subsection{Multi-layer case}
\label{sub: multi}

In the case of multi-layer network architectures, we describe a method to simulate the $l^{th}$ layer kernel matrix as a Hamiltonian.
Our approach is inspired by the quantum principle component analysis algorithm~\cite{rebentrost2014quantum} where the density matrix $\rho$ of a quantum state is treated as a Hamiltonian and used to construct the desired controlled unitary $e^{i t\rho}$ acting on a target quantum state for a time period $t$.
A thorough description of this density matrix-based Hamiltonian simulation procedure is presented in Ref.~\cite{kimmel2017hamiltonian}.
Here we will first give an overview of the quantum method, while the detailed analysis is presented later in Section \ref{sub: Element-wise}.

To apply density matrix-based Hamiltonian simulation using the $l^{th}$ layer covariance matrix, we need to incorporate techniques to compute certain element-wise matrix operations between two density matrices.
It is convenient to define the following:
\begin{align}
S_1=\sum_{ j, k} |j\rangle \langle k|\otimes |j\rangle \langle k|  \otimes  |k\rangle \langle  j |,\\
S_2=\sum_{ j, k} |j\rangle \langle j|\otimes |k\rangle \langle k|  \otimes  |k\rangle \langle  j |.
\end{align}
With an augmented version of the density matrix exponentiation scheme of Ref. \cite{rebentrost2014quantum}, $S_1$ computes the exponential of the Hadamard product of two density matrices, while $S_2$ computes the exponential of the outer product between the diagonal entries of two density matrices.
Specifically, we have
\begin{align}
{\rm tr}_{1,2} \{ e^{- i S_1 \delta} ( \rho_1 \otimes \rho_2 \otimes \sigma ) e^{ i S_1 \delta} \} 
=\exp[-i(\rho_1 \odot \rho_2)\delta]\sigma\exp[i(\rho_1 \odot \rho_2)\delta]+ \mathcal{O}(\delta^2),
\label{element-wise product}
\end{align}
where $\rho_1 \odot \rho_2$ denotes the Hadamard product between $\rho_1$ and $\rho_2$, and ${\rm tr}_{1,2}$ denotes the partial trace over the first and second subsystems. The factor $\delta$ represents a small evolution time.
We also have
\begin{align}
{\rm tr}_{1,2} \{ e^{- i S_2 \delta} ( \rho_1 \otimes \rho_2 \otimes \sigma ) e^{ i S_2 \delta} \} 
=\exp[-i(\rho_1 \oslash \rho_2)\delta]\sigma\exp[i(\rho_1 \oslash \rho_2)\delta]+ \mathcal{O}(\delta^2),
\label{diagonal outer product}
\end{align}
where $\rho_1 \oslash \rho_2$ denotes the outer product between the diagonal entries of $\rho_1$ and $\rho_2$.
The derivation of Eq. \ref{element-wise product} and Eq. \ref{diagonal outer product} are presented in Section~\ref{sub: Element-wise}.
Both $S_1$ and $S_2$ are sparse and hence can be efficiently simulated as Hamiltonians with quantum walk based algorithms \cite{berry2012black,berry2015hamiltonian}.
We then need to make use of some polynomial series in $K^0(x,x^\prime)$ to approximately compute $K^l(x,x^\prime)$. Note that the products involved in this polynomial are the Hadamard product denoted by $\odot$, and the diagonal outer product denoted by $\oslash$. We will denote the polynomial in $K^0$ to the order $N(l)$ which approximates the $l^{th}$ layer kernel function as $P^{N}_{\odot, \oslash}(K^0)$.
By using a generalised $\tilde{S}$ operator which combines the components in $S_1$ and $S_2$, one can implement a total number $N$ of $\odot$ and $\oslash$ operations in arbitrary orders. In Section~\ref{sub: Element-wise}, we will show this simply amounts to summing over the tensor product of the projectors $|j\rangle \langle j|$, $|j\rangle \langle k|$, and $|k\rangle \langle k|$. Similar polynomial series simulation problems were addressed in Refs. \cite{kimmel2017hamiltonian,rebentrost2016quantum}, but the type of product considered in these works was standard matrix multiplication instead of element-wise operations. 

The method described above allows for approximately implementing the operation $e^{i tK^l}\sigma e^{-i tK^l}$, where $\sigma$ is an arbitrary input state which in our case is taken to be $\sigma=\ket{\mathbf{b}}\bra{\mathbf{b}}$. 
Thus given multiple copies of a density matrix which encodes the initial layer covariance matrix, $K^{0}$, the unitary operator, $\exp(-it K^l)$ can be constructed to act on an arbitrary input state, as required by applying the quantum GP algorithm.


\subsection{Coherent element-wise operations}
\label{sub: Element-wise}
In this section, we give a more formal description of the quantum method to compute the polynomial $P^{N}_{\odot, \oslash}(K^0)$. The main results needed are summarised by the following Lemmas~\ref{lemma: hadamard} and~\ref{lemma: outer}, and Theorem~\ref{theorem: poly}. 

\begin{lemma}[Hadamard product simulation \cite{zhao2018bayesian}]
\label{lemma: hadamard}
Given $\mathcal{O}(t^2/\epsilon)$ copies of $d$-dimensional qubit density matrices, $\rho_1$ and $\rho_2$, let $\rho_1 \odot \rho_2$ denote the Hadamard product between $\rho_1$ and $\rho_2$.
There exists a quantum algorithm to implement the unitary $e^{-i \rho_1 \odot \rho_2 t}$ on a $d$-dimensional qubit input state $\sigma$,
for a time $t$ to accuracy $\epsilon$ in operator norm. 
\end{lemma}

\begin{proof}
The usual $SWAP$ matrix for quantum principal component analysis~\cite{rebentrost2014quantum} is given by
$
S=\sum_{ j, k} |j\rangle \langle k|  \otimes  |k\rangle \langle  j | .
$
Here we take the modified $SWAP$ operator 
$
S_1=\sum_{ j, k} |j\rangle \langle k|\otimes |j\rangle \langle k|  \otimes  |k\rangle \langle  j |.
$
With an arbitrary input state $\sigma$, the following operation can be efficiently approximated for small $\delta$:
\begin{equation}
{\rm tr}_{1,2} \{ e^{- i S_1 \delta} ( \rho_1 \otimes \rho_2 \otimes \sigma ) e^{ i S_1 \delta} \},
\end{equation}
The trace is over the subspaces of $\rho_{1}$ and $\rho_{2}$.  
Expanding to $\mathcal{O}(\delta^2)$ leads to
\begin{align}
&{\rm tr}_{1,2} \{ e^{- i S_1 \delta}( \rho_1 \otimes \rho_2 \otimes \sigma ) e^{ i S_1 \delta} \} \\ \nonumber
=& 1-i  {\rm tr}_{1,2} \{ S_1 ( \rho_1 \otimes \rho_2 \otimes \sigma ) \} \delta + i {\rm tr}_{1,2} \{ ( \rho_1 \otimes \rho_2 \otimes \sigma ) S_1 \}\delta  +\mathcal{O}(\delta^2).
\end{align}
Examining the first $\mathcal{O}(\delta)$ reveals
\begin{eqnarray}
{\rm tr}_{1,2}  \{ S_1 ( \rho_1 \otimes \rho_2 \otimes \sigma )  \} &=& {\rm tr}_{1,2} \{ \sum_{ j, k} |j\rangle \langle k|\otimes |j\rangle \langle k|  \otimes  |k\rangle \langle  j | ( \rho_1 \otimes \rho_2 \otimes \sigma )  \} \nonumber \\
&=& \sum_{ n, m,j,k}\langle n|j\rangle \langle k|\rho_1 |n\rangle \langle m |j\rangle \langle k|  \rho_2 |m\rangle  |k\rangle \langle  j | \sigma  \nonumber \\
&=& \sum_{ j,k} \langle k|\rho_1 |j \rangle  \langle k|  \rho_2 |j\rangle  |k\rangle \langle  j | \sigma  \nonumber \\
&=&   (\rho_1 \odot \rho_2 ) \sigma.
\end{eqnarray}
In the same manner we have
\begin{eqnarray}
{\rm tr}_{1,2}  \{  ( \rho_1 \otimes \rho_2 \otimes \sigma )S_1  \} &=& \sigma (\rho_1 \odot \rho_2 ).
\end{eqnarray}
Thus in summary, we have shown that
\begin{align}
{\rm tr}_{1,2} \{ e^{- i S_1 \delta} ( \rho_1 \otimes \rho_2 \otimes \sigma ) e^{ i S_1 \delta} \}=
\sigma -i [(\rho_1 \odot \rho_2 ) ,\sigma] \delta + \mathcal{O}(\delta^2).
\end{align}
The above is equivalent to applying the unitary $\exp[-i(\rho_1 \odot \rho_2)\delta]$ to $\sigma$ up to $\mathcal{O}(\delta^2)$:
\begin{align}
&\exp[-i(\rho_1 \odot \rho_2)\delta]\sigma\exp[i(\rho_1 \odot \rho_2)\delta]\nonumber\\
=&[I-i(\rho_1 \odot \rho_2)\delta+\mathcal{O}(\delta^2)]\sigma[I+i(\rho_1 \odot \rho_2)\delta+\mathcal{O}(\delta^2)]\nonumber\\
=&\sigma -i[(\rho_1 \odot \rho_2 ),\sigma]\delta + \mathcal{O}(\delta^2).
\end{align}
Comparing the above two equations validates Eq. \ref{element-wise product}. Note that if the small time parameter is taken to be $\delta=\epsilon/t $, and the above procedure is implemented $\mathcal{O}(t^2/\epsilon)$ times, the overall effect amounts to implementing the desired operation, $e^{-i\rho t}\sigma e^{i\rho t}$ up to an error $\mathcal{O} (\delta^2 t^2/\epsilon )= \mathcal{O} (\epsilon )$, while consuming $\mathcal{O} (t^2/\epsilon )$ copies of $\rho_1$ and $\rho_2$.
This concludes the proof of Lemma~\ref{lemma: hadamard}.    
\end{proof}
Note that an alternative approach for Hadamard product simulation is described in \cite{limingthesis}, where the input are given as Hamiltonian on the exponents of unitary operators, rather than density matrices as discussed here.

\begin{lemma}[Diagonal outer product simulation \cite{zhao2018bayesian}]
\label{lemma: outer}
Given $\mathcal{O}(t^2/\epsilon)$ copies of $d$-dimensional qubit density matrices, $\rho_1$ and $\rho_2$, let $\rho_1 \oslash \rho_2$ denote the outer product between the diagonal entries of $\rho_1$ and $\rho_2$.
There exists a quantum algorithm to implement the unitary $e^{-i \rho_1 \oslash \rho_2 t}$ on a $d$-dimensional qubit input state, $\sigma$,
for a time $t$ to accuracy $\epsilon$ in operator norm. 
\end{lemma}

\begin{proof}
By simply re-indexing the $S_1$ operator, one obtains $S_2=\sum_{ j, k}|j\rangle \langle j|\otimes |k\rangle \langle k|  \otimes  |k\rangle \langle  j |$. Analogously with the proof of Lemma~\ref{lemma: hadamard}, we have
\begin{align}
{\rm tr}_{1,2} \{ e^{- i S_2 \delta} ( \rho_1 \otimes \rho_2 \otimes \sigma ) e^{ i S_1 \delta} \}
=
\sigma -i [(\rho_1 \oslash \rho_2 ) ,\sigma] \delta + \mathcal{O}(\delta^2).
\end{align}
The above equation can be compared with
\begin{align}
\exp[-i(\rho_1 \oslash \rho_2)\delta]\sigma\exp[i(\rho_1 \oslash \rho_2)\delta]
=\sigma -i[(\rho_1 \oslash \rho_2 ),\sigma]\delta +\mathcal{O}(\delta^2).
\end{align}
The equivalence up to the linear term in $\delta$ validates of Eq. \ref{diagonal outer product}. As with Lemma~\ref{lemma: hadamard}, with $\mathcal{O}(t^2/\epsilon)$ repetitions consuming $\mathcal{O}(t^2/\epsilon)$ copies of $\rho_1$ and $\rho_2$, the desired $e^{-i\rho t}\sigma e^{i\rho t}$ can be implemented up to error $\epsilon$. 
\end{proof}
Given the density matrix $\rho=K^0$ which encodes the base case covariance matrix, we approximate the non-linear kernel function at $l^{th}$ layer with the order $N$ polynomial, $P^{N}_{(\odot, \oslash)}(\rho)=\sum_r^N c_r \rho ^{ (\odot, \oslash)r}$.
Here the label $(\odot, \oslash)$ indicates that we work in the setting where the types of product operation involved for taking the $r^{th}$ power of $\rho$ are arbitrary combinations of Hadamard products and diagonal outer products. Now we are in the position of presenting the main theorem required to implement the kernel function at the $l^{th}$ layer. 

\begin{theorem}[Element-wise polynomial simulation \cite{zhao2018bayesian}]\label{theorem: poly} 
Given $\mathcal{O}(N^2 t^2/\epsilon)$ copies of the $d$-dimensional qubit density matrix $\rho$, and the order-$N$ polynomial of Hadamard and diagonal outer products, $P^{N}_{\odot, \oslash}(\rho)=\sum_r^N c_r \rho ^{ (\odot, \oslash)r}$, 
there exists a quantum algorithm to implement the unitary $e^{-i P^{N}_{(\odot, \oslash)}(\rho) t}$ on a $d$-dimensional qubit input state $\sigma$
for a time $t$ to accuracy $\epsilon$ in operator norm. 
\end{theorem}
\begin{proof}
We first address how to implement the unitary $e^{-i \rho ^{ (\odot, \oslash)r } t}$. Intuitively, this can be achieved by constructing a generalized $\tilde{S}$ operator with tensor product components of $|j\rangle \langle j|$, $|j\rangle \langle k|$, $|k\rangle \langle k|$ and $|k\rangle \langle j|$, corresponding to the contributing elements in the matrices in each term. We give a recursive procedure to determine $\tilde{S}$:

In the case of $r=2$, we have already shown in Lemma~\ref{lemma: hadamard} and Lemma~\ref{lemma: outer} the desired operation can be achieved using $S_1$ and $S_2$ corresponding to the $\odot$ and $\oslash$ cases respectively. Thus we can write the base case of the recursive procedure as 
\begin{align}
\tilde{S}^{(r=2)}= \sum_{j,k} T^{(2)}(j,k)\otimes |k\rangle \langle j|,    
\end{align}
where $T^{(2)}(j,k)$ denotes the possible combinations of tensor products, $|j\rangle \langle k|\otimes |j\rangle \langle k|$ or $|j\rangle \langle j|\otimes |k\rangle \langle k|$. Now consider the $r=3$ case, the additional factor of $\rho$ will come in two possible cases. If it comes as a $\odot$ product, the updated operator $\tilde{S}^{(r=3)}_\odot$ is simply given by 
\begin{align}
\tilde{S}^{(r=3)}_\odot = \sum_{j,k} T^{(2)}(j,k)\otimes |j\rangle \langle k| \otimes |k\rangle \langle j|.    
\end{align}
If the additional $\rho$ comes in as a $\oslash$ product, the updated operator $\tilde{S}^{(r=3)}_\oslash$ is instead given by 
\begin{align}
\tilde{S}^{(r=3)}_\oslash = \sum_{j,k} |j\rangle \langle j| \otimes |j\rangle \langle j|  \otimes |k\rangle \langle k| \otimes |k\rangle \langle j|.    
\end{align}
This can be seen by observing that the contributing elements to a $\oslash$ product are exclusively diagonal, which we use $|j\rangle \langle j|$ to pick up. Any off-diagonal information about the previous element-wise product operations is irrelevant. In general, if we have the $r^{th}$ order $\tilde{S}$ operator given by
\begin{align}
\tilde{S}^{(r)} = \sum_{j,k} T^{(r)}(j,k)\otimes |k\rangle \langle j|,    
\end{align}
the operators $\tilde{S}^{(r+1)}_\odot$ and $\tilde{S}^{(r+1)}_\oslash$ can be generated as follows:
\begin{align}
\tilde{S}^{(r+1)}_\odot =& \sum_{j,k} T^{(r)}(j,k)\otimes |j\rangle \langle k| \otimes |k\rangle \langle j|,  \\
\tilde{S}^{(r+1)}_\oslash =& \sum_{j,k} (|j\rangle \langle j|)^{\otimes r}  \otimes |k\rangle \langle k| \otimes |k\rangle \langle j|.
\end{align}
We have shown a recursive procedure to construct $\tilde S^{(r)}$ up to $r=N$ such that 
\begin{align}
    {\rm tr}_{1...r} \{ e^{- i \tilde S^{(r)} \delta} ( \rho^{\otimes r} \otimes \sigma ) e^{ i \tilde S^{(r)} \delta} \} 
=\exp[-i \rho ^{ (\odot, \oslash)r }\delta]\sigma\exp[i \rho ^{ (\odot, \oslash)r }\delta]+ \mathcal{O}(\delta^2),
\end{align}
for a small evolution $\delta$. Analogously with Lemma~\ref{lemma: hadamard} and Lemma~\ref{lemma: outer}, with $\mathcal{O}(t^2/\epsilon)$ repetitions consuming $\mathcal{O}(rt^2/\epsilon)$ copies of $\rho$, the desired 
\begin{align}
\exp[-i \rho ^{ (\odot, \oslash)r }t]\sigma\exp[i \rho ^{ (\odot, \oslash)r }t]    
\end{align}
can be implemented up to an $\epsilon$ error. 
Finally one makes use of the Lie product formula for summing the terms in the polynomial     \cite{suzuki1992general,childs2003exponential,wiebe2010higher}:
\begin{align}
\mathrm{e}^{i\delta ( A + B)+ \mathcal{O} (\delta ^2/m)}=(\mathrm{e}^{i\delta A/m}\mathrm{e}^{i \delta B/m})^m,
\end{align}
where $A$ and $B$ are taken to different terms in $P^{N}_{\odot, \oslash}(\rho)=\sum_r^N c_r \rho ^{ (\odot, \oslash)r}$, and the factors $c_r$ simply amount to multiplying the $S^{(r)}$ matrices with the respective coefficients. The parameter $m$ can be chosen to further suppress the error by repeating the entire procedure. However, for the purpose of implementing 
$
\mathrm{e}^{-i P^{N}_{(\odot, \oslash)}(\rho) t}\sigma \mathrm{e}^{i P^{N}_{(\odot, \oslash)}(\rho) t} 
$ 
to our desired accuracy $\epsilon$, $\mathcal{O}(N^2 t^2/\epsilon)$ copies of $\rho$ are required. The quadratic dependency in the order of the polynomial, $N^2$ stems from implementing the unitary $\exp[-i \rho ^{ (\odot, \oslash)r }t]$ up to $r=N$, each consuming $\mathcal{O}(Nt^2/\epsilon)$ copies as previously argued. 
\end{proof}

\section{Experiments}

We have performed the following two sets of experiments to demonstrate the Hermitian matrix inversion component of the quantum GP algorithm:
\begin{enumerate}
    \item Simulations of the quantum matrix inversion on quantum virtual machines, the classical simulators of Rigetti's Forest API \cite{smith2016practical} 
with analysis of varying noise models' impacts on the outputs. 
    
    \item A small-scale ($2\times 2$) implementation of quantum matrix inversion in both PyQuil, run on Rigetti's Quantum Processing Unit (QPU), and in IBM's QISKit software stack, run on IBM's Quantum Experience \cite{cross2017open}.
\end{enumerate}
The PyQuil framework provides advanced gate decomposition features that allow for arbitrary unitary operations on a multi-qubit quantum state. The simulated noise models of the Rigetti's quantum virtual machine allows for an analysis of the expected accuracy and computational overhead of actual quantum implementations. 
QISKit also provides a noisy classical simulator, which we use to compare the performance of the quantum matrix inversion algorithm on the real QPU against simulations with realistic noise models.
The quantum processing units we use for actual implementations are IBM's 16-qubit Rueschlikon (IBMQX5)~\cite{wang2018ibm} and Rigetti's 8-qubit 8Q-Agave.
While the numbers of available qubits in both cases are higher than the number required for the implementation (a total of six for the $2\times 2$ matrix inversion), the depth requirement of the circuit grows significantly for larger matrices.

\subsection{Simulations on a quantum virtual machine}
Here we present the results from the simulations conducted with Rigetti's quantum virtual machine. 
We have performed two sets of experiments to analyse the effect of different types of noise on the algorithm. Firstly, we restrict to the simplest non-trivial case of inverting a $2\times 2$ matrix which is chosen to be $A=\frac{1}{2}\begin{pmatrix} 3 & 1 \\ 1 & 3 \end{pmatrix}$ with the problem-specific circuit in Ref.~\cite{cao2012quantum}. The circuit involved is significantly shallower than the one required by the full algorithm, which is described in Ref.~\cite{cao2013quantum}, making it more practically viable to implement on current and near-term quantum computers due to its reduced depth.
Secondly, we simulate the full quantum matrix inversion algorithm \cite{Harrow2009a,cao2013quantum}.
This requires a large number of ancillary qubits for the computation of the reciprocals of the eigenvalues. We will simulate the inversion of a $4\times 4$ matrix with four bits of precision.


We work with two noise models:
The first one, known as the ``gate noise'', applies a Pauli $X$ operator
with a certain probability on each qubit after every gate application.
The second one, known as the ``measurement noise'', applies a Pauli $X$ operator with certain probability only on every qubit that is measured before the measurement takes place. As such, the measurement noise can also be interpreted as a readout error.
\begin{figure*}[h]
    \centering
    \subfigure[]{
        \includegraphics[width=0.48\textwidth]{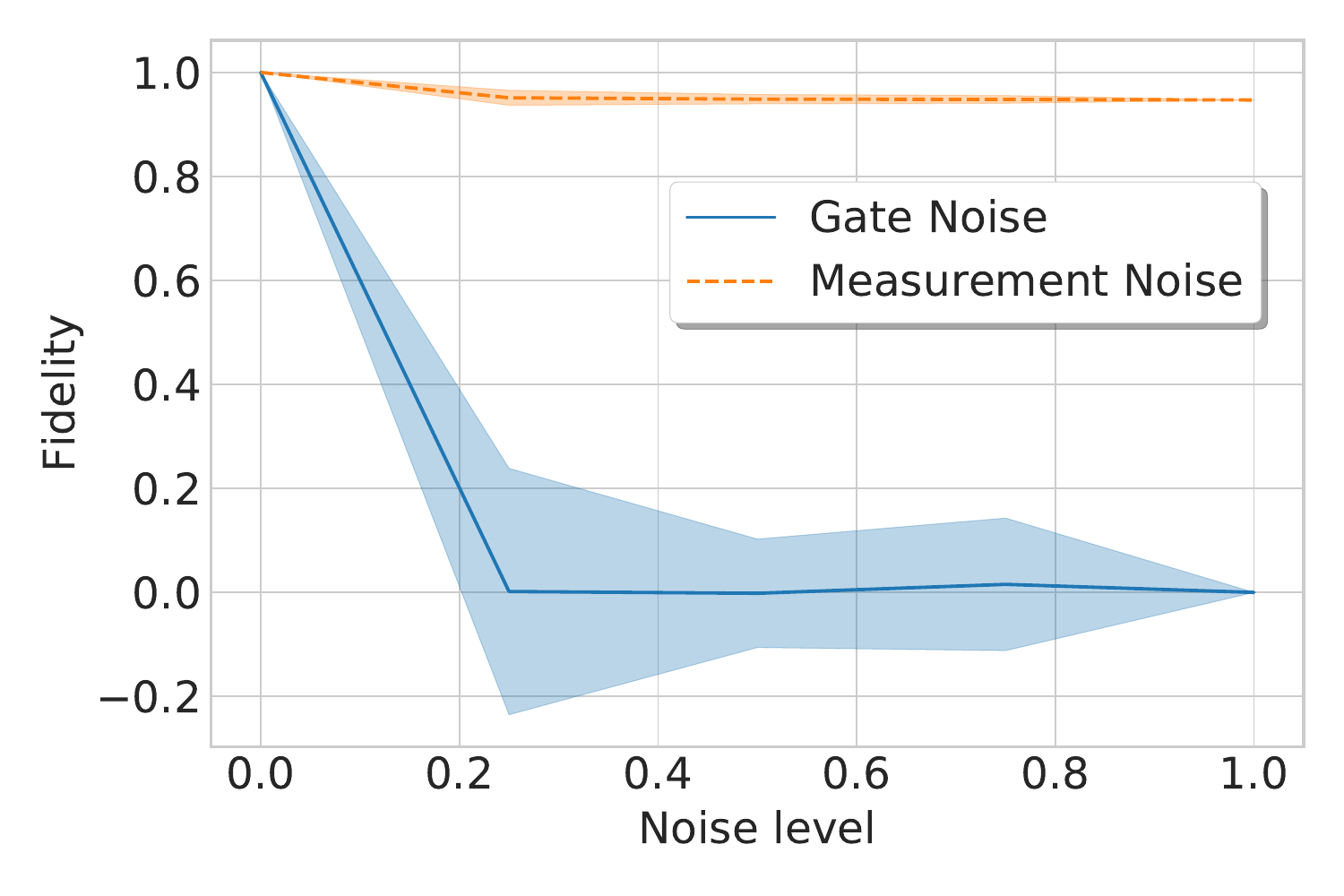}
    }
    \subfigure[]{
        \includegraphics[width=0.48\textwidth]{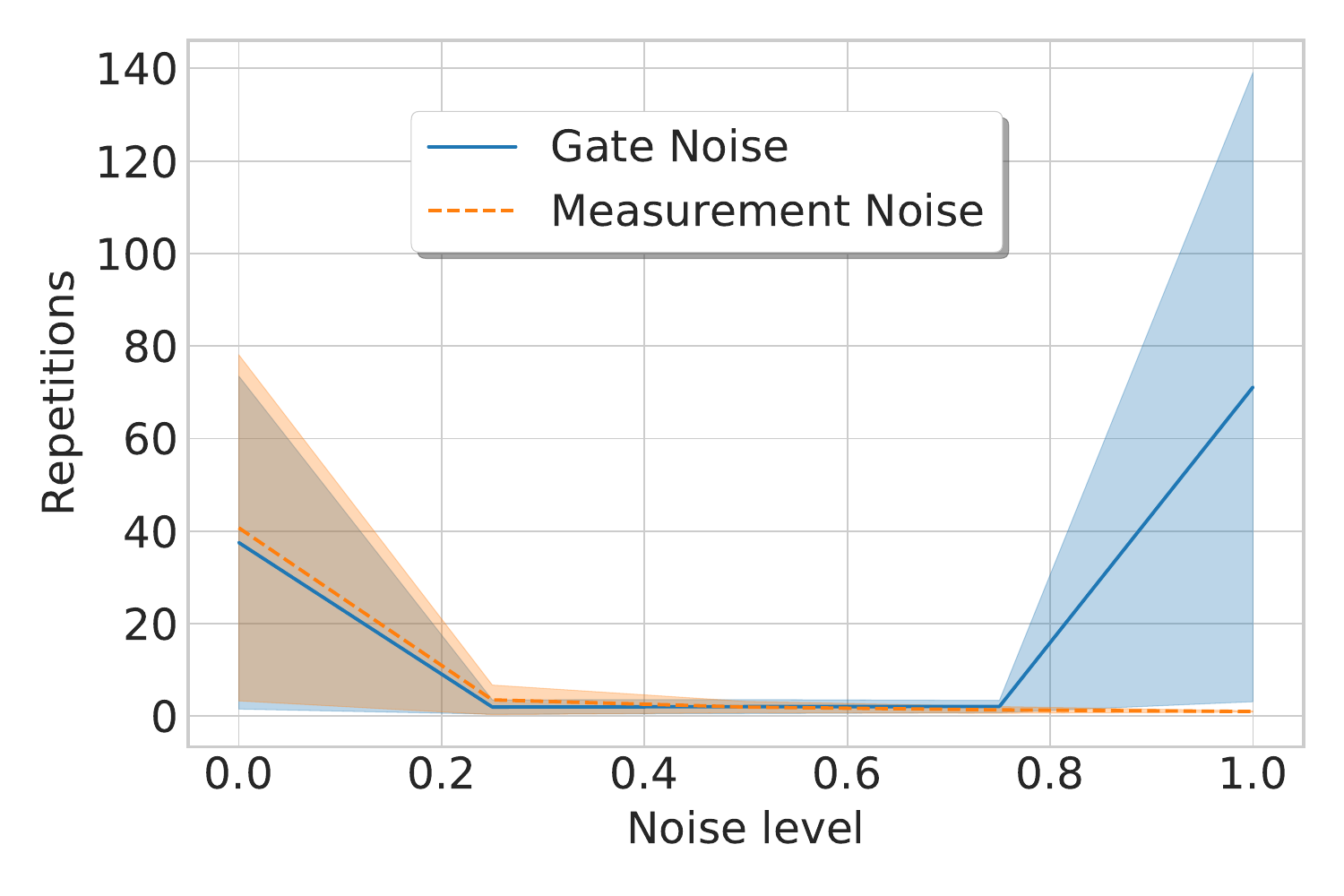}
    }
    \caption{Simulated gate and measurement noise on a specialised circuit for inverting $A$. (a) The fidelity shows the overlap with the expected correct state after the computation. A zero fidelity means the output state is orthogonal to the correct solution, while a unit fidelity means the output state is the correct result. (b) The number of repetitions indicates the average number required to execute the probabilistic program before it succeeds.}
    \label{2by2simulated}
\end{figure*}


The simulation results of the quantum inversion of the $2 \times 2$ matrix $A$ is presented in Figure \ref{2by2simulated}. 
We analyse the following two critical factors, namely the fidelity between the expected result and the simulated output, given that inversion has succeeded, and the average repetition of the coherent part of the algorithm needed to obtain a successful run. 
Note that in our noisy setting, success in the post-selection does not guarantee the correctness of the output.
Our results show that measurement noise has a smaller impact on the result than gate noise which for reasonably low noise levels already renders the output state orthogonal to the expected result.
Interestingly, as the noise level increases, the average number of repetitions decreases. 



The simulation results of general quantum matrix inversion algorithm on a random $4\times 4$ matrix is presented in Figure \ref{4by4simulated}.
We see that the output's sensitive to noise has increased as the circuit involved became deeper.
However, the noise level for which the output reaches zero fidelity is approximately the same in both the $2\times 2$ and $4\times 4$ cases, and it would be interesting to see whether it remains constant for larger instances.
The simulation still shows better robustness to measurement noise, but with its effect appearing to be stronger compared with the problem-specific algorithm of Figure \ref{2by2simulated}.
As before in the $2\times 2$ case, measurement noise introduces bit flips to registers storing measurement results, which eventually leads to an apparent low number of repetitions, but at the expense of lower fidelities with the expected output. 
\begin{figure*}[h]
    \centering
    \subfigure[]{
        \includegraphics[width=0.48\textwidth]{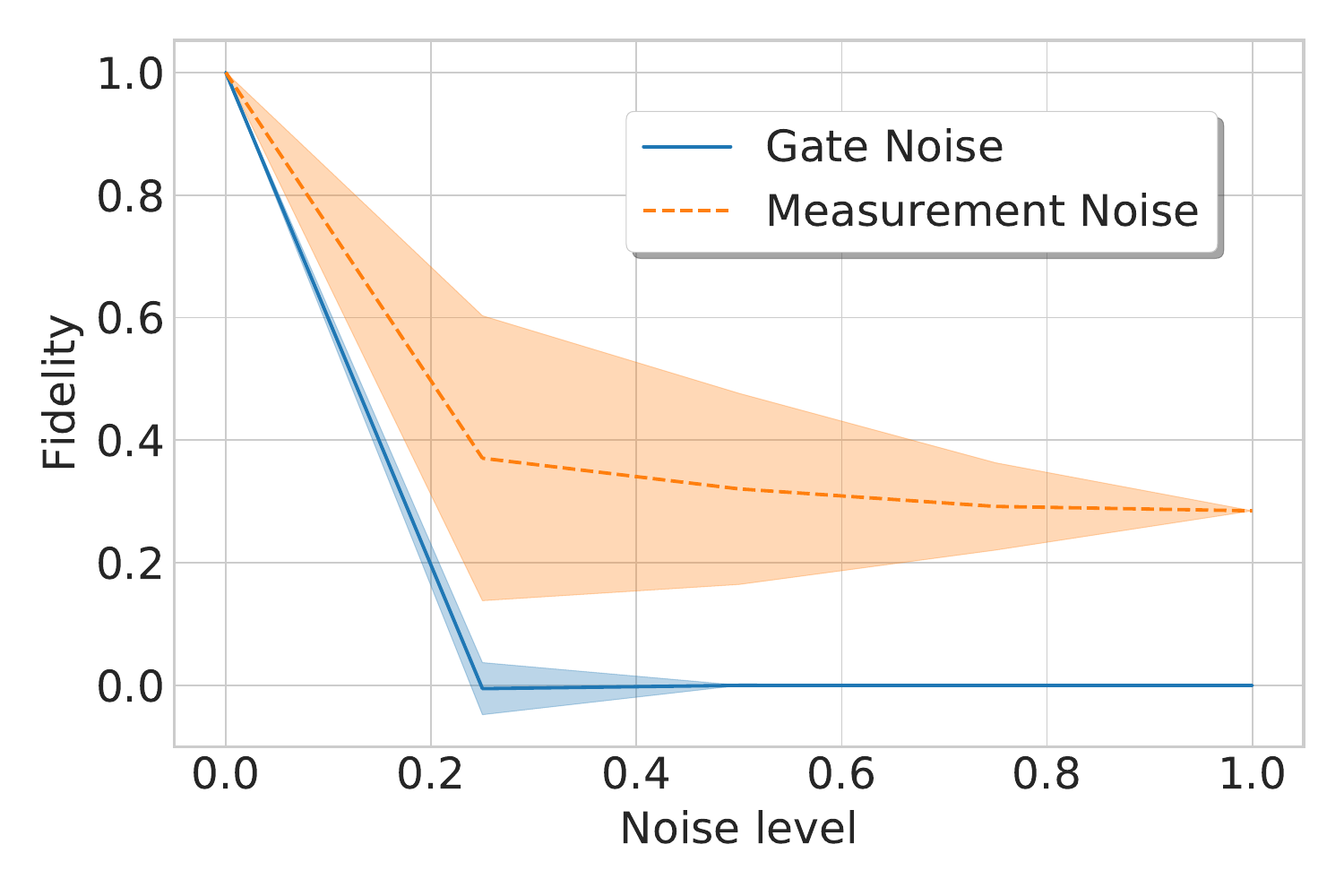}
    }
    \subfigure[]{
        \includegraphics[width=0.48\textwidth]{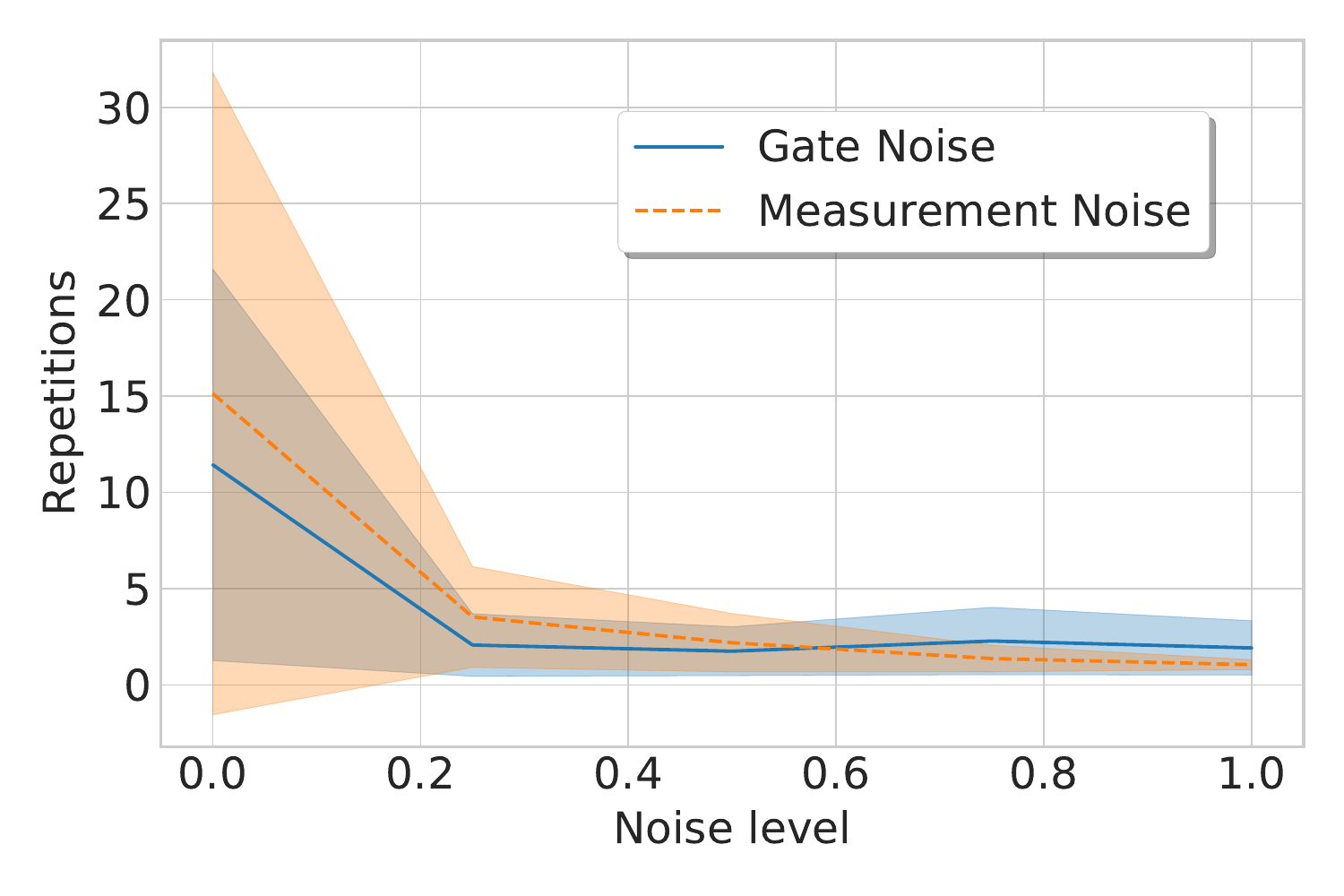}
    }
    \caption{Simulated gate and measurement noise on the generic circuit for inverting a $4\times 4$ matrix with four bits of precision on eigenvalues.}
    \label{4by4simulated}
\end{figure*}

\subsection{Implementations on quantum processing units}
\label{actualqpu}

In this section, we implement the restricted $2\times 2$-matrix inversion algorithm with two real quantum processors.
We have chosen to implement a restricted version of the algorithm due to the limitations of the currently available hardware with respect to qubit numbers, qubit-qubit connectivity, and coherence times.
Note that one does not have direct access to the complete information of the output state, but only samples of measurement results. To gauge the correctness of the output, we will perform a SWAP test\cite{gottesman2001quantum, Liming} with the expected output encoded in auxiliary qubits, and use a flag qubit to indicate a successful run of the test. With multiple runs, the figure of merit is the probability of success, $P(\textrm{success})$, which can then be related to the fidelity by $\mathcal{F}=|2P(\textrm{success})-1|$. 

We have implemented the restricted matrix inversion algorithm on both the Rigetti's 8Q-Agave and the IBM's IBMQX5 quantum processing units. The IBM QISKit software~\cite{cross2017open} also provides a classical simulator to run noisy experiments, and we use these to benchmark the performance of the runs on the real chips. As with simulations in Rigetti's software stack, we expect the measurement noise to have a smaller effect than the gate noise. 
Note that the flag qubit of the swap test is also subject to readout error under the simulated measurement noise. 
Therefore an apparent low $P(success)$ in the high measurement error regime could have included many instances of successful runs, falsely reported by the flag qubit. Gate noise on the other hand directly affects the computations in the circuit. Therefore the lower success probabilities now reflect a real discrepancy between the actual output and desired states. In this case, the success probabilities lie in the range of $[0.35, 0.6]$, which translates into fidelities in the range of $[0, 0.3]$.
The probability of success is $89\%$, which translates into a fidelity with the expected outcome of $0.78$. This is a very encouraging result, despite the small size of the matrix inverted. 
The results are shown in Figure \ref{qpus}.
\begin{figure}[ht!]
    \centering
        \includegraphics[width=1.0\textwidth]{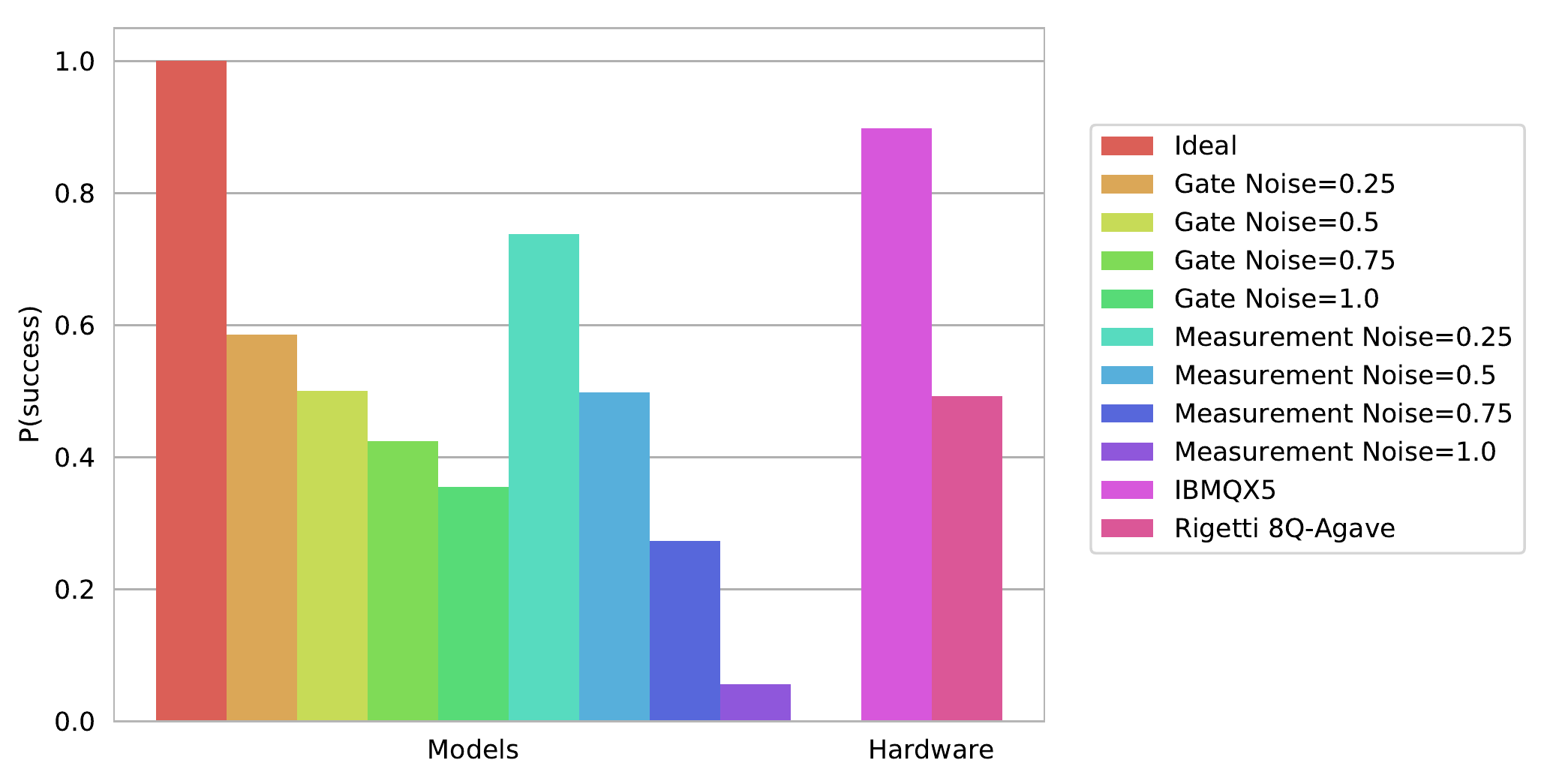}
    \caption{Success probability of the SWAP test for different noisy simulations and executions on the IBM's and Rigetti's QPUs (rightmost bars). The noise models include gate noise and measurement noise, with different probabilities of failure. The algorithm is run 8192 times for each instance, after which $P(\textrm{success})$ is evaluated.}
    \label{qpus}
\end{figure}


\section{Summary}

We have developed a quantum algorithm for a Bayesian approach to deep learning, which makes use of the quantum Gaussian processes algorithm with a kernel matrix corresponding to ReLU activation functions on each layer of the deep network with infinite width. In the simplest case of a single layer architecture, we assume the kernel matrix can be classically evaluated and efficiently simulated as a Hamiltonian to be used in the quantum GP algorithm. In the case of multi-layer, deep architectures, we worked with a model where the kernel matrix corresponding to the layer zero (the base case) can be encoded as a density matrix. We then designed a recursive procedure to simulate the Hamiltonian corresponding to the kernel matrix at an arbitrary depth, which given a fixed accuracy requirement and only consumes a quadratic number of copies of the density matrix. In order to analyse the practical feasibility of the algorithm, we implemented its core subroutine, quantum matrix inversion, on both quantum simulators and real state-of-the-art quantum processors. We observed that the accuracy drops sharply with noise, but even with current, small-scale quantum computers, reasonably high success rates can still be achieved. 

Although these experimental results are promising, we should note that they do not constitute sufficient evidence that the full quantum algorithm for Bayesian deep learning can be efficiently implemented in near-term quantum technologies. 
A fully quantum implementation, including recursively simulating the required Hamiltonian corresponding to the covariance matrix at deep layers, will be an interesting avenue for future research.

\part{Quantum correlations and causality}
\label{Part3}

\chapter{Geometry of quantum correlations} 
\label{Chapter: geometry}

\label{QCinf} 

In the previous parts of the thesis, we have seen that quantum computation can be applied to statistical inference in classical datasets. Particularly, we have focused on the statistical model of Gaussian processes, and shown that phase estimation based methods can provide provable quantum advantages. In this chapter, we take a different approach and look at another aspect of statistical inference in the quantum era, where the data itself is inherently quantum. We consider the problem of inferring quantum correlations from measurement events. The material of this chapter follows closely from Ref.\cite{zhao2017geometry}.

\section{Introduction}

The study of quantum correlations has long held an important role in fundamental physics \cite{einstein1935can, bell1964einstein}, and more recently given rise to promising prospects of quantum technologies \cite{ekert1991quantum, harrow2004superdense}. In the usual formulation of non-relativisitc quantum theory, the state of a system can extend across space but is only defined at a particular instant in time. The distinction between the roles of space and time contrasts with relativity \cite{isham1993canonical} where they are treated in an even-handed fashion, and has led to a general preference to study temporal quantum correlations in a rather separated manner from their spatial counter-parts \cite{leggett1985quantum, brukner2004quantum, budroni2013bounding, milz2017introduction, cotler2017superdensity, modi2012operational, emary2013leggett}. Here we aim at taking a unifying approach to study quantum correlations for observables defined across space-time in a general formalism. In order to do so, we make use of the pseudo-density matrix (PDM) formalism introduced in Ref. \cite{fitzsimons2015quantum} as an extended framework of quantum correlations, which generalises the notion of a quantum state to the temporal domain, treating space and time on an equal footing.

We will focus on the simplest and most fundamental case, that of two-point correlation functions. In the spatial setting, this would correspond to bipartite quantum correlations, which can exhibit entanglement. In the temporal setting, we consider the correlations between two sequential measurements separated by an arbitrary quantum channel evolution on a single qubit quantum state. Our study presents the geometry of bipartite correlations in both the spatial and temporal cases and establishes a symmetric structure between them. We observe that this symmetry is broken in the presence of certain non-unital channels. As such these non-unital channels produce a novel set of temporal correlations that are statistically identical to bipartite quantum entanglement.

\label{sec: PDMPrelim}

\subsection{Density matrices and spatial correlations}
\label{subsec: Spatial}

\paragraph{Density matrices}

As introduced previously in Section \ref{sub: DM}, a density matrix is defined as a probability mixtures of pure quantum states. However, there is also another way of interpreting the density matrices, as the mixture of the expectation values of every possible Pauli measurements resulting in a linear combination of different Pauli components. Particularly for an $n$-qubit system, we have
\begin{align}
\rho=\frac{1}{2^n}\sum\limits_{i_1=0}^{3}...\sum\limits_{i_n=0}^{3}\left< \bigotimes\limits_{j=1}^{n}\sigma_{i_j}\right>\bigotimes\limits_{j=1}^{n}\sigma_{i_j},
\end{align} 
where the indices $i$ label different Pauli operators and the identity operator with $\sigma_0=\mathbb{I}$, $\sigma_1=\mathrm{X}$, $\sigma_2=\mathrm{Y}$, and $\sigma_3=\mathrm{Z}$, while the sub-indices $j$ of each $i$ labels different qubits in the system. In order to have a valid density matrix, we need to further require $\rho$ to be positive semi-definite.

\paragraph{Geometry of spatial correlations} Consider the matrix $\mathcal{C}$ whose elements $\mathcal{C}_{kl}$ are given by the Pauli correlation functions $\langle\sigma_k\sigma_l\rangle=\Tr[(\sigma_k\sigma_l)\rho], k, l=1,2,3$ of a two-qubit bipartite state $\rho$. It is clear that under local unitary transformations, $\mathcal{C}$ can be brought into a diagonalised form $\mathcal{C}^\prime$. It is known that $\mathcal{C}^\prime$ can always be written as a convex combination of $\mathcal{C}_1=diag[1,-1,1]$, $\mathcal{C}_2=diag[-1,1,1]$, $\mathcal{C}_3=diag[1,1,-1]$ and $\mathcal{C}_4=diag[-1,-1,-1]$, which corresponds to the correlation matrices of the four maximally-entangled Bell states respectively \cite{horodecki2009quantum}. Geometrically, one can visualise this convex set of correlation functions in three-dimensional
real space as a tetrahedron whose four vertices in the $\left<XX\right>\_\left<YY\right>\_\left<ZZ\right>$ coordinate system are given by the diagonal entries of $\mathcal{C}_1$, $\mathcal{C}_2$, $\mathcal{C}_3$, and $\mathcal{C}_4$ \cite{horodecki1996separability}. We shall name such a tetrahedron the spatial tetrahedron, denoted as $\mathcal{T}_s$.
The geometry of spatial quantum correlations has been a fruitful area of research, interested readers are referred to \cite{zyczkowski2006geometry} for a comprehensive text on this subject. 

\subsection{The pseudo-density matrix formalism}
\label{sub: PDM}
The density matrix of a quantum state can be naturally extended into the temporal domain and used to define the PDM \cite{fitzsimons2015quantum} as
\begin{align}
R=\frac{1}{2^n}\sum\limits_{i_1=0}^{3}...\sum\limits_{i_n=0}^{3}\langle\{\sigma_{i_j}\}^n_{j=1}\rangle\bigotimes\limits_{j=1}^{n}\sigma_{i_j}, \label{eq:PDM}  
\end{align}
where sub-indices $j$ of each $i$ now label different measurement events in the system. The factor $\langle\{\sigma_{i_j}\}^n_{j=1}\rangle$ denotes
the expectation value of the product of the $n$ Pauli observables. Physically, it corresponds to a correlation function of a size-$n$ sequence of Pauli measurements $\sigma_{i_j}\in\{\sigma_0,...,\sigma_3\}$.
Note that $R$ is a Hermitian matrix with unit trace, as it is with conventional density matrices. Furthermore, if the measurement events are space-like separated, $R$ is positive semi-definite and hence resembles a valid density matrix. However, the mathematical structure of Eq. \ref{eq:PDM} does not exclude the possibility of having negative eigenvalues. When negative eigenvalues are present, the Pauli observables can no longer be interpreted as measurements events on distinct sub-systems of a common quantum state. In such cases, the PDM novelly captures local measurement events happening at arbitrary time instances, in contrast to the case for conventional density matrices. 
\paragraph{Measure of causality}
Since the presence of negative eigenvalues is a witness to causal relationships, it is natural to quantify temporal correlations with some measure based on the trace norm.
A causality measure was thus introduced in Ref. \cite{fitzsimons2015quantum} as $f_{tr}(R)=\|R\|_{tr}-1$, which possesses desirable properties in close analogy with entanglement monotones for spatial correlations, namely, $f_{tr}(R)\ge 0$ and $f_{tr}(R_2) = 1$ for any $R_2$ generated by two consecutive measurements of a closed system with a single qubit ($R_2$ is maximally causal);
$f_{tr}(R)$ is invariant under unitary transformations;
 $f_{tr}(R)$ is non-increasing under local operations (c.f. entanglement is non-increasing under LOCC); $f_{tr}$ is a convex function. We will revisit these properties in Chapter \ref{CauCap}, where a logarithmic variant of the trace norm measure will play a significant role.

\section{General two-time quantum correlations}
\label{sec: two-time}

Here we describe the quantum correlations between Pauli measurements at two time instances. The corresponding physical scenario is depicted in Figure \ref{fig: temporalPDM}, where a single-qubit system $\rho_A$ subject to a quantum channel between two measurement events at times $t_A$ and $t_B$. The channel is described by a completely positive trace-preserving (CPTP) map $\varepsilon_{B|A}$, which maps the family of operators from the state space $\mathcal{H}_A$ at $t_A$ to the state space $\mathcal{H}_B$ at $t_B$. 

\begin{figure}[H] 
\centering
\includegraphics[width=0.5\textwidth] {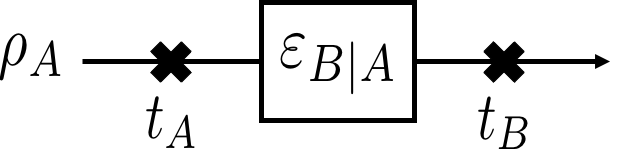} 
\caption{The physical scenario of general two-time quantum correlations: A single-qubit system $\rho_A$ is measured at $t_A$ and $t_B$ with a quantum channel in between described by the CPTP map, $\varepsilon_{B|A}$.} 
\label{fig: temporalPDM} 
\end{figure}

\subsection{The two-point temporal PDM}
It it clear from the definition of PDM, Eq. \ref{eq:PDM} that the expectation value of the product of $n$ Pauli observables is given by 
\begin{align}
\langle\{\sigma_{i_j}\}^n_{j=1}\rangle = \Tr\left[\left(\bigotimes\limits_{j=1}^{n}\sigma_{i_j}\right)R\right].\label{correlation} 
\end{align}
In the case of two sequential events, $n=2$. Supposing the evolution between $t_A$ and $t_B$ is the identity, the only non-zero Pauli correlation functions are
\begin{align}
\left<\{\sigma_1,\sigma_1\}\right>&=\left<\{\sigma_2,\sigma_2\}\right>=\left<\{\sigma_3,\sigma_3\}\right>=\left<\{\sigma_0,\sigma_0\}\right>=1,\nonumber\\
\left<\{\sigma_0,\sigma_1\}\right>&=\left<\{\sigma_1,\sigma_0\}\right>=\left<\sigma_1\right>,\nonumber\\
\left<\{\sigma_0,\sigma_2\}\right>&=\left<\{\sigma_2,\sigma_0\}\right>=\left<\sigma_2\right>,\nonumber\\
\left<\{\sigma_0,\sigma_3\}\right>&=\left<\{\sigma_3,\sigma_0\}\right>=\left<\sigma_3\right>.
\end{align}
Here $\{...\}$ denotes sets of operators, which should not be confused with a similar notation for anti-commutators.
On the other hand,  we can write a single-qubit density operator $\rho_A$ as
\begin{align}
\rho_{A}=\frac{1}{2}\left(\sigma_0+\left<\sigma_1\right>\sigma_1+\left<\sigma_2\right>\sigma_2+\left<\sigma_3\right>\sigma_3\right).
\end{align}
We now compare the coefficients of Pauli components and obtain
$
R=\{\rho_A\otimes\frac{\mathrm{I}}{2},SWAP\}, 
$
where $SWAP=\frac{1}{2}\sum_{i=0}^3\sigma_i\otimes\sigma_i$,
and here $\{...\}$ denotes the anti-commutator, such that $\{A,B\}=AB+BA$. 
In a general setting, a channel that acts on the system in between the time instances $t_A$ and $t_B$
as a CPTP map $\varepsilon_{B|A}$ is included. Note that the map does not affect any observables at $t_A$, but introduces a transformation according to its adjoint map on the observables at $t_B$. Therefore the two-time PDM across such a channel can be written as
\begin{align}
R_{AB}=(\mathcal{I}_A\otimes \varepsilon_{B|A})\left( \{\rho_A\otimes\frac{\mathrm{I}}{2},SWAP\}\right), \label{eq: PDM2}
\end{align}
where $\mathcal{I}_A$ denotes the identity super-operator acting on $A$.
The above expression is in agreement with the Jordan product representation given in Ref. \cite{horsman2017can}:
\begin{align}
R_{AB}=\{\rho_A\otimes\frac{\mathrm{I}}{2},E_{AB}\}, \label{Jordan}
\end{align}
where $E_{AB}=\sum_{ij}\left(\mathcal{I}_A\otimes\varepsilon_{B|A}\right)\left(\ket{i}\bra{j}_A\otimes\ket{j}\bra{i}_B\right)$ is an operator
acting on $\mathcal{H}_A\otimes\mathcal{H}_B$ that is Jamio\l{}kowski-isomorphic to $\varepsilon_{B|A}$. The correlations described by $R_{AB}$ are "purely" temporal in the sense that the underlying dynamics are defined by a CPTP map on a single qubit.

\subsection{Single-qubit quantum channels}
To proceed further, we need to exploit the structures of the quantum channel $\varepsilon_{B|A}$.
It was established in Ref. \cite{ruskai2002analysis} that the complete positivity requirement leads to a particularly useful trigonometric parameterisation of the set of possible $\varepsilon_{B|A}$ in the Pauli basis \cite{king2001minimal}. Concretely, this set corresponds to the convex closure of the maps defined by the following Kraus operators up to permutations among $\{\sigma_1,\sigma_2,\sigma_3\}$:
\begin{align}
K_+=&\left[\cos\frac{v}{2}\cos\frac{u}{2}\right]\sigma_0+\left[\sin\frac{v}{2}\sin\frac{u}{2}\right]\sigma_3,\nonumber\\
K_-=&\left[\sin\frac{v}{2}\cos\frac{u}{2}\right]\sigma_1-i\left[\cos\frac{v}{2}\sin\frac{u}{2}\right]\sigma_2, \label{eq: KO}
\end{align}
where $v\in[0,\pi], u\in[0,2\pi]$. The above Kraus operators act on $\sigma_i$ as the following:
\begin{align}
    K_+\sigma_0K_+^\dagger+K_-\sigma_0K_-^\dagger &=\sigma_0+\sin(u)\sin(v)\sigma_3,\nonumber\\
        K_+\sigma_1K_+^\dagger+K_-\sigma_1K_-^\dagger &=\cos(u)\sigma_1,\nonumber\\ 
        K_+\sigma_2K_+^\dagger+K_-\sigma_2K_-^\dagger &=\cos(v)\sigma_2,\nonumber\\ 
        K_+\sigma_3K_+^\dagger+K_-\sigma_3K_-^\dagger &=\cos(u)\cos(v)\sigma_3.      
\end{align}

\subsection{Convex closure}
We now expand Eq. \ref{eq: PDM2} into its Pauli components and substitute into Eq. \ref{correlation}, and obtain
\begin{align}
\langle\sigma_k\sigma_k\rangle=\Tr\left[\langle\sigma_k\rangle_{\rho_A}\varepsilon_{B|A}(\sigma_0)\sigma_k+\varepsilon_{B|A}(\sigma_k)\sigma_k\right], \label{kk}
\end{align}
where $\langle\sigma_k\rangle_{\rho_A}$ denotes the expectation value of the $\sigma_k$ observable on the initial state $\rho_A$.
By setting $\rho_A=\ket{0}\bra{0}$ and applying the Kraus operators in Eq. \ref{eq: KO}, we obtain the parametric equations which characterise the convex set of possible correlation functions as followed:
\begin{align}
\langle\sigma_1\sigma_1\rangle&=\cos(u),\nonumber\\
\langle\sigma_2\sigma_2\rangle&=\cos(v),\nonumber\\
\langle\sigma_3\sigma_3\rangle&=\cos(u-v).\label{trig}
\end{align}
The set of three Pauli correlations $\langle\sigma_k\sigma_k\rangle=\Tr\left[(\sigma_k\otimes\sigma_k) R_{AB}\right]$ fully characterises any two-point correlations $\langle\sigma_k\sigma_l\rangle$ up to local unitary transformations, for $k,l=1,2,3$. Note that the choice of permutation among $\{\sigma_1,\sigma_2,\sigma_3\}$ is arbitrary and hence does not affect the resultant convex set enclosed by the parametric surface.
We illustrate the set of attainable $\langle\sigma_k\sigma_k\rangle$ as points in the real coordinator space $\{\left<\sigma_1\sigma_1\right>,\left<\sigma_2\sigma_2\right>,\left<\sigma_3\sigma_3\right>\}$ in FIG. \ref{fig:inflate_tetrahedron}, which depicts the geometry of two-time Pauli correlations. The figure shows a parametric plot of the equations $\left<\sigma_1\sigma_1\right>=\cos(u)$, $\left<\sigma_2\sigma_2\right>=\cos(v)$ and $\left<\sigma_3\sigma_3\right>=\cos(u-v)$, where $v\in[0,\pi], u\in[0,2\pi]$. Note that a similar structure was found when three sequential observables were considered in the context of Leggett-Garg inequalities \cite{budroni2013bounding}.
\begin{figure}[h!]
\centering
\includegraphics[width=0.6\linewidth]{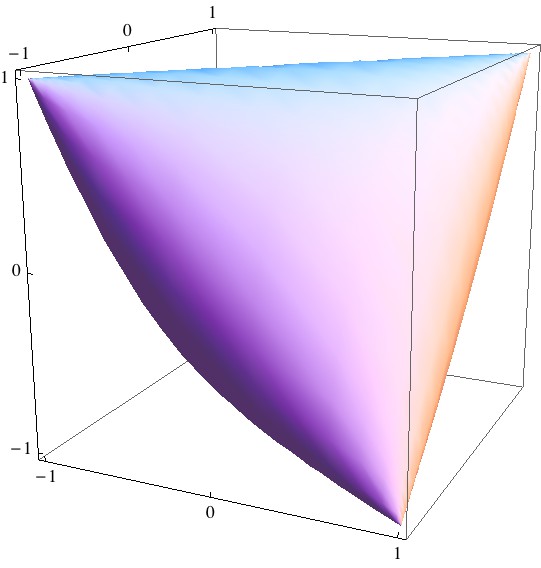}
\caption{The surface enclosing the set of possible values of two-point temporal correlations in the real space of $\{\left<\sigma_1\sigma_1\right>,\left<\sigma_2\sigma_2\right>,\left<\sigma_3\sigma_3\right>\}$.}
\label{fig:inflate_tetrahedron}
\end{figure}

\section{Two-point correlations in space-time}
In this section, we focus on the cases where the initial system $\rho_A$ is maximally-mixed.
As mentioned in Section \ref{subsec: Spatial}, the set of spatial correlations described by two-qubit density matrices can be depicted in the space of $\{\left<\sigma_1\sigma_1\right>,\left<\sigma_2\sigma_2\right>,\left<\sigma_3\sigma_3\right>\}$ as the convex hull enclosed by the tetrahedron $\mathcal{T}_s$ with vertices of odd parity $(1,1,-1)$, $(1,-1,1)$, $(-1,1,1)$ and $(-1,-1,-1)$. These vertices correspond to the four Bell states. The set of temporal correlations described by $R_{AB}$ with $\rho_A=\frac{\mathrm{I}}{2}$ is simply the reflection of $\mathcal{T}_s$ in the $\left<\sigma_1\sigma_1\right>$-$\left<\sigma_3\sigma_3\right>$ plane. The resulting tetrahedron $\mathcal{T}_t$ has vertices of even parity $(1,-1,-1)$, $(1,1,1)$, $(-1,-1,1)$ and $(-1,1,-1)$.
This follows from the relation Eq. \ref{Jordan}, $R_{AB}=\frac{1}{2}E_{AB}$, when setting $\rho_A=\frac{\mathrm{I}}{2}$. A partial transpose over sub-system $A$, which geometrically corresponds to the reflection, yields
\begin{align}
R^{\mathcal{PT}}_{AB} =& \left(\mathcal{I}_A\otimes\frac{\varepsilon_{B|A}}{2}\right)\sum\limits_{ij}\ket{ii}\bra{jj}_{AB}=\rho_{AB}^{\text{Choi}}(\varepsilon_{B|A}),\label{choi}
\end{align}
where $\rho_{AB}^{\text{Choi}}(\varepsilon_{B|A})$ is the Choi matrix of $\varepsilon_{B|A}$ \cite{choi1975completely}. For arbitrary choices of $\varepsilon_{B|A}$, the Choi matrices describe the same set of correlations, $\mathcal{T}_s$ as two-qubit density matrices. As the partial transpose over sub-system $A$ generates a reflection in the $\left<\sigma_1\sigma_1\right>$-$\left<\sigma_3\sigma_3\right>$ plane, the set $\mathcal{T}_t$ is simply an inverted copy of $\mathcal{T}_s$.
\paragraph{Distance from separability}
The Peres-Horodecki criterion \cite{horodecki1996separability} implies that the octahedron region formed by the overlap between the two tetrahedra $\mathcal{T}_t$ and $\mathcal{T}_s$ corresponds to the set of separable states. With this insight, we can make a natural connection between the entanglement measure, negativity \cite{vidal2002computable}, $
f_{\mathcal{N}}(\rho_{AB})=\frac{1}{2}(\|\rho_{AB}^{\mathcal{PT}}\|_{tr}-1)$ and the causality measure $f_{tr}$. Consider a two-qubit state $\rho^{\text{Choi}}_{AB}$ as the Choi matrix of $\varepsilon_{B|A}$ in Eq. \ref{eq: PDM2}, leading to
$
f_{tr}(R_{AB})=2f_{\mathcal{N}}(\rho^{\text{Choi}}_{AB})
$.
It was shown in Ref. \cite{mundarain2007concurrence} that
the entanglement measure $f_{\mathcal{N}}$ can be visualised as the Euclidean distance $D_{s}$ between a point in $\mathcal{T}_s$ and the nearest point in the octahedron, such that $D_{s}=\frac{4f_{\mathcal{N}}}{\sqrt{3}}$. Hence, by analogy we can establish a geometric interpretation for $f_{tr}$ as the Euclidean distance $D_t$ between a point in $T_t$ and the nearest point on the face of the octahedron, such that $D_t=\frac{2f_{tr}}{\sqrt{3}}$.

\paragraph{Mixed space-time correlations} Beyond the geometry of the purely temporal and spatial correlations, a two-point PDM generally describes an arbitrary mixture of spatial and temporal correlations. Consider sequential Pauli measurements, $\sigma_A$ and $\sigma_B$ on one sub-system of a maximally-entangled pair. If the sub-system evolves through a CP-map, $\langle\sigma_A\sigma_B\rangle$ lies in the $\mathcal{T}_t$ as shown. However, if a SWAP operation is applied before the second measurement, then the reduced dynamics on sub-system $A$ will no longer be described by a CP-map. Under these conditions the correlations $\langle\sigma_A\sigma_B\rangle$ will span $\mathcal{T}_s$. Furthermore, if SWAP is applied probabilistically, the possible correlations span the entire volume of the cube formed by the vertices of $\mathcal{T}_t$ and $\mathcal{T}_s$, fully inscribing the spatial and temporal tetrahedra. It is clear that the cube is the largest possible set of space-time quantum correlations, since $-1\le\langle\sigma_A\sigma_A\rangle\le1$, and the set of possible correlation functions forms a convex set. 
We depict the geometry of different types of two-point correlations in space-time in Figure \ref{fig:StarInCube}.

\paragraph{Unital channels}

The results of Figure \ref{fig:StarInCube} assumes the initial state $\rho_A$ is maximally-mixed. 
Interestingly, $T_t$ also describes temporal correlations for an arbitrary input state $\rho_A$ but with the channel restricted to be unital which means $\varepsilon_{B|A}(\sigma_0)=\sigma_0$.
This is because only non-unital maps act non-trivially on the local components $\sigma_k\otimes\sigma_0$ of the PDM, which leads to an augmented set of correlations. Specifically, note that the first term in the trace of Eq. \ref{kk} vanishes whenever either $\rho_A$ is maximally-mixed or $\varepsilon_{B|A}$ is a unital map, in which case the parametric equations reduce to 
\begin{align}
\langle\sigma_1\sigma_1\rangle&=\cos(u),\nonumber\\
\langle\sigma_2\sigma_2\rangle&=\cos(v),\nonumber\\
\langle\sigma_3\sigma_3\rangle&=\cos(u)\cos(v).
\end{align}
The above equations give a parametric surface with the extremal points $(1,1,1)$, $(1,-1,-1)$, $(-1,1,-1)$ and $(-1,-1,1)$. The convex enclosure of these points gives exactly the temporal tetrahedron, $\mathcal{T}_t$. Hence we can see there exists a conditional reflective symmetry between the sets of temporal and spatial correlations . This symmetry is shown to be broken in the presence of certain non-unital channels, which give rise to the set of attainable temporal correlation shown in Figure \ref{fig:inflate_tetrahedron} 

\begin{figure}[H]
\centering
\includegraphics[width=0.8\linewidth]{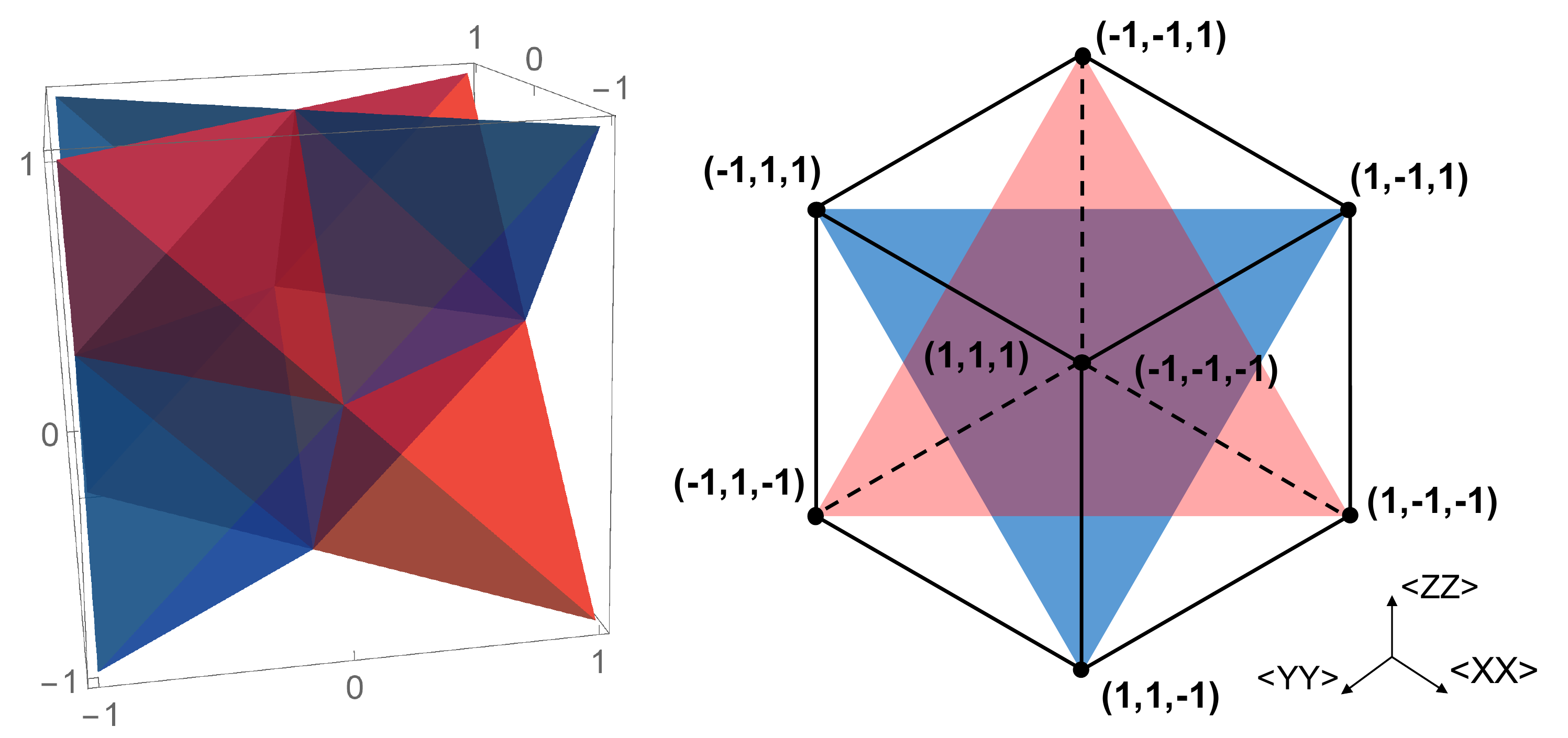}
\caption{The spatial and temporal tetrahedrons with the blue region representing $\mathcal{T}_s$, and the red region representing $\mathcal{T}_t$ (Left). A perspective plot viewing from the $(-1,-1,-1)$ direction, where the purple hexagon is a projection of the octahedron overlap, the blue and red triangles are projections of $\mathcal{T}_s$ and $\mathcal{T}_t$ respectively (Right).}
\label{fig:StarInCube}
\end{figure}

\subsection{General PDM Pauli components}
The remaining components of the two-point PDM includes all possible combinations of $\sigma_{A_1},\sigma_{A_2},\sigma_{B_1},\sigma_{B_2}\in\{\mathrm{I},\mathrm{X},\mathrm{Y},\mathrm{Z}\}$. The geometry of correlations is illustrated in Figure \ref{fig:2D}. The figure presents the types of correlations in two-point PDMs as 2-D projections onto the planes of $\{\left<\sigma_{A_1}\sigma_{B_1}\right>,\left<\sigma_{A_2}\sigma_{B_2}\right>\}$ in Figures \ref{fig:2D}a, \ref{fig:2D}b and \ref{fig:2D}c. The sets of correlations are shown for the first quadrant. The full 2-D projection is generated in a symmetric manner about the origin. In Figures \ref{fig:2D}d and \ref{fig:2D}e, we give instances in the 3-D spaces corresponding to this 2-D projections. The red region highlights extra correlations attainable in a valid PDM compared to a valid density matrix.
\begin{figure}[H]
\centering
\includegraphics[width=\linewidth]{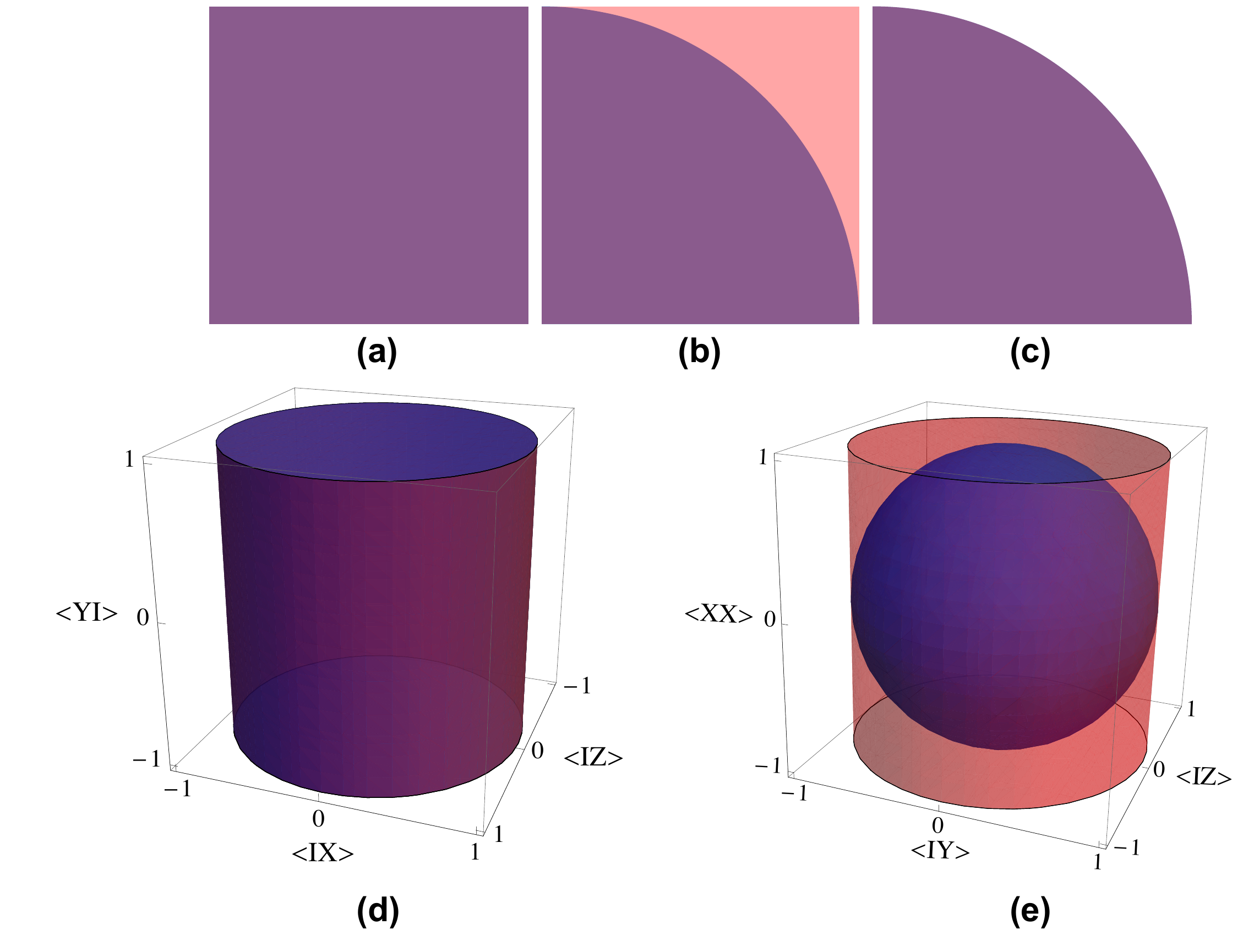}
\caption{\noindent
Here we have \textbf{(a)} Type a: $\left[\sigma_{A_1}\otimes\sigma_{B_1},\sigma_{A_2}\otimes\sigma_{B_2}\right]=0$. Temporal and spatial correlations both lie in the purple unit square; \textbf{(b)} Type b: $\{\sigma_{A_1}\otimes\sigma_{B_1},\sigma_{A_2}\otimes\sigma_{B_2}\}=0$, and one out of the four operators is $\sigma_0$. Spatial correlations lie in the purple quarter unit circle, while temporal correlations lie in the unit square. The red region is allowed by valid PDMs but not density matrices; \textbf{(c)} Type c: In all other cases, correlations are bounded by the purple quarter circle; \textbf{(d)} An example of 3-D spaces corresponding to a combination of type a and type c 2-D projections; \textbf{(e)} An example of 3-D spaces corresponding to a combination of type b and type c 2-D projections.
}\label{fig:2D}
\end{figure}
\noindent
The above completely characterises the two-point spatial and temporal correlations for qubit systems. The space of possible temporal correlations is strictly larger than the space of possible spatial correlations. These extra correlations cannot originate solely from spatially separated events, and hence are a signature of causal influence between measurement events.

\section{Discussions}

\paragraph{Quantum causal inference}
The geometric structure of spatial and temporal correlations presented in this chapter has potential application to quantum causal inference (see Ref. \cite{ried2015quantum} for an introduction). Given the outcomes of two sets of measurements, one can estimate the expectation values of two-point correlations, and identify the corresponding coordinates in the provided geometric structure, and infer whether there exists a causal relationship between the measurement events. 

\paragraph{Sequentially mimicked entanglement}
Note that the "inflated tetrahedron" in Figure \ref{fig:inflate_tetrahedron} inscribes a larger volume than $\mathcal{T}_t$. Therefore it partially overlaps with the non-separable regions in $\mathcal{T}_s$. Hence there exist temporal correlations that are statistically identical to entangled correlations in space.
Physically, this implies entanglement can be partially mimicked by sequential correlation described by a single-qubit PDM, and that it is impossible to distinguish between the two cases by only examining the correlation statistics. 
An instance of this result is reflected in the violation of the temporal CHSH inequality \cite{brukner2004quantum}, which can be expressed entirely in terms of $\left<\sigma_1\sigma_1\right>$ and $\left<\sigma_2\sigma_2\right>$ correlations. The "inflated tetrahedron" imposes constraints in the space of all three Pauli correlations, hence serves as a stronger geometric criterion for classifying quantum correlations and can act as a causal witness.
Here we should emphasize the vertices of $\mathcal{T}_s$ that correspond to maximally-entangled states do not overlap with the temporally attainable set. The inability to simulate correlations generated by Bell states with sequential measurements is related to the impossibility of constructing a quantum universal-NOT gate \cite{buvzek1999optimal}.

\chapter{Causality in quantum communication} 
\def\>{\rangle} 
\def\<{\langle}

\label{CauCap} 
We have introduced the PDM formalism in Section \ref{sub: PDM} of the previous chapter and used it as a framework to demonstrate the geometric structure of quantum correlation in both the spatial and the temporal domains. One particularly novel aspect of the PDM formalism is the ability to quantify causal relations between sequential measurement events with a causality measure which is computed by the trace norm of a given PDM.
In this chapter, we show that quantum causality plays an operational role in quantum communication. Since realistic channels for quantum communication task are noisy, it is of practical interests to quantify the capacity of the channel being used. 
Existing results have successfully connected quantum channel capacities with spatial correlations \cite{shor2002quantum, devetak2005private}. For instance, the quantum capacity of a channel is known to be equivalent to the highest rate at which it can be used to generate entanglement \cite{wilde2013quantum}. 
The notion of quantum causality characterises the temporal aspect of quantum correlations analogously with entanglement in the spatial case.
Here we take the intuitive step to uncover a connection between quantum causality and channel capacity. 
Concretely, we prove the amount of temporal correlations between two ends of the noisy quantum channel, as quantified by a logarithmic variant of the causality measure, implies a general upper bound on its channel capacity, which we will call the causal bound of quantum channel capacities.
Conveniently, the mathematical expression of the causal bound is more straightforward to evaluate than most previously known bounds for quantum capacities. We will further demonstrate the utility of the causal bound by applying it to a class of shifted depolarising channels, which shows improvement over previously known results of Ref. \cite{ouyang2014channel} and \cite{holevo2001evaluating}.
The material presented in this chapter is closely based on Ref.\cite{pisarczyk2018causal}.

\section{Bounding quantum channel capacities}
One of the central objectives of information theory is to determine the maximum rate of reliable transmission of information using a given communication channel. 
In classical information theory, the early work of Shannon proved that a simple expression governs the capacity of discrete memoryless channels  \cite{shannon1948mathematical}.

When considering the capacity of a quantum channel, $\mathcal{N}$, one has to take into account the possible necessity of encoding information in states entangled across multiple copies of the channels, to obtain the maximal capacity per use. Hence an exact computation of the capacity of a quantum channel amounts to taking the supremum over tensor products of an arbitrary number of copies of the same channel.
As such, an exact characterisation of a channels' capability to transmit quantum information has proved to be a much more challenging task.
In the absence of formulae for the exact capacities, one is often forced to rely on bounds for the quantum capacity that are tractable to evaluate \cite{takeoka2014squashed, muller2016positivity, wang2016semidefinite, sutter2015approximate, wang2017semidefinite, tomamichel2017strong, berta2017amortization}. 
Nevertheless, a significant amount of progress in the context of quantum communication has been made in determining the achievable rates for
transmitting quantum information over noisy channels\cite{holevo2001evaluating,lloyd1997capacity, shor2002quantum, devetak2005private}. However, the existing formulae for quantum capacities often involve inherent optimisation problems, leading to significant computational difficulties.
The reader is referred to Ref. \cite{wilde2013quantum} for a review of related results.

Here we take a new approach and present a general upper bound on the
quantum capacities of quantum channels that are based on causality considerations. Apart from the theoretical novelty of connecting between quantum causality and concrete communication problems, these new causal bounds further allow direct computation without requiring optimisation.

\subsection{Logarithmic Causality}
\label{sub: log Caus}
Recall in Section \ref{sub: PDM}, we reviewed a measure of causality based on the trace norm of the pseudo-density matrix. Here we introduce a useful logarithmic variant of this trace norm measure, $F(R) = \log_2 \|R\|_1$. The logarithmic causality measure is similar to causality monotones introduced in Ref. \cite{fitzsimons2015quantum} (reviewed in Section \ref{sub: PDM}), but it sacrifices convexity in favour of additivity when applied to tensor products of PDMs. Being in close analogy to the logarithmic negativity in entanglement measures, the logarithmic causality also satisfies the following important properties:
\begin{enumerate}
\item $F(R) \geq 0$, with $F(R)=0$ if $R$ is positive semi-definite, and $F(R_2) = 1$ for $R_2$ generated by two consecutive measurements of a closed system with a single qubit,
\item $F(R)$ is invariant under unitary transformations,
\item $F(R)$ is non-increasing under local operations,
\item $F(\sum_i p_i R_i) \leq \max_i F(R_i)$, for any probability distribution $\{p_i\}$.
\item $F(R\otimes S) = F(R) + F(S)$.
\end{enumerate}
Since $F(R) = \log_2 (f_\text{tr}(R)+1)$ and the logarithm function is monotonic, Properties 1-3 in the above follow straight-forwardly from the corresponding properties of the causality monotone $f_\text{tr}(R) = \|R\|_1 - 1$ which were proved in Ref. \cite{fitzsimons2015quantum}. Property 4 also follows from the monotonicity of the logarithm function which implies 
$
F(\sum_i p_i R_i) 
\leq 
\max_i F(R_i\sum_j p_j)     
$, and hence
$
F(\sum_i p_i R_i)
\leq
\max_i F(R_i).
$
As for Property 5, we note the fact that 
\begin{align}
\log_2 \|R\otimes S\|_1 = \log_2 \|R\|_1 \|S\|_1 = \log_2 \|R\|_1 + \log_2 \|S\|_1,     
\end{align}
which is equivalent to the desired property, $F(R\otimes S) = F(R) + F(S)$.

\subsection{PDM representation of quantum channels}
We consider a qubit-to-qubit channel, denoted as $\mathcal N_1$ acting on a single qubit quantum state specified by an initial density matrix $\rho$. The PDM associated to such a process, denoted by $R_{\mathcal N_1}$ involves a single use of the channel $\mathcal N_1$ and two measurements, one before and one after $\mathcal N_1$. By Eq. \ref{eq: PDM2} in Section \ref{sec: two-time}, we have
\begin{equation}
R_{\mathcal N_1}=    (\mathcal{I}\otimes \mathcal N_1) \left( \{\rho\otimes\frac{\mathrm{I}}{2},\textrm{SWAP}\}\right). \label{single}
\end{equation}
For the purpose of this chapter, we fix the input to be a maximally mixed state, so that $\rho=\frac{\mathrm{I}}{2}$. As a result, we are able to generalise Eq. \ref{single} to describe an arbitrary quantum channel $\mathcal{N}$ acting on a collection of $l$ qubits, which leads to 
\begin{equation}
R_{\mathcal N}=(\mathcal{I}\otimes \mathcal{N})\left(\frac{\textrm{SWAP}^{\otimes l}}{2^l}\right). \label{SWAP}
\end{equation}
We further focus our attention to one-way quantum communications. Thus we ought to consider the most general procedure for approximating the ideal (identity) channel with multiple copies of available resource channels. This amounts to combining $n$ parallel uses of the resource channel preceded by some encoding operations and followed by some decoding operations. The schematic diagram of the communication process being considered is shown below in Figure \ref{fig:Enc}.
\begin{figure}[H]
    \centering
    \includegraphics[width=1.0\linewidth]{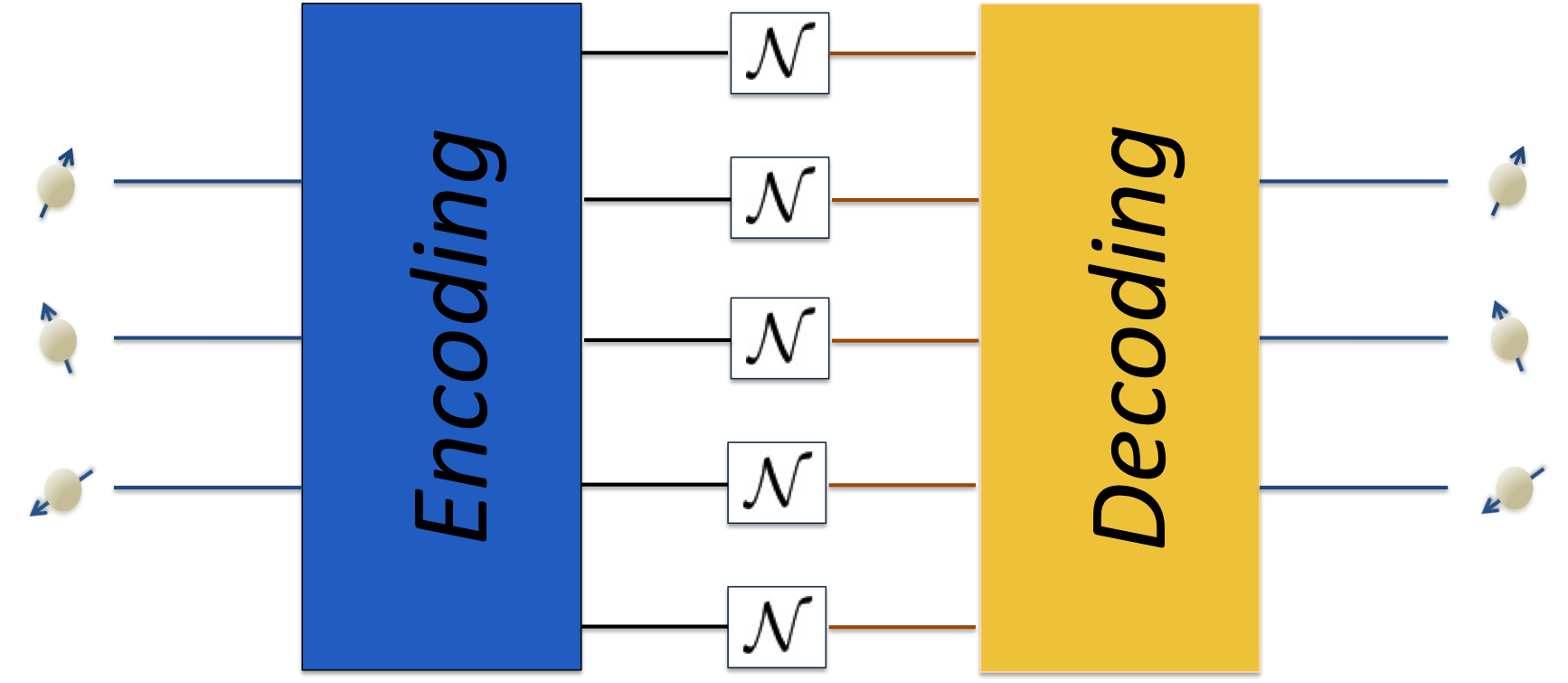}
    \caption{The state of a system of $k$ qubits is encoded into a larger Hilbert space. The encoded quantum information is then passed forward using $n$ parallel copies of the resource channel $\mathcal N$. The sent information is then decoded back into a system of $k$ qubits. In the most general setting, the input and the output of channel $\mathcal{N}$ need not have the same dimension. The encoding and decoding operations are both described by CPTP maps.}
    \label{fig:Enc}
\end{figure}

\subsection{Causal bound}

\paragraph{High-level outline} The logarithmic causality, $F(R_{\mathcal N})$ can be used to bound the number of uses of the resource $\mathcal{N}$ needed to approximate the ideal channel $\mathcal I ^{\otimes k}$.
To do so, we compare the causality across the collection of channels with the causality across the identity channel.
As a result of the Property 4. in Section \ref{sub: log Caus} , $F(\sum_i p_i R_i) \leq \max_i F(R_i)$, and the fact that for quantum channel capacity consideration it suffices to consider isometric encodings \cite{barnum2000quantum}, the causality across the combined channels does not increase under encoding and decoding. Further exploiting the additivity of causality to relate $k$ to the number of uses of the channel leads to the result that quantum capacity $Q$ of channel $\mathcal N$ is upper bounded by $F(R_{\mathcal N})$,
\begin{equation}
Q(\mathcal N) \leq F(R_{\mathcal N}). \label{eq:causal bound}
\end{equation}

\paragraph{Remarks} 
\begin{itemize}
\item Computing $F(R_{\mathcal N})$ is efficient for channels acting on relatively small Hilbert spaces, as it only requires finding the logarithm of the trace norm of a matrix. Importantly evaluating the causal bound does not involve any optimisation. 
\item Note that the relation Eq. \ref{eq:causal bound} implies that any channel with $F(R_{\mathcal{N}}) = 0$ has quantum capacity equal to zero. This reflects the fact that such a channel exhibits correlations that could have been produced by measurements on distinct subsystems of a quantum state, and so the system is necessarily constrained by the no-signalling theorem. 

\item When $F(R_{\mathcal{N}})$ is strictly positive, the correlations between the two ends of the channel cannot be captured by bipartite density matrices, thus signifying information being passed forward in time. 
\item Although the causal bound was presented for channels acting on the collection of qubits, this result applies to channels with arbitrary input and output dimensions. In such cases, it suffices to restrict the channel to act only on a subspace of the $2^k$ dimensional Hilbert space.
\end{itemize}

\subsection{Proof}

In order to prove Eq. \ref{eq:causal bound},
we start with constructing the PDM that corresponds to a channel obtained by using $n$ copies of the resource channel $\mathcal{N}$ preceded by the encoding operation $E$ and followed the decoding $D$. Let $\mathcal{M} = D \circ \mathcal{N}^{\otimes n} \circ E$. The PDM to consider, $R_\mathcal{M}$ is related to that of the ideal channel via
\begin{align}
R_\mathcal{M} = ({\mathcal I}^{\otimes k} \otimes \mathcal{M})(R_{\mathcal{I} ^{\otimes k}}).    \label{eq: R_M}
\end{align}
We add and subtract $R_{\mathcal{I} ^{\otimes k}}$ on the left hand side of  Eq. \ref{eq: R_M}, and apply the reverse triangle inequality to obtain
\begin{align}
\| R_{\mathcal{M}}\|_1 
\ge
\|R_{{\mathcal I ^{\otimes k}}} \|_1 -  \|R_{\mathcal{M}} -  R_{\mathcal I ^{\otimes k}}  \|_1.    
\end{align}
The trace distance between two pseudo-density matrices can be related to distance in the diamond norm \cite{wilde2013quantum} as follows,
 \begin{align}
 \|R_{\mathcal{M}} - R_{{\mathcal I} ^{\otimes k}}\|_1 
 &= 
 \|({\mathcal I} ^{\otimes k}\otimes 
 (\mathcal{M} -  {\mathcal I}^{\otimes k}))(R_{\mathcal I ^{\otimes k}}) \|  _1 \nonumber \\
 &\leq 
 \|
 \mathcal{M} -  {\mathcal I}^{\otimes k} \|_\diamond
 \|R_{\mathcal I ^{\otimes k}}\|_1 , \label{eq: 1 norm upper bound}
 \end{align}
 where $\|\cdot\|_\diamond$ denotes the diamond norm. We define $\epsilon=  \|\mathcal{M} -  {\mathcal I}^{\otimes k} \|_{\diamond}$ and use the upper bound of Eq. \ref{eq: 1 norm upper bound} as well as the positivity of $\|R_{{\mathcal I ^{\otimes k}}}\|_1$ to obtain
\begin{align}
\frac{ \| R_{\mathcal{M}}\|_1 }{\|R_{{\mathcal I ^{\otimes k}}} \|_1  } 
\ge
1 - \epsilon.     
\end{align}
Taking the logarithm on both sides of the above inequality leads to 
\begin{align}
F(R_{\mathcal M}) - F(R_{\mathcal I ^{\otimes k}}) \ge \log_2(1-\epsilon).\end{align}
The connection between the PDM and SWAP matrix and the non-increasing property of the trace norm under the partial trace together leads to the fact that the causality does not increase under decoding and encoding operations. This gives 
\begin{align}
F(R_{\mathcal{M}}) \leq F(R_{\mathcal{N}}^{\otimes n}).    
\label{eq: nonincrea}
\end{align}
A detailed proof of the above inequality Eq. \ref{eq: nonincrea} is presented in \ref{sub: Non-increasing proof}. 
Furthermore, this same property of $F$ guarantees that even if we had allowed the encoding and decoding operations to operate on entangled ancillary registers, Eq. \ref{eq: nonincrea} is still valid \cite{bennett1999entanglement,bennett2002entanglement,shor2004classical}.
Hence the resultant bounds based on Eq. \ref{eq: nonincrea} are also bounds on the entanglement-assisted capacities. It is important to note this non-increasing property does not hold for any other Schatten norm but the trace norm.

The additivity property of $F$ under tensor products (Property 5. in Section \ref{sub: log Caus}) implies $F(R_{\mathcal{N}}^{\otimes n}) = n F(R_{\mathcal N})$ and $F(R_{\mathcal I ^{\otimes k}})=k F(R_{\mathcal I})$, which leads to
\begin{align}
n F(R_{\mathcal N}) - k  F(R_{\mathcal I}) \ge \log_2(1-\epsilon) .    
\end{align}
Finally, using $F(R_{\mathcal I})=1$ for a quantum capacity with respect to a single qubit system, we obtain
\begin{align}
\frac{k}{n} \le F(R_{\mathcal N})  - \frac{\log_2(1-\epsilon)}{n}.
\label{eq: k/n}
\end{align}
The relation between the $\epsilon$ distance in diamond norm and the distance in the completely bounded infinity norm in turn guarantees $\epsilon$ goes to zero as $n$ approaches infinity. More details on this fact is presented in Section \ref{sub: large n proof}. This concludes the proof of the bound, $Q(\mathcal N) \le F(R_{\mathcal N}).$

\section{Mathematical details}
\subsection{Non-increasing property}
\label{sub: Non-increasing proof}

An important property used in proving the causal bound was that the decoding and encoding operations do not increase causality, such that $F(R_{\mathcal{M}}) \leq F(R_{\mathcal{N}}^{\otimes n})$.
To prove this property, we need to make use of the following lemma,
\begin{lemma}
    Let $K$ be a linear map from $k$ qubits to $m$ qubits. Then
    \begin{equation}
    ( I  \otimes K ) \textrm{SWAP} ^{\otimes k}    (I \otimes K^{\dagger} ) = (K^{\dagger} \otimes I) \textrm{SWAP} ^{\otimes m}(K \otimes I), \label{swap}
    \end{equation}
    where ($A \otimes B$) means that $A$ and $B$ are applied to the first and second subsystems of each of the $\textrm{SWAP}$s respectively.
\end{lemma}

\begin{proof}
    Let $K= \sum_{i=0}^{2^k-1} \sum_{j=0}^{2^m-1} e_{ij} \ket{j}\bra{i}.$ The tensor product of $k$-qubit $\textrm{SWAP}$s can be written as
    \begin{align}
            \textrm{SWAP}^{\otimes k} = \sum_{u,v=0}^{2^{k}-1} (\ket{u}\otimes \ket{v})(\bra{v}\otimes \bra{u}). \label{eq: swap decomp}
    \end{align}
Substituting Eq. \ref{eq: swap decomp} into the left hand side of Eq. \ref{swap} leads to
    \begin{align}
    &( I ^{\otimes k} \otimes K ) \textrm{SWAP} ^{\otimes k}  ( I ^{\otimes k} \otimes K^{\dagger} ) \nonumber\\
    =&\sum_{i,j,i',j',u,v} ( I \otimes \ket{j} \bra{i})\ket{u} \ket{v} \bra{v}\bra{u}( I \otimes \ket{i'} \bra{j'}) e_{ij}  e_{i'j'}^{*}\nonumber\\
    =&\sum_{j,j'=0}^{2^m-1} \sum_{u,v=0}^{2^k-1} \ket{u} \ket{j} \bra{v} \bra{j'} e_{vj} e_{uj'}^* .
    \end{align}
    Similarly evaluating the right hand side of Eq. \ref{swap} we get
    \begin{align}
    &(K^{\dagger} \otimes  I^{\otimes m}) \textrm{SWAP} ^{\otimes m} (K \otimes  I^{\otimes m})\nonumber\\
    =&\sum_{i,j,i',j',u,v} ( \ket{i} \bra{j}\otimes  I)\ket{u} \ket{v} \bra{v}\bra{u}(\ket{j'} \bra{i'}\otimes  I) e_{ij} ^{*} e_{i'j'}\nonumber\\
    =&\sum_{i,i'=0}^{2^k-1} \sum_{u,v=0}^{2^n-1} \ket{i} \ket{v} \bra{i'} \bra{u} e_{i'v} e_{iu}^{*} \nonumber\\
    =&\sum_{j,j'=0}^{2^n-1} \sum_{u,v=0}^{2^k-1} \ket{u} \ket{j} \bra{v} \bra{j'} e_{vj} e_{uj'}^{*},
    \end{align}
    where in the last step we have relabelled the indices. 
\end{proof}
We are now in the position to prove the desired non-increasing property of logarithmic causality, which is summarised in the lemma below.
\begin{lemma}
    Let $\mathcal{E}$ and $\mathcal{D}$ be encoding and decoding operations and $\mathcal{M} = \mathcal{D} \circ \mathcal{N}^{\otimes n} \circ \mathcal{E}$. Then
    \begin{align}
        \log_2  \|R_{\mathcal M} \|_1 \leq \log_2 \| R_{\mathcal N} ^{\otimes n} \|_1 . \label{eq: lemma2}
    \end{align}
\end{lemma}
\begin{proof}
The decoding procedure is a local operation and therefore from Property 4. of $F(R)$ in Section \ref{sub: log Caus} , we have
    \begin{align}
        \| R_{\mathcal M}\|_1 \leq \| (\mathcal I \otimes (\mathcal N ^{\otimes n} \circ \mathcal{E})) (R_{\mathcal I ^ {\otimes k}})\|_1.
    \end{align}
    Let $\mathcal{E}$ encode $k$ qubits into $m$ qubits. Using Lemma 1, we have
    \begin{align}
    \| (\mathcal I \otimes (\mathcal N ^{\otimes n} \circ \mathcal{E})) (R_{\mathcal I ^ {\otimes k}})\|_1 &=
    \| (\mathcal{E}^{\dagger} \otimes \mathcal N ^{\otimes n} ) (R_{\mathcal I ^ {\otimes m}})\|_1 \nonumber\\
    &= \| (\mathcal{E}^{\dagger} \otimes \mathcal I ) (R_{\mathcal N} ^ {\otimes n})\|_1. 
    \end{align}
    The PDM $R_{\mathcal N} ^ {\otimes n}$ can be decomposed into its positive and negative part, and rewritten as
    \begin{align}
        R_{\mathcal N} ^ {\otimes n} = R_+ - R_-,
    \end{align}
    where both $R_+$ and $R_-$ are positive semi-definite. Applying the triangle inequality gives
    \begin{align}
    \| (\mathcal{E}^{\dagger} \otimes \mathcal I ) (R_{\mathcal N} ^ {\otimes n})\|_1 &\leq  \| (\mathcal{E}^{\dagger} \otimes \mathcal I  )( R_+) \|_1 +\| (\mathcal{E}^{\dagger} \otimes \mathcal I  ) (R_-) \|_1 \nonumber \\
    &=\tr ((\mathcal{E}^{\dagger} \otimes \mathcal I )( R_+)) + \tr ((\mathcal{E}^{\dagger} \otimes \mathcal I ) (R_-)) \nonumber \\
    &= \tr((\mathcal{E}^{\dagger} \otimes \mathcal I ) (R_+ + R_-)). 
    \end{align}
    It is well-known that in bounding quantum channel capacity, one can restrict $\mathcal{E}$ to be an isometry with only one non-zero Kraus operator, which we denote by $K$ \cite{barnum2000quantum}. This allows us to write 
    \begin{align}
    \tr((\mathcal{E}^{\dagger} \otimes \mathcal I ) (R_+ + R_-))&=
    \tr((K^{\dagger} \otimes  I  ) (R_+ + R_-)  (K \otimes  I  )) \nonumber\\
    &=\tr ((KK^{\dagger} \otimes  I ) (R_+ + R_-)),
    \end{align}
    where the second equality follows from the cyclic property of the trace. Since $P = KK^{\dagger} \otimes \mathcal I ^{\otimes n}$ is a projector, so that $PP=P$, we have 
\begin{align}
\tr(P (R_+ + R_-))= \tr (P(R_+ + R_-)P) = \|P(R_+ + R_-)P\|_1.
\end{align}
Next we applying the H\"older's inequality twice and make use of the fact the infinity norm of a projector equals one, and obtain
    \begin{align}
    \|P(R_+ + R_-)P\|_1 &\leq \|P\|_{\infty} \|R_+ + R_-\|_1\|P\|_{\infty} \nonumber \\
    &=\|R_+ + R_-\|_1,
    \end{align}
where $\|\cdot\|_{\infty}$ denotes the infinity norm and is defined by the largest singular value of a matrix. Furthermore, $R_+$ and $R_-$ are by definition orthogonal. Hence 
    \begin{align}
 \| R_+ + R_-\|_1 = \| R_+ - R_-\|_1 = \| R_{\mathcal N}^{\otimes n}\|_1, 
    \end{align}
    which leads to
$
    \| R_{\mathcal M}\|_1 \leq \| R_{\mathcal N}^{\otimes n}\|_1.
$
Finally, use the fact that logarithm is a monotonic function, the desired property Eq. \ref{eq: lemma2} follows.
\end{proof}

\subsection{Large-$n$ limit}
\label{sub: large n proof}
Here we prove that the error parameter $\epsilon$ in Eq. \ref{eq: k/n} goes to zero in the limit of large $n$. By the definition of the distance in diamond norm, we have
\begin{align}
\epsilon&=
\| {\mathcal I} ^{\otimes k} \otimes  (\mathcal{M} -  {\mathcal I}^{\otimes k}) \|_1 \nonumber
\\
&=\sup_{\| X\|_1 = 1}
\|( {\mathcal I} ^{\otimes k} \otimes  (\mathcal{M} -  {\mathcal I}^{\otimes k})) (X) \|_1 .
\end{align}
Consider the spectral decomposition of Hermitian $X=\sum_{i} \lambda_i |\psi_i\>\<\psi_i|$, where $\{|\psi_i\>\}$ denotes an orthonormal basis, and $\{\lambda_i\}$ are the corresponding eigenvalues.
Define $\mathcal A = ( {\mathcal I} ^{\otimes k} \otimes  (\mathcal{M} -  {\mathcal I}^{\otimes k}))$,
and we have
\begin{align}
\epsilon
&= 
\sup_{\{|\psi_i\>\}_i, \sum_{i}|\lambda_i|=1}
\left\|\mathcal A \left( \sum_{i} \lambda_i |\psi_i\>\<\psi_i| \right) \right\|_1 \notag\\
&\le
\sup_{\{|\psi_i\>\}_i, \sum_{i}|\lambda_i|=1} \left( \sum_{i} |\lambda_i|
\|\mathcal A ( |\psi_i\>\<\psi_i| ) \|_1 \right) \notag\\
&\le 
\sup_{ |\psi\> }\|\mathcal A ( |\psi\>\<\psi| ) \|_1 .
\end{align}
Note that $\mathcal A$ represents the difference of two linear maps 
${\mathcal I} ^{\otimes k} \otimes {\mathcal I} ^{\otimes k}$ and 
${\mathcal I} ^{\otimes k} \otimes \mathcal M$, by linearity
we have
\begin{align}
&\sup_{ |\psi\> }\|\mathcal A ( |\psi\>\<\psi| ) \|_1
=&
\sup_{ |\psi\> }\|
( {\mathcal I} ^{\otimes k} \otimes {\mathcal I} ^{\otimes k})( |\psi\>\<\psi| )
-
({\mathcal I} ^{\otimes k} \otimes \mathcal M) ( |\psi\>\<\psi| ) \|_1. 
\end{align}
In the above we have inside a supremum the trace distance between two quantum states. Now we need to relate the distance between quantum states measured by the 1-norm to that measured in terms of the fidelity. Let
\begin{align}
f(\rho,\sigma) = \tr \sqrt{ \sqrt \rho \sigma \sqrt \rho} 
\end{align}
denote the fidelity between two positive semi-definite matrices.
If $\rho=|\psi\>\<\psi|$, then $f(\rho,\sigma)= \sqrt{\<\psi|\sigma|\psi\>}.$
The Fuchs-van de Graaf inequalities \cite{fuchs1999cryptographic} imply
\begin{align}
1-f(\rho,\sigma)
\le 
\frac{1}{2}\| \rho - \sigma \|_1
\le
\sqrt{1-f(\rho,\sigma)^2}.
\end{align}
Hence we have
\begin{align} 
&\frac 1 2 \sup_{ |\psi\> }\|\mathcal A ( |\psi\>\<\psi| ) \|_1 \le 
&\sqrt{1- \inf_{ |\psi\> }
    f( ( {\mathcal I} ^{\otimes k} \otimes {\mathcal I} ^{\otimes k})( |\psi\>\<\psi| )  ,
    ({\mathcal I} ^{\otimes k} \otimes \mathcal M) ( |\psi\>\<\psi| )   )^2}.
\end{align} 
The above inequality is related to entanglement fidelity $F_e(\rho,\Phi)$ of a state $\rho$ with respect to the channel $\Phi$ which has a set of Kraus operators, $\mathcal{K}$, and acts on the state as 
$\Phi(\rho) = \sum_{A \in K} A \rho A ^\dagger$.
From Schumacher's formula \cite{Sch96}, we have
\begin{align}
F_e(\rho,\Phi) &= \<\phi | (\Phi \otimes \mathcal I) (|\phi\>\<\phi| )|\phi \>\nonumber\\
&= f( (\Phi \otimes \mathcal I) (|\phi\>\<\phi| ), |\phi\>\<\phi|) \nonumber\\
&=    \sum_{A \in K} |\tr \rho A|^2,
\end{align}
where $|\phi\>$ is introduced as a purification of $\rho$.
We denote $F_e(\Phi)=
\inf_\rho F_e(\rho,\Phi)$, and have 
\begin{align}
F_e(\Phi) 
=\inf_{|\phi\>} \<\phi |( \Phi \otimes\mathcal I) (|\phi\>\<\phi| )|\phi \> 
= 
\inf_{|\phi\>} f( |\phi\>\<\phi|,   (\Phi \otimes \mathcal I) (|\phi\>\<\phi| ) )^2.
\end{align}
Hence using the notation for the entanglement fidelity, we can write 
\begin{align}
\frac 1 2 \sup_{ |\psi\> }\|\mathcal A ( |\psi\>\<\psi| ) \|_1
\le
\sqrt{1- F_e(  \mathcal M) }.
\end{align}
Thus
$
\epsilon \le 2 \sqrt{1- F_e(  \mathcal M) }.
$
Now use the following relation proved by Kretschmann and Werner in Proposition 4.3 of Ref. \cite{KrW04}
\begin{align}
1- F_e(\Phi) 
\le
4 \sqrt {\|\Phi- \mathcal I \|_{\rm cb}} 
\le 
8 \left( 1- F_e(\Phi)  \right)^{1/4},
\end{align}
where $\| \cdot \|_{\rm cb}$ denotes the completely bounded norm induced on the operator infinity norm \cite{paulsen1986completely}.
We obtain  
\begin{align}
\epsilon 
\le & 
2\sqrt{4 \sqrt {\|\mathcal M- \mathcal I\|_{\rm cb}}}\nonumber\\
=&
4 \|\mathcal M- \mathcal I\|_{\rm cb}^{1/4}.
\end{align}
Since $ \|\mathcal M- \mathcal I\|_{\rm cb}$ is guaranteed to approach zero as $n$ approaches infinity in the channel capacity theorems, $\epsilon$ here also approaches zero. This concludes the proof.

\section{Application of causal bound}
\subsection{Comparison with Holevo and Werner bound}
Here we compare the causal bound with a simple well-known bound on quantum capacities of Holevo and Werner (HW) which is general, and has a similar form, but requires optimisation \cite{holevo2001evaluating}. 
Given a quantum channel $\mathcal N$, and a transpose map $\mathcal T$, 
the Holevo-Werner upper bound on the quantum capacity is 
\begin{align}
Q_{\mathcal T}(\mathcal N) 
=\log_2  \| \mathcal N  \mathcal T \|_{\diamond} 
=\log_2  \| \mathcal I \otimes \mathcal N \mathcal T \|_1.
\end{align}
By the definition of the induced norm this can be rewritten as
\begin{align}
Q_{\mathcal T}(\mathcal N) = \sup_{\rho} \left( \log_2  \| (\mathcal I \otimes \mathcal N \mathcal T) (\rho)\|_1 \right).    \label{eq: HW rewrite}
\end{align}
Now we compare the above to the causal bound. In the case of the maximally mixed input, the pseudo-density matrix becomes
\begin{align}
R_{\mathcal N}= (\mathcal{I} \otimes \mathcal N ) \left(\frac{\textrm{SWAP}^{\otimes k}}{2^k}\right) = (\mathcal{I} \otimes \mathcal N \mathcal T) (\ket{\Phi^{+}}\bra{\Phi^{+}})^{\otimes k}.
\end{align}
The causal bound then reads
\begin{align}
F (R_{\mathcal N})  =\log_2  \| (\mathcal{I} \otimes \mathcal N  \mathcal T)  (\ket{\Phi^{+}}\bra{\Phi^{+}})^{\otimes k} \|_1 .    
\end{align}
Comparing this to the HW bound in Eq. \ref{eq: HW rewrite}, it is clear that
$
F (R_{\mathcal N}) \leq Q_{\mathcal T}(\mathcal N),
$
and the two are equal when the supremum is achieved at the maximally entangled state $(\ket{\Phi^+} \bra{\Phi^+})^{\otimes k}$. Hence we have shown the causal bound is better or equal to the HW bound.

\subsection{Shifted depolarising channel}

As an illustration of applying the causal bound, we consider the class of shifted depolarising channels. A shifted depolarising channel generalises the well-studied quantum depolarising channel \cite{king2003capacity,smith2008additive}. It outputs either the input state or the state $\frac{{I}+\gamma Z}{2}$ shifted from the maximally mixed state with probability $4p$. For a single qubit the shifted depolarising channel can be defined by 
\begin{align}
\mathcal{N}_{\gamma} (\rho) = (1-4p)\rho + 4p\left(\frac{{I}+\gamma Z}{2} \right),\label{eq: shifted depo}    
\end{align}
where the parameter $\gamma \in [0,1]$ parametrises the shift, with a zero $\gamma$ corresponding to the standard depolarising channel. The PDM representation of the single qubit shifted depolarising channel, $R_{\mathcal{N}_{\gamma}}$ is derived using Eq. \ref{SWAP}, from which we obtain an analytic expression for the value of $F(R_{\mathcal{N}_{\gamma}})$, and hence an upper bound on its quantum capacity of the channel,
\begin{align}
\mathcal{Q}(\mathcal{N}_{\gamma}) &\leq F(R_{\mathcal{N}_{\gamma}}) \nonumber \\
&= \log_2 \bigg( 1-p + \frac{1}{2} \sqrt{1-8p+16 p^2+4\gamma^2 p^2} \nonumber\\
&\hspace{1.2cm}+\frac{1}{2}\left| 2p-\sqrt{1-8p+16p^2+4\gamma^2 p^2}\right|\bigg).
\end{align}
We show in Figure \ref{fig:HW} the difference between the HW bound and the causal bound  on quantum channel capacity of a shifted depolarising channel. Note that the two bounds are identical for standard depolarising channel where there is no shift. However, the causal bound  is tighter when the shift $\gamma$ increases.
\begin{figure}[H]
    \centering
    \includegraphics[width=1.0\linewidth]{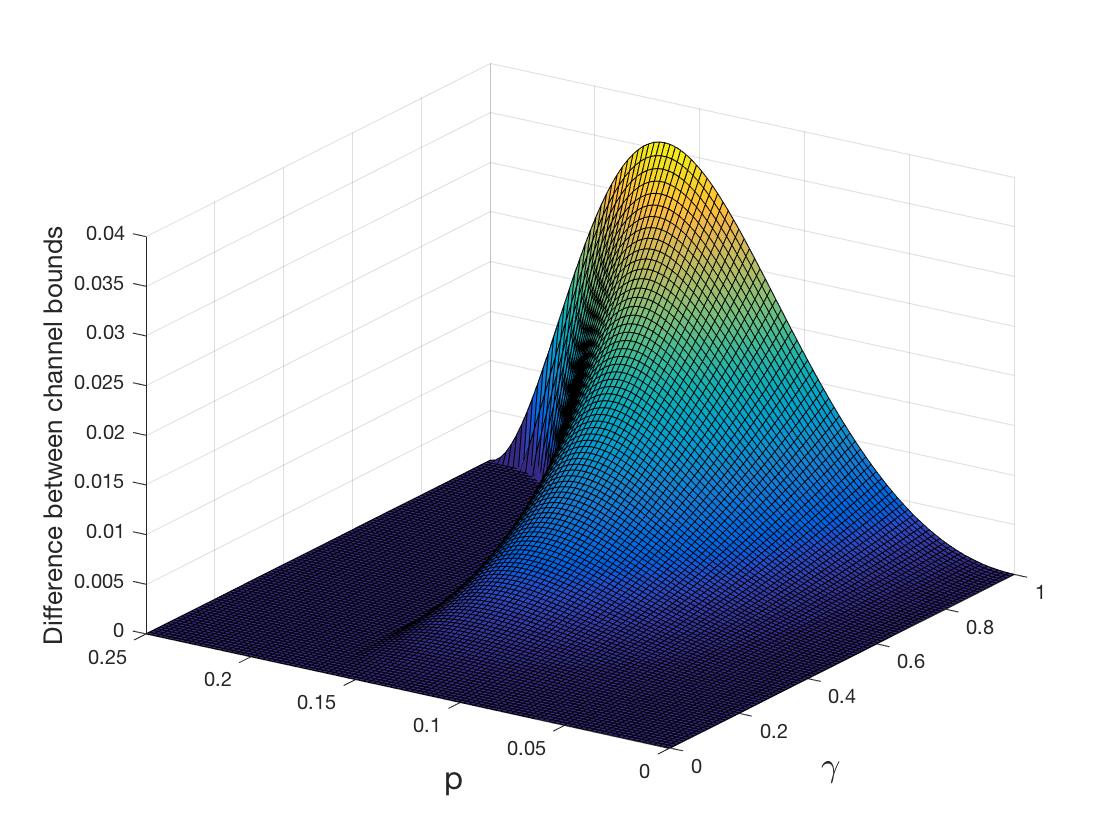}
    \caption{Difference between the HW and causal bound s on quantum channel capacity of a shifted depolarising channel.}
    \label{fig:HW}
\end{figure}
\noindent Hence the shifted depolarising channel constitutes a class of examples for which the causal bound  is strictly tighter than the HW bound. Furthermore, we found that the causal bound  $F(R_{\mathcal{N}_{\gamma}})$ also shows improvement upon the best known bound from Ref. \cite{ouyang2014channel}.
In Figure \ref{fig:Kai}, we show the difference between the previously known bound from Ref. \cite{ouyang2014channel} and the causal bound on the quantum channel capacity of a shifted depolarizing channel. 
\begin{figure}[H]
    \centering
    \includegraphics[width=1.0\linewidth]{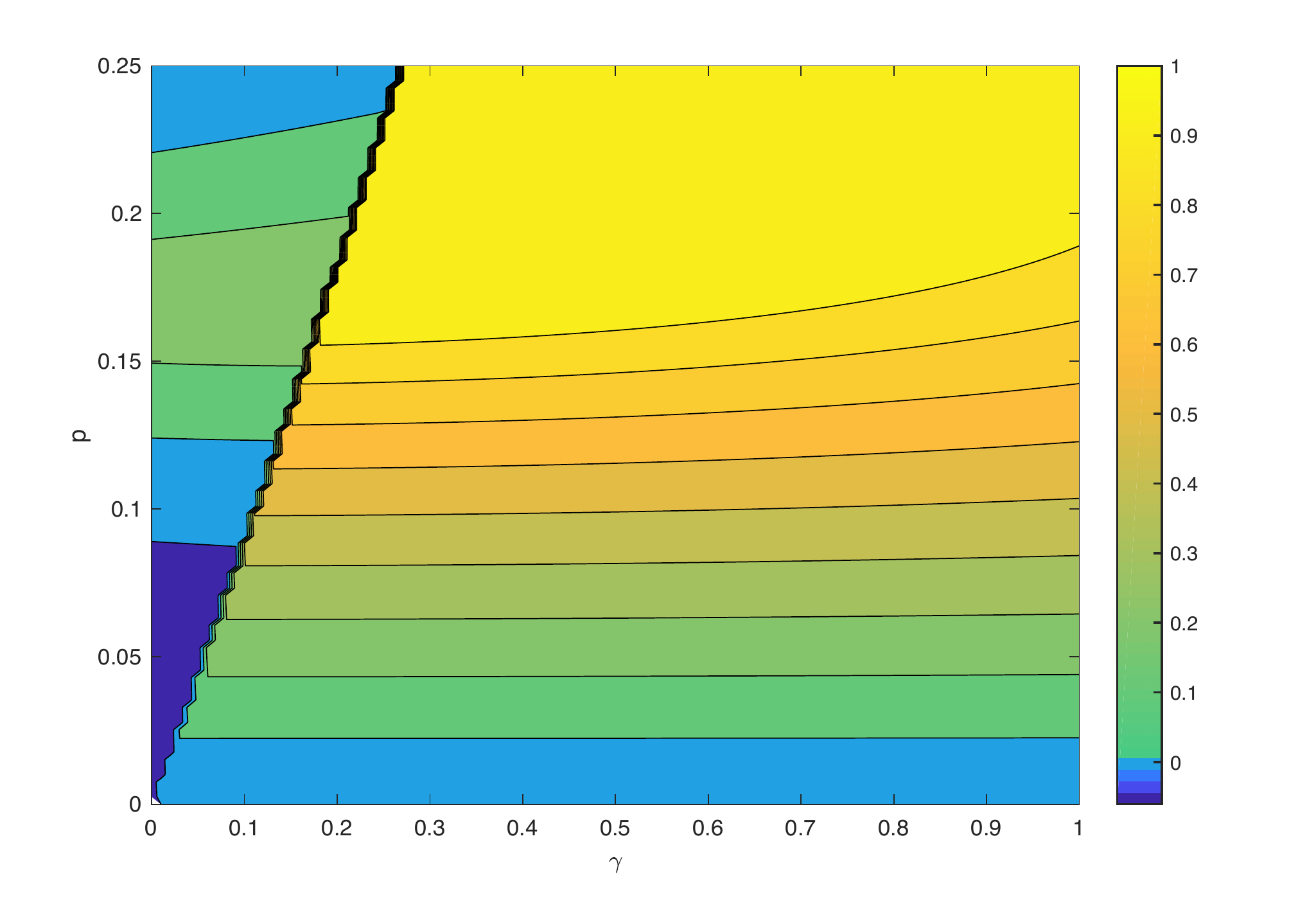}
    \caption{Difference between the previously known bound from Ref. \cite{ouyang2014channel} and the causal bound on quantum channel capacity of a shifted depolarising channel.}
    \label{fig:Kai}
\end{figure}
\noindent The causal bound is tighter for almost all values of $\gamma$ and $p$. Only in the region of small shift $\gamma$ and small probability $p$, which corresponds to the bottom left corner of the diagram, the causal bound is less tight. Note that the shifted depolarising channel reduces to the standard depolarising channel when $\gamma=0$, and the identity channel when $p=0$. The causal bound is not the tightest known bound for the standard depolarising channel, while it evaluates exactly to the quantum channel capacity for the identity channel.

\section{Summary and discussions}
In this chapter, we have presented a general upper bound on the quantum capacity of noisy quantum channels based on fundamental causality considerations. Contrary to most other existing bounds, the computation of the causal bound does not involve an explicit optimisation problem. The logarithmic causality measure used here is in close analogy with the entanglement logarithmic negativity and possesses desired properties which make it useful for studying channel capacities.

Our approach based on quantum causality is generally applicable to arbitrary quantum channels and can produce non-trivial upper bounds for any given channel. 
Therefore, this result could further help the understanding of the communication rate of complex systems for which optimisation methods are computationally too costly, including quantum networks and quantum communication between many parties \cite{leung2010quantum, hayashi2007quantum}.

Research on the spatial quantum correlations has lead to the formulation of various entanglement monotones with various corresponding operational meanings and applications, e.g., distillable entanglement, entanglement cost, squashed entanglement  \cite{horodecki2009quantum, christandl2004squashed}. As a temporal counterpart of quantum correlations, the result presented in this chapter initiates research on the operational significance of causality measures that might prove useful in a broader range of applications.

\chapter{Conclusion} 
\label{Conclusion} 

In this thesis, I started by describing the useful quantum algorithms for linear algebra, then moved on to illustrate how the quantum algorithm machinery can be applied to enhance classical supervised learning. In the last part of the thesis, we studied the notion of causality in an ensemble of quantum states. I presented results on inferring causal corrections, and the connection between quantum causality and the limit of transmitting quantum information over a noisy channel. Here we provide a summary of new research progress discussed in this thesis and give a brief outlook for avenues of future research.

\section{Summary}

In Chapter \ref{QDLSA}, I have shown a new quantum algorithm for solving the quantum linear system problem. This approach is based on a quantum singular value decomposition technique which in turn makes use of a data structure that provides oracle access to the row vectors of a matrix and the vector of row norms. Since our approach does not involve explicitly simulating the system's defining matrix as a Hamiltonian, the resultant runtime does not depend on sparsity, which gives the new linear systems algorithm an advantage over the existing approach for dense matrices.
As a result of the error dependence in singular value decomposition, the fixed-error runtime of our linear system algorithm has a linear dependence on the Frobenius norm of the matrix. Nevertheless, we have proved our algorithm has a $\tilde{\mathcal O}(\sqrt{n})$ runtime in the general case, providing a polynomial speedup over the previous state-of-the-art. In the special case of the matrix having a low-rank structure, our algorithm exhibits an even more advantageous $\tilde{\mathcal O}(\log n)$ runtime scaling.

In Chapters \ref{QGP} and \ref{QGPT}, we applied quantum algorithms to supervised machine learning using Gaussian processes. For computing the mean and variance predictor of a given GP model, we have shown the quantum linear systems approach can be applied to achieve exponential or polynomial speedups over classical implementations depending on whether the covariance matrix is sparse or not. For the purpose of training GPs, we have presented a quantum approach for evaluating the logarithm of the marginal likelihood of the model on a given dataset. The quantum GP training approach has two main components, the augmented quantum linear system algorithm for quantifying the model's performance on the training data, and the quantum log determinant algorithm for quantifying the complexity of the model. We have shown the quantum GP training approach allows for efficiently evaluating the variation of marginal likelihood on each training step, which is the main computation bottleneck for model selections for GPs. The quantum GP prediction and training procedures together provide a concrete use-case in supervised learning for which quantum computation has a provable advantage over the best-known classical implementation.

In Chapter \ref{Chapter: QBDL}, we built upon the previously discussed quantum GP algorithm and leveraged a connection between deep neural network models and Gaussian processes to develop a quantum algorithm for deep learning. The presented quantum approach to deep learning is Bayesian as the training of the parameters in the neural network amounts to evaluating a Gaussian posterior distribution instead of the more conventional methods, such as backpropagation with stochastic gradient descent. To simulate the Hamiltonian that represents the multi-layer kernel matrix, we designed a quantum method based on density matrix exponentiation and proved the computational overhead in terms of the required number of resource density matrix which encodes the base case kernel matrix. Furthermore, we have demonstrated the matrix inversion component of quantum GP regression by performing experiments on quantum simulators as well as the state-of-the-art quantum processing units, which have shown encouraging results, despite the implementation being a small-scale variant of the full algorithm. 

In Chapter \ref{Chapter: geometry} and \ref{CauCap}, we looked into the concept of causality in the quantum domain. Specifically, we have made use of the pseudo-density matrix formalism to derive the geometric structure of spatial and temporal two-point quantum correlations, which serves as an analytical toolkit for inferring causal relations in quantum datasets. Furthermore, the geometric structure can be seen as a strong witness of quantum entanglement, distinguishing it from possible sequentially generated statistics. We then further apply quantum causality in the pseudo-density matrix formalism to quantum communication and derived a general upper bound for the quantum channel capacity of a given a noisy channel. 

\section{Outlook}

The results presented in this thesis provide numerous potential avenues for further research. As discussed earlier in Chapter \ref{QDLSA}, it would be useful to conduct a detailed resource analysis for the QDLS algorithm. Since it circumvents the costly Hamiltonian simulation subroutine, as required by the previous quantum linear system algorithms, implementing the QDLS algorithm may require significantly less elementary gate operations compared to the analysis presented in Ref. \cite{Scherer2017}. Given the close analogy between the quantum walk based approach of QDLS algorithm and the quantum search algorithm \cite{grover1996fast}, it is also interesting to ask whether the runtime $\tilde{\mathcal{O}}(\sqrt{n}\log n)$ is optimal given the required memory model following a similar logic of the optimality of quantum search \cite{zalka1999grover}. 

The main direction of interest for quantum enhanced GPs and Bayesian deep learning presented in Chapter \ref{QGP}, \ref{QGPT} and \ref{Chapter: QBDL} is experimental. Although the development of practical hardware for quantum computing is still at its infancy, early quantum computers have already become available and will continue to grow in scale and noise tolerance. It is an exciting time to ask whether near-term quantum computing can truly enhance machine learning, either on a qualitative or a quantitative level. We hope that before too long the quantum GP algorithms, its corresponding training algorithms, and the quantum GP induced Bayesian deep learning approach can be fully implemented with real quantum devices on large-scale datasets, and ultimately produce analytical power beyond what is classically achievable.

The geometric structure presented in Chapter \ref{Chapter: geometry} identifies a class of quantum operations which can generate sequential statistics that mimics entanglement. It would be interesting to observe these correlations experimentally. Furthermore, as quantum entanglement is famously given a significant role in quantum cryptography \cite{ekert1991quantum}, it is interesting to ask whether its temporal counter-part, causality would provide similar applicational prospects. The results presented in Chapter \ref{CauCap} are a concrete example of the operational meaning of causality, where it is shown to be significant to the field of quantum communication. Thus the presented work initiates a thread of research on the practical applications of quantum causality.  

\bibliographystyle{unsrt}
\bibliography{thesis}

\end{document}